\documentclass[11pt, a4paper]{article}
\usepackage[margin=1in]{geometry}
\usepackage{amsmath, amsthm, amssymb, enumitem, mathtools, bbm, xcolor, caption, tikz}
\usepackage[utf8]{inputenc}
\usepackage[T1]{fontenc}

\usepackage[colorlinks=false]{hyperref}
\usepackage{thm-restate}
\usepackage[noadjust]{cite}

\usepackage[absolute]{textpos}

\tikzset{MyNode/.style={circle, draw, inner sep=2,outer sep=0, fill=gray}}
\tikzset{MyBoldNode/.style={circle, draw, inner sep=3,outer sep=0, fill=red}}

\usepackage{xspace}
\usepackage{todonotes}

\newtheorem{theorem}{Theorem}[section]
\newtheorem{claim}{Claim}[theorem]
\newtheorem{corollary}[theorem]{Corollary}

\newtheorem{lemma}[theorem]{Lemma}
\newtheorem{observation}[theorem]{Observation}
\newtheorem{proposition}[theorem]{Proposition}
\theoremstyle{definition}
\newtheorem{definition}[theorem]{Definition}

\newenvironment{claimproof}[1][Proof of Claim.]{\noindent {\emph{#1} }}{\hfill$\lrcorner$\medskip}

\setenumerate{label=\textit{(\roman*)},itemsep=0pt}
\let\OLDthebibliography\thebibliography
\renewcommand\thebibliography[1]{
  \OLDthebibliography{#1}
  \setlength{\parskip}{0pt}
  \setlength{\itemsep}{0pt plus 0.3ex}
}
\AtBeginDocument{
   \def\MR#1{}
}

\newcommand{\N}{\mathbb{N}}
\newcommand{\T}{\mathcal{T}}

\newcommand{\C}{\mathcal{C}}
\newcommand{\cF}{\mathcal{F}}
\newcommand{\citeReference}[2]{\cite[#1]{#2}}

\newcommand{\td}{\mathrm{td}}
\newcommand{\tw}{\mathrm{tw}}
\newcommand{\dege}{\mathrm{deg}}
\newcommand{\Sol}{\mathtt{Sol}}
\newcommand{\Oh}{\mathcal{O}}
\newcommand{\carvers}{\mathcal{C}}
\newcommand{\pmc}{\Omega}

\newcommand{\cmsotwo}{$\mathsf{CMSO}_2$\xspace}
\newcommand{\cmso}{$\mathsf{CMSO}$\xspace}
\newcommand{\lmset}{\{\!\{}
\newcommand{\rmset}{\}\!\}}
\newcommand{\Aa}{\mathcal{A}}
\newcommand{\lbl}{\mathsf{label}}
\newcommand{\Multi}{\mathsf{Multi}}
\newcommand{\Tp}{\mathsf{Sentences}}
\newcommand{\tp}{\mathrm{tp}}
\newcommand{\wh}[1]{\widehat{#1}}

\newcommand{\state}{\sigma}

\renewcommand{\phi}{\varphi}
\renewcommand{\leq}{\leqslant}

\renewcommand{\geq}{\geqslant}

\title{Sparse induced subgraphs in $P_6$-free graphs
\thanks{
This research is a part of a project that has received funding from the European Research Council (ERC)
under the European Union's Horizon 2020 research and innovation programme
Grant Agreement 714704 (Rose, Marcin) and 948057 (Micha\l{}, Pawe\l{}).
Maria is supported by NSF-EPSRC Grant DMS-2120644 and by AFOSR grant FA9550-22-1-008. Rose is also supported by NSF Grant DMS-2202961.
Marcin is also partially funded by BARC, supported by the VILLUM Foundation grant 16582, and by Polish National Science Centre SONATA BIS-12 grant number 2022/46/E/ST6/00143.}}
\author{
Maria Chudnovsky
\thanks{Department of Mathematics, Princeton University, USA} \and
Rose McCarty
\thanks{Department of Mathematics, Princeton University, USA and Institute of Informatics, University of Warsaw, Poland} \and
Marcin Pilipczuk
\thanks{Institute of Informatics, University of Warsaw, Poland and IT University of Copenhagen, Denmark} \and
Micha\l{} Pilipczuk
\thanks{Institute of Informatics, University of Warsaw, Poland} \and
Pawe\l{} Rz\k{a}\.{z}ewski
\thanks{Warsaw University of Technology, Poland and Institute of Informatics, University of Warsaw, Poland}
}

\date{}

\begin{document}
\begin{titlepage}
\maketitle
\begin{abstract}
We prove that a number of computational problems that ask for the largest sparse induced subgraph satisfying some property definable in \textsf{CMSO}$_2$ logic, most notably \textsc{Feedback Vertex Set},
are polynomial-time solvable in the class of $P_6$-free graphs. This generalizes the work of Grzesik, Klimo\v{s}ov\'{a}, Pilipczuk, and Pilipczuk on the \textsc{Maximum Weight Independent Set} problem in $P_6$-free graphs~[SODA 2019, TALG 2022],
and of Abrishami, Chudnovsky, Pilipczuk, Rz\k{a}\.zewski, and Seymour on problems in $P_5$-free graphs~[SODA~2021].

The key step is a new generalization of the framework of \emph{potential maximal cliques}. We show that instead of listing
a large family of potential maximal cliques, it is sufficient to only list their \emph{carvers}: vertex sets that contain the same vertices
from the sought solution and have similar separation properties.
\end{abstract}
\def\thepage{}
\thispagestyle{empty}

\begin{textblock}{20}(0, 12.3)
\includegraphics[width=40px]{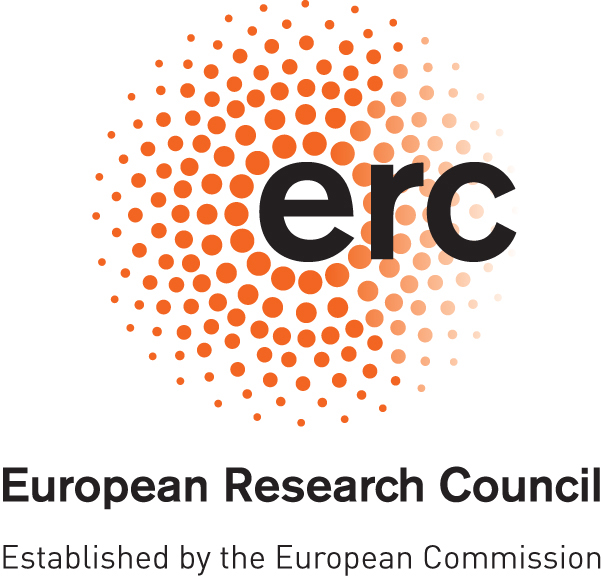}
\end{textblock}
\begin{textblock}{20}(0, 13.1)
\includegraphics[width=40px]{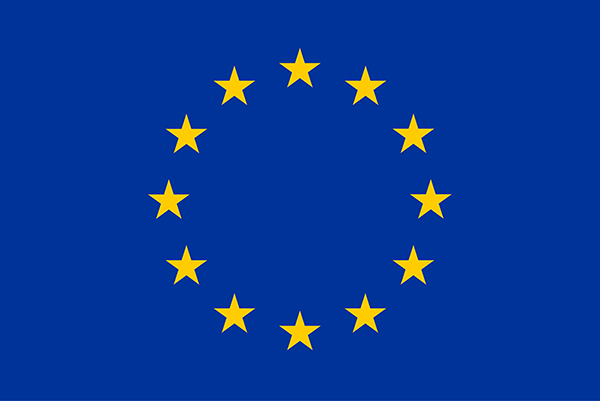}
\end{textblock}

\end{titlepage}

\section{Introduction}
\label{sec:intro}
The landmark work of Bouchitt\'{e} and Todinca~\cite{BouchitteT01} uncovered the pivotal role that \emph{potential maximal cliques}
(PMCs for short) play in tractability of the classic \textsc{Maximum (Weight) Independent Set} problem (\textsc{MIS} or \textsc{MWIS} for short).
The \textsc{MIS} (\textsc{MWIS}) problem asks for a set of pairwise nonadjacent vertices (called an \emph{independent} set or a \emph{stable} set)
in a given graph of maximum possible cardinality (or weight, in the weighted setting, where every vertex is given a positive integral weight).
Without giving a precise definition, a potential maximal clique is a set of vertices of the graph that
can be seen as a ``reasonable'' choice for a bag in a tree decomposition
of the graph, which in turn can be seen as a ``reasonable'' choice of a separating set for a divide-and-conquer algorithm.

Fomin and Villanger~\cite{DBLP:conf/stacs/FominV10} showed that \textsc{MWIS} and its generalizations asking for a maximum induced subgraph of given treewidth are solvable in time polynomial in the size of the graph and \emph{the number of PMCs} of the input graph.
At the time, this result unified a number of earlier tractability results for \textsc{MWIS} and related problems
in various hereditary graph classes, giving an elegant common explanation for tractability.
Later, Fomin, Todinca, and Villanger~\cite{FominTV15} showed that the same result applies to a wide range of combinatorial problems, captured via the following formalism.

For a fixed integer $k$ and a \cmsotwo{} formula
\footnote{\cmsotwo{} stands for monadic second-order logic in graphs with quantification over edge subsets and modular counting predicates. In this logic, one can quantify both over single vertices and edges and over their subsets, check membership and vertex-edge incidence, and apply modular counting predicates with fixed moduli to set variables. See Section~\ref{sec:logic} for a formal introduction of the syntax and semantics of \cmsotwo{}.} $\phi$ with one free vertex set variable, consider the following problem.
Given a graph $G$, find a pair $(\Sol,X)$ maximizing $|X|$ such that $X \subseteq \Sol \subseteq V(G)$, $G[\Sol]$ has treewidth at most $k$, and $\phi(X)$ is satisfied in $G[\Sol]$.
This problem can also be considered in the weighted setting, where vertices of $G$ have positive integral weights
and we look for $(\Sol, X)$ maximizing the weight of $X$.
For fixed $k$ and $\phi$, we denote this weighted problem as $(\tw \leq k,\phi)$-\textsc{MWIS}.
\footnote{Here, \textsc{MWIS} stands for ``maximum weight induced subgraph.''}
Fomin, Todinca, and Villanger showed that $(\tw \leq k,\phi)$-\textsc{MWIS} is solvable in time polynomial
in the size of the graph and the number of its PMCs.
Clearly, \textsc{MWIS} can be expressed as a $(\tw \leq 0, \phi)$-\textsc{MWIS} problem.
Among the many problems captured by this formalism, we mention that \textsc{Feedback Vertex Set} can be expressed as a $(\tw \leq 1, \phi)$-\textsc{MWIS} problem: indeed, the complement of a minimum (weight) feedback vertex set is a maximum (weight) induced forest.

Another application is as follows.
Let $\mathcal{G}$ be a minor-closed graph class that does not contain all planar graphs.
Thanks to the Graph Minor Theorem of Robertson and Seymour~\cite{graphMinors20},
there exists a finite set $\cF$ of graphs such that $G \in \mathcal{G}$ if and only if $G$ does not contain any
graph of $\cF$ as a minor. Consequently, the property of belonging to $\mathcal{G}$ can be expressed in \cmsotwo{}.
Furthermore, as $\mathcal{G}$ does not contain all planar graphs, $\mathcal{G}$ is of bounded treewidth~\cite{graphMinors5}.
Thus the problem of finding a largest induced subgraph that belongs to $\mathcal{G}$
is a special case of $(\tw \leq k, \phi)$-\textsc{MWIS}.

In both applications above we have $\Sol=X$. To see an example where these sets are different, consider the problem of packing the maximum number of vertex-disjoint and pairwise non-adjacent induced cycles.
To see that it is also a special case of $(\tw \leq k, \phi)$-\textsc{MWIS},
let $k=2$ and $\phi$ be the formula enforcing that $G[\Sol]$ is 2-regular (i.e., a collection of cycles) and no two vertices from $X$ are in the same component of $G[\Sol]$.

\medskip
Unfortunately, the results of~\cite{BouchitteT01} and~\cite{FominTV15} do not cover all cases where we expect even the original \textsc{MWIS} problem to be
polynomial-time solvable. A key case arises from excluding an induced path. For a fixed graph $H$, the class of \textit{$H$-free} graphs consists of all graphs that do not contain $H$
as an induced subgraph. For an integer $t$, we denote the path on $t$ vertices by $P_t$.
While the class of $P_4$-free graphs has bounded clique-width,
the class of $P_t$-free graphs does not seem to exhibit any apparent structure for $t \geq 5$.
Still, as observed by Alekseev~\cite{alekseev1982effect,Alekseev03}, \textsc{MWIS} is not known to be \textsf{NP}-hard in $P_t$-free graphs for any fixed $t$.

At first glance, the PMC framework of~\cite{BouchitteT01} does not seem applicable to $P_t$-free graphs for~$t \geq 5$, as even co-bipartite
graphs can have exponentially-many PMCs. (A graph is \textit{co-bipartite} if its complement is bipartite; these graphs are $P_5$-free.)
In 2014, Lokshtanov, Vatshelle, and Villanger~\cite{LokshantovVV14} revisited the framework of Bouchitt\'{e} and Todinca and showed that it is not necessary to use \emph{all} PMCs of the input graph, but only some carefully
selected subfamily of~PMCs.
They also showed that for $P_5$-free graphs, one can efficiently enumerate a suitable family of polynomial size, thus proving
tractability of \textsc{MWIS} in $P_5$-free graphs.
The arguments of~\cite{LokshantovVV14} were then expanded to $P_6$-free graphs by Grzesik et al.~\cite{p6FreeMaxInd22}.
The case of $P_7$-free graphs remains~open.

A general belief is that the \textsc{MWIS} problem is actually tractable in $P_t$-free graphs for any constant $t$.
This belief is supported by the existence of \emph{quasi-polynomial-time} algorithms that work for every $t$~\cite{DBLP:conf/focs/GartlandL20,DBLP:conf/sosa/PilipczukPR21}.
Extending these results, Gartland et al.~\cite{GartlandLPPR21} proved that for every $t$, $k$, and $\phi$, the $(\tw \leq k,\phi)$-\textsc{MWIS} problem is solvable in quasi-polynomial time on $P_t$-free graphs via a relatively simple branching algorithm. Actually, their algorithm solves the $(\dege \leq k,\phi)$-\textsc{MWIS} problem, where instead of a subgraph of bounded treewidth we ask for a subgraph of bounded \emph{degeneracy}.
We remark that $(\tw \leq k,\phi)$-\textsc{MWIS} can be expressed as $(\dege \leq k',\phi')$-\textsc{MWIS}. Indeed, degeneracy is always upper-bounded by treewidth and the property of being of bounded treewidth is expressible by a \cmsotwo{} formula.
On the other hand, the language of $(\dege \leq k,\phi)$-\textsc{MWIS} allows us to capture more problems. One well-known example is \textsc{Vertex Planarization}~\cite{DBLP:conf/soda/JansenLS14,DBLP:journals/dam/Pilipczuk17}, which asks for a maximum (or maximum weight) induced planar subgraph.
Indeed, planar graphs have degeneracy at most 5, but they might have unbounded treewidth, so \textsc{Vertex Planarization} is not a special case of $(\tw \leq k,\phi)$-\textsc{MWIS}.
However, Gartland et al.~\cite{GartlandLPPR21} showed that in (a superclass of) $P_t$-free graphs, treewidth and degeneracy are functionally equivalent.
Consequently, even though $(\dege \leq k,\phi)$-\textsc{MWIS} is more general than $(\tw \leq k,\phi)$-\textsc{MWIS}, in $P_t$-free graphs both formalisms  describe the same family of problems.

\medskip
One of the obstructions towards extending the known polynomial-time algorithms for \textsc{MWIS} beyond $P_5$-free and $P_6$-free graphs is the technical complexity of the method.
The algorithm for $P_5$-free graphs~\cite{LokshantovVV14} is already fairly involved, and the generalization to $P_6$-free graphs~\cite{p6FreeMaxInd22}
resulted in another significant increase in the amount of technical work.
In particular, it is not clear how to apply such an approach to solve $(\tw \leq k,\phi)$-\textsc{MWIS} (or, equivalently, $(\dege \leq k,\phi)$-\textsc{MWIS}).
Furthermore, in a recent note~\cite{GrzesikKPP21}, the authors of~\cite{p6FreeMaxInd22} discuss limitations of applying the method to solving \textsc{MWIS} in graph classes excluding longer paths.

Both algorithms for $P_5$-free~\cite{LokshantovVV14} and for $P_6$-free graphs~\cite{p6FreeMaxInd22} focused on restricting
the family of needed PMCs, but the algorithms still listed the PMCs exactly.
A major twist was made by Abrishami at al.~\cite{containerMethod21}
who showed that instead of determining a PMC exactly, it suffices to find only a \emph{container} for it: a superset
that does not contain any extra vertices from the solution. They also showed that with this container method, the arguments for $P_5$-free graphs
from~\cite{LokshantovVV14} greatly simplify to some elegant structural observations about $P_5$-free graphs.
Moreover, the container method from~\cite{containerMethod21} works with any $(\dege \leq k,\phi)$-\textsc{MWIS} problem. In particular, the authors of~\cite{containerMethod21} showed that \textsc{Feedback Vertex Set} is polynomial-time solvable in $P_5$-free graphs.

While the container method of~\cite{containerMethod21} pushed the boundary of tractability, it does not seem to be easily applicable
to $P_t$-free graphs for $t \geq 6$; in particular, we do not know how to significantly simplify the arguments of~\cite{p6FreeMaxInd22}
using containers.

\subsection{Our contribution}

Our contribution is three-fold.

\paragraph*{Identifying treedepth as the relevant width measure.}
    Previous work on \textsc{MWIS} in $P_5$-free and $P_6$-free graphs~\cite{GrzesikKPP21,LokshantovVV14}
    first fixed a sought solution $I$ (which is an inclusion-wise maximal independent set),
    then observed that it suffices to focus on PMCs that contain at most one vertex from the fixed solution $I$. They also distinguished between
    PMCs that have one vertex in common with $I$, and PMCs that have zero. The former ones turn out to be easy to handle,
    but the latter ones, called ``$I$-free'' or ``$I$-safe,'' are trickier; to tackle them one has to rely on some additional properties stemming from the fact that $I$ is maximal.

    We introduce generalizations of these notions to induced subgraphs of bounded \emph{treedepth},
    a structural notion more restrictive than treewidth.
    It turns out that the correct analog of independent sets are induced subgraphs of bounded treedepth
    with a fixed elimination forest. Maximality corresponds to the inability to extend the subgraph by adding a leaf
    vertex to the elimination forest, while $I$-freeness corresponds to not containing any \emph{leaf} of the fixed
    elimination forest of the sought solution.

    Focusing on treedepth naturally leads us to the $(\td \leq k, \phi)$-\textsc{MWIS} problem, where $G[\Sol]$ is required to have treedepth at most $d$.
    Luckily, in $P_t$-free graphs treedepth is functionally equivalent to treewidth (and thus to degeneracy, too),
    so in this setting the $(\td \leq k, \phi)$-\textsc{MWIS}, $(\tw \leq k, \phi)$-\textsc{MWIS},
    and $(\dege \leq k, \phi)$-\textsc{MWIS} formalisms define the same class of~problems.

    While simple in their form and proofs,
    the above generalizations allow us to adapt many arguments of~\cite{GrzesikKPP21,LokshantovVV14}
    to all $(\td \leq k, \phi)$-\textsc{MWIS} problems.

\paragraph*{Generalizing containers to carvers.}
    We introduce a notion of a \emph{carver} that generalizes containers.
    Our inspiration comes from thinking of a PMC as a ``reasonable'' separation in a divide-and-conquer algorithm.
Instead of determining a PMC exactly, we want to find an ``approximation'' that, on one hand, contains the same vertices from the sought
solution, and, on the other hand, \emph{splits the graph at least as well as the PMC}.
The crux lies in properly defining this latter notion.

Note that a container should satisfy any reasonable definition of ``splitting at least as well.'' Indeed, if $X$ is a set of vertices which contains a PMC $\pmc$,
then each component of $G-X$ is a subset of a component of $G-\pmc$.
However, if we allow that the approximation $X$ of $\pmc$ does not contain some vertices of $\pmc$, then we need to somehow restrict the way the vertices of $\pmc \setminus X$ connect the components of $G-(\pmc \cup X)$.
The first natural idea, to ask that no component of $G-X$ intersects more than one component of $G-\pmc$, turns
out to be not very useful. The actual definition allows $\pmc \setminus X$ to glue up
some components of $G-(\pmc \cup X)$ as long as we can show that
another carver, for a different PMC, will later separate them.

We prove that carvers are sufficient to solve all problems of our interest in $P_t$-free graphs.

\begin{theorem}[informal statement of Theorem~\ref{thm:dp}]\label{thm:informal-dp}
 Any $(\dege \leq k,\phi)$-\textsc{MWIS} problem is solvable on $P_t$-free graphs in time
polynomial in the size of the input graph and the size of the supplied carver family.
\end{theorem}

\paragraph*{Finding carvers in $P_6$-free graphs.}
 We showcase the strength of Theorem~\ref{thm:informal-dp} by lifting the approach of Grzesik et al.~\cite{p6FreeMaxInd22} from just \textsc{MWIS} to arbitrary $(\dege \leq k,\phi)$-\textsc{MWIS} problems on $P_6$-free graphs. Formally, we prove the following.

\begin{theorem}\label{thm:main-algo}
For any choice of $k$ and $\phi$, the $(\dege \leq k,\phi)$-\textsc{MWIS} problem is polynomial-time solvable
on $P_6$-free graphs.
\end{theorem}

Note that Theorem~\ref{thm:main-algo} in particular implies that {\sc{Feedback Vertex Set}} is polynomial-time solvable on $P_6$-free graphs, which was a well-known open problem~\cite{DBLP:conf/wg/PaesaniPR22,DBLP:journals/siamdm/PaesaniPR22,DBLP:journals/algorithmica/BonamyDFJP19}. Apart from being applicable to a wider class of problems, our carver-based approach also significantly simplifies, or even makes obsolete, many of the technical parts of~\cite{p6FreeMaxInd22}.

On high level, the proof of~\cite{p6FreeMaxInd22} consists of two parts. In the first part, PMCs that in some sense
``have more than two principal components'' are analysed. Here, the arguments are arguably neat and elegant in many places.
The second part deals with PMCs with exactly two ``principal components,'' that can chain up into long sequences.
Here, a highly technical replacement argument is developed to ``canonize'' an $I$-free minimal chordal completion
in such parts of the input graph.

Using the newly developed notions of treedepth structures, we lift the (more elegant part of the) arguments of~\cite{p6FreeMaxInd22}
to $(\td \leq k, \phi)$-\textsc{MWIS} problems, showing that PMCs with ``more than two principal components'' admit containers,
not only carvers.
Furthermore, we use the power of the new notion of the carver to construct carvers
for PMCs with two ``principal components,'' replacing the highly technical part of~\cite{p6FreeMaxInd22}
with arguably shorter and more direct arguments.

We refrain from providing a more detailed analysis of the running time
bounds of our algorithms beyond merely stating polynomial time.
Admittedly, the exponent in the polynomial bound of Theorem~\ref{thm:main-algo}
has a terrible dependency on $k$. In our opinion, extracting
a more precise bound would introduce substantial technical clutter
in the proof, while not bringing any new insight.

\subsection{Technical overview}

Let us now have a closer look at the three aforementioned contributions.

To this end, we need to introduce some definitions regarding chordal completions and PMCs.
Given a graph $G$, a set $\Omega \subseteq V(G)$ is a \textit{potential maximal clique} (or a \textit{PMC}) if there exists a minimal chordal completion of $G$ in which $\Omega$ is a maximal clique. A \textit{chordal completion} of $G$ is a supergraph of $G$ which is chordal and has the same vertex-set as $G$; it is \textit{minimal} if it has no proper subgraph which is also a chordal completion of $G$. (Recall that a graph is \textit{chordal} if it has no holes, where a \textit{hole} is an induced cycle of length at least~$4$.) Since chordal completions are obtained by adding edges to $G$, it is convenient to write them as $G+F$, where $F$ is a set of non-edges of $G$.

Chordal completions in a certain sense correspond to tree decompositions, and it is often more convenient to work with the latter.
(The formal definition of a tree decomposition can be found in Section~\ref{sec:prelims}.)
It is a folklore result that a graph $H$ is chordal if and only if it has a tree decomposition whose bags are exactly the maximal cliques of $H$ (meaning, in particular, that the number of nodes of the tree is equal to the number of maximal cliques of $H$). Such a tree decomposition is called a \emph{clique tree} of $H$; note that while the set of bags of a clique tree is defined uniquely, the actual tree part of the tree decomposition is not necessarily unique.

In the other direction, observe that if we have a tree decomposition of a given graph $G$, then by completing every bag of this
tree decomposition into a clique, we obtain a chordal supergraph.
Hence, minimal chordal completions correspond to ``the most refined'' tree decompositions of $G$, and this
supports the intuition that PMCs are ``reasonable'' choices of bags in a tree decomposition of $G$.

For a set $S \subseteq V(G)$ in a graph $G$, a \emph{full component} of $S$ is a connected component $A$ of $G-S$
such that $N(A) = S$. A set $S$ is a \emph{minimal separator} if $S$ has at least two full components.
It is well-known (cf.~\cite{BouchitteT01}) that if $\pmc$ is a PMC in $G$, then for every component $D$ of $G-\pmc$,
$N(D)$ is a minimal separator with $D$ as a full component and another full component containing $\pmc \setminus N(D)$.
Furthermore, if $st$ is an edge of $T$ for a clique tree $(T,\beta)$ of a minimal chordal completion $G+F$,
then $\beta(s) \cap \beta(t)$ is a minimal separator with one full component containing $\beta(s) \setminus \beta(t)$
and one full component containing $\beta(t) \setminus \beta(s)$.
Thus, in some sense, minimal separators are building blocks from which PMCs are constructed. While PMCs
correspond to bags of tree decompositions of~$G$, minimal separators correspond to \emph{adhesions} (intersections
of neighboring bags).

\paragraph{Treedepth structures.} The starting insight of Lokshtanov, Vatshelle, and Villanger~\cite{LokshantovVV14}
is that if $I$ is a maximal independent set in $G$, then by completing $V(G)\setminus I$ into a clique we obtain
a chordal graph (even a split graph), and thus there exists a minimal chordal completion $G+F$ that does not
add any edge incident to $I$; we call such a chordal completion \emph{$I$-free}.
In $G+F$, every maximal clique contains at most one vertex of $I$
and, if $I \cap \pmc = \{v\}$ for a maximal clique~$\pmc$, then $\pmc \subseteq N_G[v]$ and $N_G[v]$ is a good container
for $\pmc$.
They argue that it is sufficient to list a superset of all maximal cliques of $G+F$, and hence it suffices to focus on PMCs of $G$ that are disjoint
from the sought solution $I$. Such PMCs are henceforth called \emph{$I$-free}.

Let $\pmc$ be an $I$-free PMC. Since $I$ is maximal, every $v \in \pmc$ has a neighbor in $I$ that is outside $\pmc$,
as $\pmc$ is $I$-free. The existence of such neighbors is pivotal to a number of proofs of~\cite{LokshantovVV14,p6FreeMaxInd22}.

To discuss our generalization to induced subgraphs of bounded treedepth, we need a few standard definitions.
A \textit{rooted forest} is a forest $\T$ where each component has exactly one specified vertex called its \textit{root}. The \textit{depth} of a vertex $v \in V(\T)$ is the number of vertices in the unique path from $v$ to a root (so roots have depth~$1$). The \textit{height} of $\T$ is the maximum depth of any of its vertices. A path in $\T$ is \textit{vertical} if one of its ends is an ancestor of the other. (We consider each vertex to be both an ancestor and a descendent of itself.) Two vertices are \textit{$\T$-comparable} if they are connected by a vertical path; otherwise they are \textit{$\T$-incomparable}.
An \textit{elimination forest} of a graph $G$ is a rooted forest $\T$ such that $V(\T) = V(G)$ and the endpoints of each edge of $G$ are $\T$-comparable. The \textit{treedepth} of $G$ is then the smallest integer $d$ such that $G$ has an elimination forest of height $d$.

Let us now move to the new definitions.
Let $G$ be a graph and $d$ be a positive integer. A \emph{treedepth-$d$ structure in $G$} is a rooted forest $\T$ of height at most $d$ such that $V(\T)$ is a subset of $V(G)$ and $\T$ is an elimination forest of the subgraph of $G$ induced by $V(\T)$. We say that $\T$ is \emph{maximal} if there is no treedepth-$d$ structure $\T'$ in $G$ such that $\T$ is a proper induced subgraph of $\T'$ and every root of $\T$ is a root of $\T'$. In other words, $\T$ is maximal if one cannot extend it by appending a leaf while preserving the bound on the height.

Note that if $H$ is a maximal induced subgraph of $G$ of treedepth at most $d$, and $\T$ is a height-$d$ elimination forest of that subgraph, then $\T$ is a maximal treedepth-$d$ structure in $G$. Consequently, in the context of $(\td \leq d, \phi)$-\textsc{MWIS}, we can consider $\Sol$ as being in fact a maximal set inducing a subgraph of treedepth at most $d$ in $G$; if $(\Sol, X)$ is an actual solution, then there exists a maximal treedepth-$d$ structure $\Sol'$ that is a superset of $\Sol$, and

we can extend $\phi$ by saying that there exists a set $\Sol \subseteq \Sol'$ with all the desired properties.
Thus, most of the structural results in this work consider the set of all maximal treedepth-$d$-structures, which are more detailed versions of maximal sets inducing a subgraph of treedepth at most $d$.

Recall that for any independent set $I$, there is a minimal chordal completion of $G$ that is $I$-free, that is,
it does not add any edge incident to $I$.
This statement generalizes to chordal completions \emph{aligned} with a given treedepth-$d$ structure $\T$;
we say that a chordal completion $G+F$ is \textit{$\T$-aligned} if $F$ does not contain any pair $uv$ such that
\begin{enumerate}
    \item $u$ or $v$ is a depth-$d$ vertex of $\T$, or
    \item $u$ and $v$ are vertices of $\T$ which are $\T$-incomparable.
\end{enumerate}
\noindent The second condition equivalently says that $\T$ is a treedepth-$d$ structure in $G+F$.

We show that there is always a $\T$-aligned minimal chordal completion (see Lemma~\ref{lem:alignedChordalCompl}) and argue that it suffices to focus on PMCs
that come from an aligned minimal chordal completion.

The analog of the notion of ``$I$-freeness'' is as follows:
A PMC $\Omega$ is \textit{$\T$-avoiding} if it is a maximal clique of a minimal chordal completion that is $\T$-aligned, and it does not contain any depth-$d$ vertex of $\T$.
Similarly as in the case of PMCs that are not $I$-free,
if $\pmc$ is $\T$-aligned but not $\T$-avoiding, it contains exactly one vertex of $\T$ of depth $d$ and
one can argue that the closed neighborhood of such vertex gives rise to a container for $\pmc$
(after excluding the vertices of $\T \setminus \pmc$ that accidentally got into it, but there are at most $d-1$
one of them, because they all are ancestors of the guessed\footnote{Throughout this paper, by \emph{guessing} we mean branching into polynomially many choices of fixing the object in question.} vertex of $\T \cap \pmc$ of depth $d$ in the rooted
forest $\T$).

Thus, it remains to focus on $\T$-avoiding PMCs.
In the $I$-free setting, the important property of an $I$-free PMC was that every $v \in \pmc$ has a neighbor in $I$.
Here, one can argue that every $v \in \pmc \setminus \T$ in a $\T$-avoiding PMC $\pmc$ has a neighbor
in $\T \setminus \pmc$, as otherwise it can be added to $\T$ without increasing the maximum depth of $\T$, contradicting
the maximality of $\T$.

This concludes the overview of the adaptation of the notion of $I$-freeness to induced subgraphs of bounded treedepth.

\paragraph{Carvers.}
Let $G$ be a graph and let $I$ be an optimal solution to \textsc{MWIS} in $G$.
Assume that we are given a polynomial-sized family $\mathcal{F}$ of PMCs in $G$ that contains all maximal cliques of some
$I$-free minimal chordal completion $G+F$ of $G$. The crucial insight of~\cite{LokshantovVV14} is that this is enough
to solve \textsc{MWIS} in $G$ in polynomial time by a dynamic programming algorithm.
The algorithm considers the following set of states: for every $\pmc \in \mathcal{F}$, every $J \subseteq \pmc$ of size at most $1$,
and every component $D$ of $G-\pmc$, it tries to compute the best possible independent set $I[\pmc,J,D]$ in $G[\pmc \cup D]$
with $I[\pmc,J,D] \cap \pmc = J$.
The assumption that $\mathcal{F}$ contains all PMCs of $G+F$ allows one to argue that there is a computation path
of this dynamic programming algorithm that finds an independent set that is at least as good as $I$
(we may not find $I$ itself).

The crucial insight of~\cite{containerMethod21} is that for the dynamic programming algorithm to work, it is enough
to know \emph{containers} for the maximal cliques of $G+F$, that is, it is fine if the provided sets in $\mathcal{F}$ are larger,
as long as they do not contain extra vertices from the sought solution. The intuition here is that the dynamic programming
algorithm relies on the separation properties of PMCs as bags of a clique tree of $G+F$, and a superset is an even better separator
than a PMC itself.

From the point of view of separation, the following relaxation of a container would suffice.
A set $X$ is a \emph{weak container} of a PMC $\pmc$ if it contains the same vertices from the sought solution
and every connected component of $G-X$ intersects at most one connected component of $G-\pmc$
(that is, the vertices of $\pmc \setminus X$ do not connect two components of $G-(\pmc \cup X)$).

However, in the context of $P_6$-free graphs, we are unable to provide even weak containers to some PMCs,
and there seems to be a good reason for this failure.
Namely, there are examples of $P_6$-free graphs $G$
with an ($I$-free or $\T$-avoiding, depending on the problem
we are solving) PMC $\pmc$ with a subset $\mathcal{D}$ of components of $G-\pmc$ such that some local modifications
to the minimal chordal completion $G+F$ modify $\pmc$ slightly, but completely reshuffle the vertices of $\mathcal{D}$
into new components.
The intuition is that the dynamic programming algorithm should not attempt to separate $\mathcal{D}$ into components
while looking at a (weak) container of $\pmc$, but while looking at another PMC $\pmc'$ that is ``closer'' to $\mathcal{D}$.

\begin{figure}[tb]
\begin{center}
\small

\begin{tikzpicture}[scale=1.0]
\tikzstyle{ivertex}=[circle,draw=black,fill=red,minimum size=0.15cm,inner sep=0pt]
\tikzstyle{vertex}=[circle,draw=black,fill=black,minimum size=0.1cm,inner sep=0pt]

\draw[rounded corners] (0,0) rectangle (1, 7);
\draw[rounded corners] (3,0) rectangle (4, 7);
\foreach \n in {1,2,3,4,5,6,7} {
  \begin{scope}[shift={(0,7-\n)}]
    \draw (0.5, 0.5) node {$A_\n$};
    \draw (3.5, 0.5) node {$B_\n$};
    \foreach \x in {1,2,3,4} {
      \foreach \y in {1,2,3,4} {
        \draw (1, \x*0.2) -- (3, \y*0.2);
      }
    }
  \end{scope}
}
\foreach \m in {0,3} {
\foreach \n in {1,2,3,4,5,6} {
  \begin{scope}[shift={(\m,7-\n)}]
    \draw (0,0) -- (1,0);
    \foreach \x in {1,2,3,4} {
      \foreach \y in {1,2,3,4} {
        \draw (\x*0.2, 0.1) -- (\y*0.2, -0.1);
      }
    }
  \end{scope}
}}

\foreach \n in {1,2,3,4} {
  \node[ivertex] (ui) at (3.85, 5+\n*0.2) {};
}

\node[ivertex] (va) at (-2, 1.5) {};
\node[ivertex] (vb) at (-2, 3.5) {};
\node[ivertex] (vc) at (-2, 5.5) {};

\draw[left] (va) node {$v_{\{5,6,7\}}$};
\foreach \x in {5,6,7} {
  \foreach \y in {1,2,3,4} {
  \draw (va) -- (0, 7-\x+\y*0.2);
  }
}

\draw[left] (vb) node {$v_{\{2,3,4,5\}}$};
\foreach \x in {2,3,4,5} {
  \foreach \y in {1,2,3,4} {
  \draw (vb) -- (0, 7-\x+\y*0.2);
  }
}
\draw[left] (vc) node {$v_{\{1,2\}}$};
\foreach \x in {1,2} {
  \foreach \y in {1,2,3,4} {
  \draw (vc) -- (0, 7-\x+\y*0.2);
  }
}
\end{tikzpicture}
\caption{An example of a $P_6$-free graph with a maximal independent set where a weak container seems to be too restrictive a notion.
  The red vertices are the vertices of a maximal independent set.
    Here, $n=7$, $i_0 = 2$, and $\mathcal{F} = \{\{1,2\}, \{2,3,4,5\}, \{5,6,7\}\}$.}\label{fig:intro-example}
\end{center}
\end{figure}
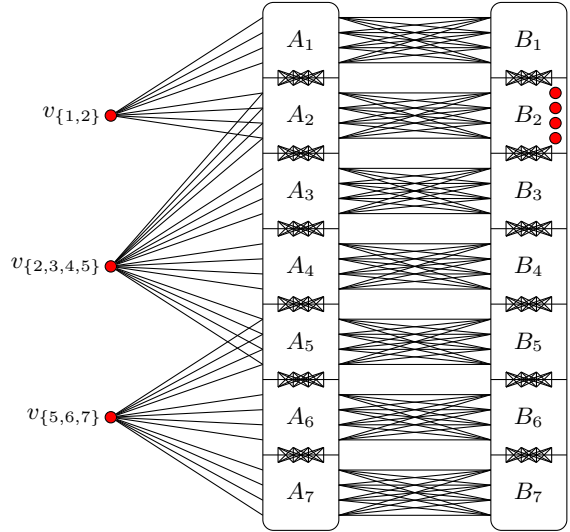

More precisely, consider the following example (cf. Figure~\ref{fig:intro-example}).
Let $A_1,\ldots,A_n$ and $B_1,\ldots,B_n$ be two sequences of $P_6$-free graphs and let $\mathcal{F}$
be a family of subsets of $[n]$ of size being a large polynomial in $n$ with $\bigcup \mathcal{F} = [n]$;
all subsets of $[n]$ of size at most $C$ for a large constant $C$ would do the job.
Construct a graph $G$ as follows.
Start with a disjoint union of $A_1,\ldots,A_n$ and $B_1,\ldots,B_n$. For every $i,j \in [n]$, $i \neq j$,
add all edges between $A_i$ and $A_j$ and all edges between $B_i$ and $B_j$. For every $i \in [n]$, add all edges between
$A_i$ and $B_i$. Finally, for every $K \in \mathcal{F}$, introduce a vertex $v_K$ and make it adjacent to $\bigcup_{i \in K} A_i$.
A direct check shows that $G$ is $P_6$-free, for every choice of $i_0 \in [n]$ and a maximal independent set $I_0$ in $B_{i_0}$,
the set $I_{i_0,I_0} := I_0 \cup \{v_K~|~K \in \mathcal{F}\}$ is a maximal independent set in $G$,
and, for every $\emptyset \neq J \subsetneq [n]$ that is not contained in any set of $\mathcal{F}$, the set
$S_J := \bigcup_{i \in J}A_i \cup \bigcup_{i \in [n] \setminus J} B_i$ is a minimal separator with one full component
$B_J := \bigcup_{i \in J}B_i$ and a second full component $A_J := \bigcup_{i \in [n] \setminus J} A_i \cup \{v_K~|~K \in \mathcal{F}, K \not\subseteq J\}$, and a number of single-vertex components $\{v_K\}$ for $K \in \mathcal{F}, K \subseteq J$.
Observe that a weak container for $S_J$ should separate $v_K$ for $K \subseteq J$ from those $v_K$ for which $K \not\subseteq J$.
The only way to make a small family of (weak) containers for all such separators $S_J$ is to make containers containing whole
$\bigcup_{i \in I} A_i$ but none of the vertices $v_K$; however, distinguishing $\bigcup_{i \in I} A_i$ and $\{v_K~|~K \in \mathcal{F}\}$
seems difficult using the toolbox used in~\cite{LokshantovVV14,p6FreeMaxInd22}.

Let us now consider how an $I_{i_0,I_0}$-free chordal completion
may look like.
As it needs to leave $I_0$ intact, to complete $C_4$s that take
two vertices in $I_0$ and two vertices in other set $B_i$ or in $A_{i_0}$,
the completion has to turn every $B_i$ for $i \in [n] \setminus \{i_0\}$
as well as $A_{i_0}$ into a clique.
To complete all $C_4s$ spanning across two sets $A_i$,
the completion has to turn every $A_i$ for $i \in [n]$ into a clique
(except for possibly one set $A_i$, which we ignore here).
To complete $C_4$s of the form $A_i-A_j-B_j-A_i$ for some $1 \leq i \neq j \leq n$, it needs to take any permutation $\pi$ of $[n]$ with $\pi(1) = i_0$
and add edges between $A_{\pi(i)}$ and $B_{\pi(j)}$ for every $1 \leq i < j \leq n$. This corresponds to turning
$S_J$ for $J = \pi(\{1,2,\ldots,i\})$ for every $1 \leq i \leq n$ into a clique.
Intuitively, the algorithm should not bother with the choice of $\pi$, which corresponds to ignoring
how vertices $v_K$ are separated while looking at intermediate separators $S_J$.

Recall that the correctness of the dynamic programming algorithm of~\cite{LokshantovVV14} relies on the observation
that a clique tree of $G+F$ provides a computation path in which the algorithm finds a solution at least as good as $I$.
In our setting, our trouble is a (possibly large) family $\mathcal{D}$
of components of $G-\Omega$ that we do not want to separate while
processing a container for $\Omega$, but in a different place,
``closer'' to the vertices of $\mathcal{D}$.
From the point of view of the dynamic programming algorithm, this is
not a problem as long as all vertices of $\mathcal{D}$
are confined
in one subtree of a clique tree of $G+F$. This consideration brings us to the final definition of a carver.

\begin{definition}\label{def:intro:carvers}
Let $G$ be a graph and $d$ and $k$ be positive integers.
A family $\carvers \subseteq 2^{V(G)}$ is a \emph{tree-depth-$d$ carver family} of \emph{defect $k$} in $G$
if for every treedepth-$d$ structure $\T$ in $G$,

there exists a tree decomposition $(T,\beta)$ of $G$
such that for each $t\in V(T)$ there exists $C \in \carvers$
such that
\begin{enumerate}
\item $C \cap \T$ contains $\beta(t)\cap \T$ and has size at most $k$, and
\item each component of $G-C$ is contained in $\beta(t) \cup \bigcup_{s \in T'}\beta(s)$ for some component $T'$ of $T-\{t\}$.
\end{enumerate}
Such a set $C$ as above is called a \emph{$(\T, (T,\beta))$-carver} for $\beta(t)$ of defect $k$; it might not be unique. We use this definition independently of that of carver families.
\end{definition}

We prove that this definition works as intended:
a tree-depth-$d$ carver family of small defect in $G$ is enough to design a dynamic programming
routine that solves the $(\td \leq d, \phi)$-\textsc{MWIS} problem on $G$.
\begin{theorem}\label{thm:intro:dp}
For any positive integers $d$ and $k$ and any \cmsotwo formula $\phi$, there exists an algorithm that, given a vertex-weighted graph $G$
and a tree-depth-$d$ carver family $\carvers \subseteq 2^{V(G)}$ of defect $k$ in $G$, runs in time polynomial in the input size
and either outputs an optimal solution to the $(\td \leq d, \phi)$-\textsc{MWIS} problem on~$G$, or determines that no feasible solution exists.
\end{theorem}

We remark that the proof of Theorem~\ref{thm:intro:dp} is far from being just an involved verification of a natural approach.
There is a significant technical hurdle coming from the fact that, with fixed $\T$ and $(T,\beta)$,
carvers for neighboring bags of $(T,\beta)$ may greatly differ from each other in terms of the amount of non-solution vertices
added to them. One needs to design careful tie-breaking schemes for choices in partial solutions in the dynamic programming algorithm
in order to avoid conflicting tie-breaking decisions made while looking at different carvers.

\paragraph{Application to $P_6$-free graphs.}
The starting point of the work of~\cite{p6FreeMaxInd22} on \textsc{MWIS} in $P_6$-free graph
is an analysis of minimal separators that identifies a crucial case distinction between full components
of a minimal separator, into ones whose complement is disconnected (so-called \emph{mesh components}) or connected
(\emph{non-mesh components}).
The analysis splits minimal separators in a $P_6$-free graph $G$ into three categories:
\begin{description}
    \item[Simple,] being a proper subset of another minimal separator, having more than two full components,
    or having two non-mesh full components. Here, one can enumerate a polynomial-sized family of candidates that contains
    all such separators.
    \item[Somewhat complicated,] having exactly two full components, both being mesh.
    Here, one can enumerate a polynomial-sized family that contains a  ``weak container'' for every such separator,
    which is equally good for our applications as just knowing the separator exactly.
    \item[Really complicated,] having exactly two full components, one mesh and one non-mesh.
    Here, we can only enumerate a polynomial-sized family of ``semi-carvers'' that separate the mesh component from the
    other components, but such a semi-carver is not guaranteed to separate the non-mesh full component from some non-full components.
    (This weakness corresponds to examples mentioned earlier about inability to split some family $\mathcal{D}$
    of components of $G-\pmc$ for a PMC $\pmc$; note that all components $S_J$ in the aforementioned examples
    are of the really complicated type.)
\end{description}
This analysis generalizes to our setting, using the new notions of treedepth structures.

We proceed to discussing the PMCs.
Then, the following case distinction is identified in~\cite{p6FreeMaxInd22}.
A potential maximal clique $\pmc$ in a graph $G$ is \emph{two-sided} if there are exactly two components $D_1,D_2$ of $G-\pmc$ with maximal neighborhoods in $\pmc$.  More precisely, for every connected component $D$ of $G-\pmc$, we have $N(D) \subseteq N(D_1)$
or $N(D) \subseteq N(D_2)$.

The following statement has been essentially proven in~\cite{p6FreeMaxInd22}.
However, it has been proven only with the \textsc{Max Weight Independent Set} problem in mind,
so we need to adjust the argumentation using the notion of  treedepth structures.
\begin{theorem}\label{thm:intro:not-two-sided}
For every positive integer $d$
there exists a polynomial-time algorithm that, given a $P_6$-free graph $G$
outputs a family $\carvers \subseteq 2^{V(G)}$ with the following guarantee:
for every maximal treedepth-$d$ structure $\T$ in $G$
and every potential maximal clique $\pmc$ of $G$ that is $\T$-avoiding and not two-sided,
there exists $C \in \carvers$ that is a container for $\pmc$, i.e., $\pmc \subseteq C$
and $C \cap V(\T) = \pmc \cap V(\T)$.
\end{theorem}

It remains to study two-sided PMCs, which were the main cause of technical hurdles in~\cite{p6FreeMaxInd22}.
Here we depart from the approach of~\cite{p6FreeMaxInd22} and use the power of carvers instead.

To use carvers, we would like to choose not only a minimal chordal completion $G+F$ (which, following the developments in the first
part of our work, would be any $\T$-aligned minimal chordal completion, where $\T$ is the sought solution) but
also a clique tree $(T,\beta)$ of $G+F$.
Recall that adhesions in $(T,\beta)$ correspond to minimal separators in $G$, and the really complicated minimal separators
are the ones with one mesh and one non-mesh full component, in which case it is difficult to isolate the non-mesh component.
So, we would like the clique tree $(T,\beta)$ to be imbalanced in the following way:
if $st \in E(T)$ is such that $\beta(s) \cap \beta(t)$ is a really complicated minimal separator with
the non-mesh full component $A_s$ containing $\beta(s) \setminus \beta(t)$ and
the mesh full component $A_t$ containing $\beta(t) \setminus \beta(s)$, then as much as possible of the decomposition $(T,\beta)$
should be reattached to the component of $T-\{st\}$ that contains $s$.

More precisely, for a clique tree $(T,\beta)$ of $G+F$, for every edge $st$ as above, orient $st$ from $t$ to $s$
(and keep all edges of $T$ that do not correspond to really complicated minimal separators undirected).
Consider now an edge $st$ as above and assume that there exists $s' \in N_T(t)$, $s \neq s'$ such that
\begin{equation}\label{eq:sst}
\beta(s') \cap \beta(t) \subseteq \beta(s) \cap \beta(t).
\end{equation}
Then, the minimal separator $\beta(s') \cap \beta(t)$ is a simple one (it is contained in another minimal separator
$\beta(s) \cap \beta(t)$), so the edge $s't$ is undirected. Observe that the assumption~\eqref{eq:sst} allows the following
modification of $(T,\beta)$: replace the edge $s't$ with an edge $s's$.
This modification corresponds to the intuition that while studying the really complicated minimal separator
$\beta(s) \cap \beta(t)$, it is difficult to separate the component $A_s$ from the full component of $G-(\beta(s') \cap \beta(t))$
that contains $\beta(s') \setminus \beta(t)$, and thus --- from the point of view of the PMC $\beta(t)$ --- both these components
should be contained in bags of the same component of $T-\{t\}$.

A simple potential argument shows that such modifications cannot loop indefinitely and there exists a clique tree $(T,\beta)$
where no modification is possible. This is the clique tree for which we are finally able to construct carvers using
the aforementioned analysis of minimal separators, in particular semi-carvers for the really difficult minimal separators.
The actual construction is far from straightforward, but arguably simpler than the corresponding argumentation of~\cite{p6FreeMaxInd22}
that handles two-sided PMCs.

\subsection{Organization}
After the preliminaries (Section~\ref{sec:prelims}), we
introduce the notion of carvers and carver families and provide the main algorithmic engine in Section~\ref{sec:dynProg}.
The remaining sections are devoted to $P_6$-free graphs and the proof of Theorem~\ref{thm:main-algo}.
Sections~\ref{sec:minSepCarving} and~\ref{sec:impMixedSeps} study approximate guessing of minimal separators.
Section~\ref{sec:not2Sided} recalls the main (and most elegant) structural results of $P_6$-free graphs of~\cite{p6FreeMaxInd22},
essentially extracting from~\cite{p6FreeMaxInd22} a family of containers for all PMCs that in some sense have ``more than two sides.''
Section~\ref{sec:2SidedAlignedPMCs} uses the results for minimal separators of Sections~\ref{sec:minSepCarving} and~\ref{sec:impMixedSeps}
to provide carvers for the remaining PMCs; this is the place where we crucially rely on the fact that we want to provide only carvers,
not containers. Finally, Section~\ref{sec:wrapUp} wraps up the proof of Theorem~\ref{thm:main-algo}, and Section~\ref{sec:conclusions} gives a concluding remark about $P_7$-free graphs.

\section{Preliminaries}\label{sec:prelims}
\label{sec:treeDepthDecomp}
We use standard graph-theoretic notation, and all graphs are simple, loopless, and finite. We consider the edge-set of a graph $G$ to be a subset of $\binom{V(G)}{2}$, which is the set of all $2$-element subsets of $V(G)$. We write $uv$ for an element $\{u,v\}$ of $\binom{V(G)}{2}$. A \textit{non-edge} of $G$ is then a pair of vertices $uv$ which is not in $E(G)$. Given a set $F \subseteq \binom{V(G)}{2}$, we write $G+F$ for the graph with vertex set $V(G)$ and edge set $E(G) \cup F$; so $G+F$ is obtained from $G$ by adding all pairs from $F$ as edges if they were not already present.

Given a graph $G$ and a set of vertices $S \subseteq V(G)$, we write $N(S)$ and $N[S]$, respectively, for the open and closed neighborhood of $S$ in $G$. That is, $N(S) \coloneqq \{u \in V(G)-S: uv \in E(G)$ for some $v \in S\}$ and $N[S] \coloneqq S \cup N(S)$. We do not distinguish between induced subgraphs and their vertex sets, except when it might cause confusion. So we typically use $S$ and $G[S]$ interchangeably. Finally, if $v_0, v_1, \ldots, v_k$ are distinct vertices of $G$, then we write $N(v_0, v_1, \ldots, v_k)$ for $N(\{v_0, v_1, \ldots, v_k\})$ and $N[v_0, v_1, \ldots, v_k]$ for $N[\{v_0, v_1, \ldots, v_k\}]$.

We use the following notation to talk about paths. If $X_1, X_2, \ldots, X_k \subseteq V(G)$, then a $P_k$ \textit{of the form $X_1 X_2 \ldots X_k$} is an induced copy of $P_k$ in $G$ so that the first vertex is in $X_1$, the second vertex is in $X_2$, and so on. If $X_i = \{v\}$ for some vertex $v$, then we may put
$v$ instead of $X_i$ in the sequence denoting the form. For instance, given a vertex $v$ and a set $A \subseteq V(G)$, a $P_4$ of the form $vAAA$ is one that
starts at a vertex $v$ and has the rest of its vertices in $A$.

We say that two disjoint sets $X,Y \subseteq V(G)$ are \textit{complete} if every vertex in $X$ is adjacent to every vertex in $Y$. If $X = \{v\}$ for some vertex $v$, then we say that $v$ and $Y$ are \textit{complete}. Similarly, we say that two disjoint sets, or a vertex and a set not containing that vertex, are \textit{anticomplete} if they are complete in the complement of $G$. The complement of $G$ is denoted by $\overline{G}$.

The following observation is straightforward and will be often used implicitly.

\begin{observation}
Let $G$ be a graph, $X$ be a connected subset of $V(G)$, and $v \in V(G)\setminus X$ be neither complete nor anticomplete to $X$.
Then there exists a $P_3$ of the form $vXX$.
\end{observation}

\subsection{Logic}\label{sec:logic} In this paper we use the logic \cmsotwo, which stands for monadic second-order logic with quantification over edge subsets and modular counting predicates, as a language for expressing graph problems. In this logic we have variables of four sorts: for single vertices, for single edges, for vertex subsets, and for edge subsets. The latter two types are called {\em{monadic variables}}. Atomic formulas of \cmsotwo are as follows:
\begin{itemize}
    \item equality $x=y$ for any two variables $x,y$ of the same sort;
    \item membership $x\in X$, where $X$ is a monadic variable and $x$ is a single vertex/edge variable;
    \item modular counting predicates of the form $|X|\equiv a\bmod m$, where $X$ is a monadic variable and $a,m$ are integers, $m\neq 0$; and
    \item incidence $\mathsf{inc}(x,f)$, checking whether vertex $x$ is incident to edge $f$.
\end{itemize}
Then \cmsotwo consists of all formulas that can be obtained from the atomic formulas by means of standard boolean connectives, negation, and universal and existential quantification (over all sorts of variables). This gives the syntax of \cmsotwo, and the semantics is obvious.

Note that a formula may have {\em{free variables}}, which are variables not bound by any quantifier. A formula without free variables is called a {\em{sentence}}.

It will be sometimes useful to consider graphs with some vertex annotations.
Formally, for a finite set $\Sigma^0$, a \emph{$\Sigma^0$-annotated graph $G$} is a graph $G$
together with a function $\lambda^0 : V(G) \to \Sigma^0$. A \cmsotwo{} formula $\phi$
over $\Sigma^0$-annotated graphs has additionally access to atomic formulas $\lambda^0(x) = \sigma$ for a single
vertex variable $x$ and $\sigma \in \Sigma^0$. The $(\tw \leq k, \phi)$-\textsc{MWIS} problem
naturally generalizes to $\Sigma^0$-annotated graphs $G$.

Logic \cmsotwo is usually associated with tree-like graphs through the following fundamental result of Courcelle~\cite{Courcelle90}: given an $n$-vertex graph $G$ of treewidth at most $k$ and a sentence $\phi$ of \cmsotwo, one can determine whether $\phi$ holds in $G$ in time $f(k,\phi)\cdot n$, for a computable function $f$. The proof of this result brings the notion of {\em{tree automata}} to the setting of tree-like graphs, which is a connection that will be also exploited in this work. For an introduction to this area, see the monograph of Courcelle and Engelfriet~\cite{CourcelleE12}.

\subsection{Treewidth and treedepth}
We now introduce treedepth because it turns out to be a more natural width parameter than treewidth in the context of $P_t$-free graphs. It is convenient to begin with some definitions on forests.

A \textit{rooted forest} is a forest $\T$ where each component has exactly one specified vertex called its \textit{root}. The \textit{depth} of a vertex $v \in V(\T)$ is the number of vertices in the unique path from $v$ to a root (so roots have depth~$1$). The \textit{height} of $\T$ is the maximum depth of any of its vertices. A path in $\T$ is \textit{vertical} if one of its ends is an ancestor of the other. (We consider each vertex to be both an ancestor and a descendent of itself.) Two vertices are \textit{$\T$-comparable} if they are connected by a vertical path; otherwise they are \textit{$\T$-incomparable}.

An \textit{elimination forest} of a graph $G$ is a rooted forest $\T$ such that $V(\T) = V(G)$ and the endpoints of each edge of $G$ are $\T$-comparable. The \textit{treedepth} of $G$ is then the smallest integer $d$ such that $G$ has an elimination forest of height $d$. Finally, we define the problem $(\td \leq d, \phi)$-\textsc{MWIS} analogously to $(\tw \leq k, \phi)$-\textsc{MWIS}, where the only difference
is that $G[\Sol]$ is required to have treedepth at most $d$ (instead of treewidth at most $k$).

It is well known that for every graph, its degeneracy is upper-bounded by treedepth, which in turn is upper-bounded by treewidth.
On the other hand, there are graphs of bounded degeneracy and unbounded treewidth (e.g., grids) and graphs with bounded treewidth and unbounded pathwidth (e.g., paths). However, in the context of $P_t$-free graphs, all these parameters are functionally equivalent due to the following theorem.

\begin{theorem}\label{thm:deg2td}
For any integers $t$ and $\ell$, there exists an integer $d$ such that if $G$ is a $P_t$-free graph with degeneracy at most $\ell$, then the treedepth of $G$ is at most~$d$.
\end{theorem}
Theorem~\ref{thm:deg2td} has been discussed in~\cite{GartlandLPPR21}, but let us recall the reasoning.
The first step is the following result of~\cite{GartlandLPPR21}.
(A graph is \textit{$C_{>t}$-free} if it does not contain a cycle longer than $t$ as an induced subgraph; note that the class of $C_{>t}$-free graphs is a proper superclass of the class of $P_t$-free graphs.)
\begin{theorem}[\cite{GartlandLPPR21}]\label{thm:deg2tw}
For every pair of integers $\ell$ and $t$, there exists an integer $k \in (\ell t)^{\Oh(t)}$ such that
every $C_{>t}$-free graph of degeneracy at most $\ell$ has treewidth at most $k$.
\end{theorem}

Treewidth and treedepth are functionally equivalent on $P_t$-free graphs by the following result of~\cite{BONAMY2022353}.

\begin{theorem}[{\cite[Lemma 29]{BONAMY2022353}}]
For any integer $t$, if $G$ is a $P_t$-free graph, then
\[ \mathrm{treedepth}(G) \leq (\mathrm{treewidth}(G) + 1)^{t-1}. \]
\end{theorem}

Since the property of having treewidth at most $k$ and the property of having treedepth at most $d$
can be expressed in \cmsotwo{} (for more on expressibility of \cmsotwo{}, see \cite[Section 7.4]{DBLP:books/sp/CyganFKLMPPS15}), we obtain that the $(\tw \leq k,\phi)$-\textsc{MWIS} and $(\td \leq d,\phi)$-\textsc{MWIS}
formalisms describe the same class of problems in $P_t$-free graphs for any fixed $t$; every $(\tw \leq k,\phi)$-\textsc{MWIS} problem
has an equivalent definition as a $(\td \leq d,\phi')$-\textsc{MWIS} for some $d$ and $\phi'$ depending on $k$ and $\phi$,
and vice-versa. Hence, in this paper we can focus on solving problems formulated in the $(\td \leq d,\phi)$-\textsc{MWIS}
formalism.

\subsection{Chordal completions and PMCs}

Recall that our overall approach is based on potential maximal cliques. We introduce this approach now.

Given a graph $G$, a set $\Omega \subseteq V(G)$ is a \textit{potential maximal clique} (or a \textit{PMC}) if there exists a minimal chordal completion of $G$ in which $\Omega$ is a maximal clique. A \textit{chordal completion} of $G$ is a supergraph of $G$ which is chordal and has the same vertex-set as $G$; it is \textit{minimal} if it has no proper subgraph which is also a chordal completion of $G$. (Recall that a graph is \textit{chordal} if it has no holes, where a \textit{hole} is an induced cycle of length at least~$4$.) Since chordal completions are obtained by adding edges to $G$, it is convenient to write them as $G+F$, where $F$ is a set of non-edges of $G$.

The following classic result characterizes PMCs.

\begin{proposition}[\citeReference{Theorem~3.15}{BouchitteT01}]
\label{prop:PMCsChar}
Given a graph $G$, a set $\Omega \subseteq V(G)$ is a PMC if and only if both of the following conditions hold. \begin{enumerate}
    \item For each component $D$ of $G-\Omega$, $N(D)$ is a proper subset of $\Omega$.
    \item If $uv$ is a non-edge of $G$ with $u,v \in \Omega$, then there exists a component $D$ of $G - \Omega$ such that $u,v \in N(D)$.
\end{enumerate}
\end{proposition}

Chordal completions in a certain sense correspond to tree decompositions and it is often more convenient to work with the latter. So recall that a \emph{tree decomposition} of a graph $G$ is a pair $(T, \beta)$ such that $T$ is a tree, $\beta$ is a function from $V(T)$ to $2^{V(G)}$, and the following conditions are satisfied:\begin{enumerate}
    \item for each $u \in V (G)$, the set $\{t \in V(T): u \in \beta(t)\}$ induces a non-empty and connected subtree of $T$, and
    \item for each $uv \in E(G)$, there is a node $t$ of $T$ such that ${u, v} \subseteq \beta(t)$.
\end{enumerate}For a node $t$ of $T$, the set $\beta(t)$ is called the \textit{bag} of $t$, and for an edge $st \in E(T)$, the set $\beta(s) \cap \beta(t)$ is called the \textit{adhesion} of $st$, and is denoted by $\sigma(st)$.

It is a folklore result that a graph $H$ is chordal if and only if it has a tree decomposition whose bags are exactly the maximal cliques of $H$ (meaning, in particular, that the number of nodes of the tree is equal to the number of maximal cliques of $H$). Such a tree decomposition is called a \emph{clique tree} of $H$; note that while the set of bags of a clique tree is defined uniquely, the actual tree part of the tree decomposition is not necessarily unique. For example, if $H = K_{1,s}$, then there are $s$ maximal cliques (corresponding to edges of $H$), but they can be arranged into a tree decomposition in essentially an arbitrary manner.
We also remark that a chordal graph on $n$ vertices has at most $n$ maximal cliques, and hence its clique tree has at most $n$ nodes.

We will need some additional facts about clique trees of minimal chordal completions. Let $G$ be a graph. Given a set $S \subseteq V(G)$, a \textit{full component} of $S$ is a component $A$ of $G-S$ such that $N(A) = S$. A \textit{minimal separator} of $G$ is then a set $S \subseteq V(G)$ which has at least two full components.

The next two lemmas were proven in~\cite{p6FreeMaxInd22} using the toolbox from~\cite{BouchitteT01}. The first one shows how to obtain minimal separators from adhesions.

\begin{lemma}[\citeReference{Proposition~2.7}{p6FreeMaxInd22}]
\label{lem:cliqueTreesAdhesion}
Let $G$ be a graph, $G+F$ be a minimal chordal completion of $G$, and $(T, \beta)$ be a clique tree of $G+F$. Then for each edge $st \in E(T)$, the adhesion $\sigma(st)$ is a minimal separator of $G$, and it has full components $A$ and $B$ such that $\beta(s)\setminus \sigma(st)\subseteq A$ and $\beta(t)\setminus \sigma(st)\subseteq B$.
\end{lemma}

Notice that the full component $A$ which satisfies Lemma~\ref{lem:cliqueTreesAdhesion} is unique given the vertex $s$ and the edge $st \in E(T)$. (This uses the fact that $\beta(s)\setminus \sigma(st)$ is non-empty, which holds since $\beta(s)$ and $\beta(t)$ are distinct maximal cliques of $G+F$.) When the graph, chordal completion, and clique tree are clear from context, we call $A$ \textit{the full component of $\sigma(st)$ on the $s$-side}.

Lemma~\ref{lem:cliqueTreesAdhesion} immediately implies also the following.
\begin{lemma}\label{lem:adh-to-comp}
Let $G$ be a graph, $G+F$ be a minimal chordal completion of $G$, and $(T, \beta)$ be a clique tree of $G+F$.
Then for every $st \in E(T)$, there exists a connected component $D$ of $G-\beta(t)$ such that $N(D) = \sigma(st)$
and $D \subseteq \bigcup_{t' \in V(T_s)} \beta(t')$ where $T_s$ is the component of $T-\{t\}$ that contains $s$.
\end{lemma}
\begin{proof}
Use Lemma~\ref{lem:cliqueTreesAdhesion} and take $D$ to be the full component of $\sigma(st)$ on the $s$-side.
\end{proof}

The next lemma shows how to obtain minimal separators from PMCs.

\begin{lemma}[\citeReference{Proposition~2.10}{p6FreeMaxInd22}]
\label{lem:PMCComponents}
Let $G$ be a graph, $\Omega$ be a PMC of $G$, and $D$ be a component of $G-\Omega$. Then $N(D)$ is a minimal separator of $G$, and it has a full component $D^{\Omega}\neq D$ which contains $\Omega \setminus N(D)$.
\end{lemma}

We will also need the following well-known facts about chordal completions.
\begin{lemma}\label{lem:cliques-stay}
Let $G$ be a graph and $G+F$ be a minimal chordal completion of $G$.
Let $S \subseteq V(G)$ be such that $(G+F)[S]$ is a clique.
Then $F$ contains no edges between different connected components of $G-S$.
\end{lemma}
\begin{proof}
Let $\mathcal{D}$ be the family of connected components of $G-S$.
For every $D \in \mathcal{D}$, $F \cap \binom{N[D]}{2}$ is a chordal completion of $G[N[D]]$ that turns $N(D)$ into a clique.
Since $(G+F)[S]$ is a clique, $(F \cap \binom{S}{2}) \cup \bigcup_{D \in \mathcal{D}} F \cap \binom{N[D]}{2}$ is a chordal completion of $G$.
The claim follows by the minimality of $G+F$.
\end{proof}
\begin{lemma}\label{lem:comp-to-adh}
Let $G$ be a graph, $G+F$ be a minimal chordal completion of $G$, and $(T, \beta)$ be a clique tree of $G+F$.
Let $S$ be a minimal separator of $G$ such that $(G+F)[S]$ is a clique and let $A$ and $B$ be two full sides of $S$.
Then there exists an edge $t_At_B \in E(T)$ such that $\sigma(t_At_B) = S$,
     $A \subseteq \bigcup_{t \in V(T_A)} \beta(t)$, $B \subseteq \bigcup_{t \in V(T_B)} \beta(t)$,
     where $T_A$ and $T_B$ are the components of $T-\{t_At_B\}$ that contain $t_A$ and $t_B$, respectively.
\end{lemma}
\begin{proof}
Let $Z_A = \{t \in V(T)~|~A \cap \beta(t) \neq \emptyset\}$ and similarly define $Z_B$.
Since $A$ and $B$ are connected, $Z_A$ and $Z_B$ are connected in $T$.
By Lemma~\ref{lem:cliques-stay}, $Z_A \cap Z_B = \emptyset$.
Let $Q$ be the unique path in $T$ that has one endpoint in $Z_A$, the second endpoint in $Z_B$, and all internal vertices outside $Z_A \cup Z_B$. Note that the length of $Q$ is at least one.
Let $q_A$ and $q_B$ be the endpoints of $Q$ in $Z_A$ and $Z_B$, respectively.

Since $(T,\beta)$ is a tree decomposition of $G$, $N_G[A] \subseteq \bigcup_{t \in Z_A} \beta(t)$.
Since $(T,\beta)$ is a clique tree of the chordal graph $G+F$, we have $N_{G+F}[A] \supseteq \bigcup_{t \in Z_A} \beta(t)$.
Lemma~\ref{lem:cliques-stay} implies that $N_G[A] = N_{G+F}[A]$.
Thus $N_G[A] = N_{G+F}[A] =  \bigcup_{t \in Z_A} \beta(t)$ and, similarly, $N_G[B] = N_{G+F}[B] = \bigcup_{t \in Z_B} \beta(t)$.

Since $S = N_G[A] \cap N_G[B]$, $S \subseteq \beta(s)$ for every $s \in V(Q)$.
By the definition of $Z_A$, we have $\beta(s) \cap N_G[A] \subseteq S$ for every $s \in V(Q) \setminus \{q_A\}$.
Hence, if $q$ is the unique neighbor of $q_A$ on $Q$, then $\sigma(qq_A) = S$.
The lemma follows with $t_A = q_A$ and $t_B = q$.
\end{proof}

\subsection{Aligning chordal completions and treedepth structures}

Throughout the paper we will try to find a maximal induced subgraph with treedepth at most $d$. We will do so by considering a fixed elimination forest of this induced subgraph, as well as a chordal completion which ``aligns with'' the elimination forest. We now formalize these ideas.

Let $G$ be a graph and $d$ be a positive integer. A \emph{treedepth-$d$ structure in $G$} is a rooted forest $\T$ of height at most $d$ such that $V(\T)$ is a subset of $V(G)$ and $\T$ is an elimination forest of the subgraph of $G$ induced by $V(\T)$. We sometimes write $\T$ instead of $V(\T)$ when it is clear that we are working with a set of vertices; in particular, if $X$ is a set of vertices of $G$, then we write $X \cap \T$ instead of $X \cap V(\T)$. We say that $\T$ is \emph{maximal} if there is no treedepth-$d$ structure $\T'$ in $G$ such that $\T$ is a proper induced subgraph of $\T'$ and every root of $\T$ is a root of $\T'$.

Note that if $H$ is a maximal induced subgraph of $G$ of treedepth at most $d$, and $\T$ is a height-$d$ elimination forest of that subgraph, then $\T$ is a maximal treedepth-$d$ structure in $G$. Consequently, in the context of $(\td \leq d, \phi)$-\textsc{MWIS}, we can consider $\Sol$ being in fact a maximal set inducing a subgraph of treedepth at most $d$ in $G$: if $(\Sol, X)$ is an actual solution, then there exists a maximal treedepth-$d$ structure $\Sol'$ that is a superset of $\Sol$ and quantification over $\Sol$ can be implemented inside $\phi$.
(This step is formally explained in Section~\ref{sec:dynProg}.)
Thus, most of the structural results in this work consider the set of all maximal treedepth-$d$-structures, which are more detailed versions of maximal sets inducing a subgraph of treedepth at most $d$.

We conclude this section by discussing ``aligned'' chordal completions and by proving some basic lemmas about them. Let $G$ be a graph, $d$ be a positive integer, and $\T$ be a treedepth-$d$ structure in $G$. We say that a chordal completion $G+F$ is \textit{$\T$-aligned} if $F$ does not contain any pair $uv$ so that \begin{enumerate}
    \item $u$ or $v$ is a depth-$d$ vertex of $\T$, or
    \item $u$ and $v$ are vertices of $\T$ which are $\T$-incomparable.
\end{enumerate}
\noindent The second condition equivalently says that $\T$ is a treedepth-$d$ structure in $G+F$. First we show that there is always a $\T$-aligned minimal chordal completion.

\begin{lemma}
\label{lem:alignedChordalCompl}
For any positive integer $d$, graph $G$, and treedepth-$d$ structure $\T$ in $G$, there exists a minimal chordal completion of $G$ that is $\T$-aligned.
\end{lemma}
\begin{proof}
    Let $F$ denote the set of all non-edges $uv$ of $G$ which are not incident to a depth-$d$ vertex of $\T$, and are not between two vertices of $\T$ which are $\T$-incomparable. It suffices to prove that $G+F$ is chordal, since any chordal subgraph of $G+F$ is $\T$-aligned.

    Going for a contradiction, suppose that $C$ is a hole of $G+F$. As $(G+F)-\T$ is a clique, there is a vertex in $C \cap \T$; choose one, say $u$, which has maximum depth in $\T$. Consider the two neighbors of $u$ in $C$; they are either outside of $\T$ or ancestors of $u$ in $\T$. However, the set of all such vertices forms a clique in $G+F$, which contradicts the fact that $C$ has length at least~$4$.
\end{proof}

Throughout the paper we consider PMCs and minimal separators which might come from an aligned chordal completion. So, to state these definitions, let $G$ be a graph, $d$ be a positive integer, and $\T$ be a treedepth-$d$ structure in $G$. A PMC $\Omega$ is \textit{$\T$-avoiding} if it is a maximal clique of a minimal chordal completion that is $\T$-aligned, and it does not contain any depth-$d$ vertex of $\T$. We deal with the case that $\Omega$ does contain a depth-$d$ vertex separately, in the next lemma. Finally, a minimal separator $S$ of $G$ is \emph{$\T$-avoiding} if $S \cap \T$ is contained in a vertical path of $\T$ and has no depth-$d$ vertex. (So these are the separators that can come from $\T$-avoiding PMCs.)

For a fixed treedepth-$d$ structure $\T$, a set $\widetilde{Y} \subseteq V(G)$ is a \emph{container} for a set $Y \subseteq V(G)$ if $Y \subseteq \widetilde{Y}$ and $\widetilde{Y} \cap \T = Y \cap \T$, that is, $\widetilde{Y} \setminus Y$ is disjoint from $\T$.

\begin{lemma}
\label{lem:leafyPMC}
For each positive integer $d$, there is a polynomial-time algorithm which takes in a graph $G$ and returns a collection $\mathcal{L}\subseteq 2^{V(G)}$ such that for any maximal treedepth-$d$ structure $\T$ in $G$, any $\T$-aligned minimal chordal completion $G+F$ of $G$, and any maximal clique $\Omega$ of $G+F$ which contains a depth-$d$ vertex of $\T$, $\mathcal{L}$ contains a set $\widetilde\Omega$ that is a container for $\Omega$, i.e., $\Omega\subseteq \widetilde\Omega$ and $\widetilde\Omega\cap \T = \Omega \cap \T$.
\end{lemma}
\begin{proof}
    We guess the vertex $v \in \Omega$ which is a depth-$d$ vertex of $\T$ (this vertex is unique). Thus $v$ is adjacent in $G$ to every other vertex of $\Omega$, because $\Omega$ is a clique in a $\T$-aligned chordal completion. Moreover, $v$ has at most $d-1$ neighbors in $\T$. We then guess the set $X$ of all neighbors of $v$ which are in $\T$ but are not in $\Omega$. Finally, for all guesses of $v$ and $X$, we add the set $N[v]\setminus X$ to $\mathcal{L}$. This collection $\mathcal{L}$ is as desired.
\end{proof}

The final lemma is how we will use the maximality of a treedepth-$d$ structure.

\begin{lemma}
\label{lem:PMCmaximality}
Let $G$ be a graph, $d$ be a positive integer, and $\T$ be a maximal treedepth-$d$ structure in $G$. Then for any $\T$-avoiding potential maximal clique $\Omega$ of $G$, each vertex in $\Omega\setminus\T$ has a neighbor in $\T\setminus\Omega$.
\end{lemma}
\begin{proof}
    Recall that the set $\Omega \cap \T$ is contained in a vertical path of $\T$.
    Moreover, since $\Omega$ is $\T$-avoiding, $\Omega\cap \T$ does not contain any depth-$d$ vertex of $\T$.
    For contradiction, suppose there is a vertex $u \in \Omega \setminus \T$ that has no neighbor in $\T\setminus\Omega$
    Then, another treedepth-$d$ structure $\T'$ in $G$ can be obtained from $\T$ by adding $u$.
    This contradicts the maximality of $\T$.
\end{proof}

\section{Dynamic programming}
\label{sec:dynProg}

Now let us recall the definition of the main object of study in this paper.

\begin{definition}\label{def:carvers}
Let $G$ be a graph and $d$ and $k$ be positive integers.
A family $\carvers \subseteq 2^{V(G)}$ is a \emph{tree-depth-$d$ carver family} of \emph{defect $k$} in $G$
if for every treedepth-$d$ structure $\T$ in $G$,

there exists a tree decomposition $(T,\beta)$ of $G$
such that for each $t\in V(T)$ there exists $C \in \carvers$
such that
\begin{enumerate}
\item $C \cap \T$ contains $\beta(t)\cap \T$ and has size at most $k$, and
\item each component of $G-C$ is contained in $\beta(t) \cup \bigcup_{s \in T'}\beta(s)$ for some component $T'$ of $T-\{t\}$.
\end{enumerate}
Such a set $C$ as above is called a \emph{$(\T, (T,\beta))$-carver} for $\beta(t)$ of defect $k$; it might not be unique. We use this definition independently of that of carver families.
\end{definition}

It is important to compare the notion of a carver family with the notion of \emph{containers} of~\cite{containerMethod21}.
There, instead of the properties above, we mandate that $|\beta(t) \cap \T|\leq k$, that $C \cap \T = \beta(t) \cap \T$, and that $\beta(t) \subseteq C$
(so that, in particular, the choice of the tree $T$ in the tree decomposition $(T,\beta)$ is irrelevant for the definition).
These requirements imply that parts \textit{(i)} and \textit{(ii)}
 of the definition of a carver family hold; for the second part, observe that
if $\beta(t) \subseteq C$, then any component of $G-C$ is contained in a component of $G-\beta(t)$ which, by the properties of a tree decomposition, lies in the union of bags of a single component of $T-\{t\}$. The main difference is that in the notion of a carver, we actually allow a carver $C$ to miss some vertices of $\Omega$, as long as this does not result in ``gluing'' connected components of $G-\Omega$ residing in different subtrees of $T-\{t\}$ within the same connected component of $G-C$.

The main result of this section is that a tree-depth-$d$ carver family of small defect in $G$ is enough to design a dynamic programming
routine that solves the $(\td \leq d, \phi)$-\textsc{MWIS} problem on $G$.
We state it in the slightly more general form that allows $\Sigma^0$-annotated graphs for fixed finite set $\Sigma^0$.
\begin{theorem}\label{thm:dp}
For any positive integers $d$ and $k$, any finite set $\Sigma^0$ and any \cmsotwo formula $\phi$
over the signature of $\Sigma^0$-annotated graphs, there exists an algorithm that, given a vertex-weighted $\Sigma^0$-annotated graph $G$
and a tree-depth-$d$ carver family $\carvers \subseteq 2^{V(G)}$ of defect $k$ in $G$, runs in time polynomial in the input size
and either outputs an optimal solution to the $(\td \leq d, \phi)$-\textsc{MWIS} problem on~$G$, or determines that no feasible solution exists.
\end{theorem}

The remainder of this section is devoted to the proof of Theorem~\ref{thm:dp}.

\subsection{Canonizing and extending partial solutions}
Fix an integer $d$ and let $G$ be a graph. A \emph{partial solution} in $G$ is any tuple $(\T,X,\Sol)$ such that $\T$ is a tree-depth-$d$ structure in $G$
and $X \subseteq \Sol \subseteq V(\T)$.

Very roughly, the dynamic programming routine will have a table with some entries for each partial solution $(\T,X,\Sol)$ such that
$\T$ has at most $k$ leaves.
Each of these entries will contain a partial solution $(\T',X',\Sol')$ which ``extends'' $(\T,X,\Sol)$ into a specified part of the graph. We will update this partial solution $(\T',X',\Sol')$ when we find a ``better'' one. Sometimes this choice is arbitrary. So, in order to have more control over arbitrary choices, we now introduce a consistent tie-breaking scheme over partial solutions. More formally, we introduce a quasi-order $\preceq$ over partial solutions.

First, fix an arbitrary enumeration of $V(G)$ as $v_1,v_2,\ldots,v_{|V(G)|}$.
Second, define a total order $\preceq_1$ on subsets of $V(G)$ as follows: $X \prec_1 Y$ if $|X| > |Y|$
or if $|X| = |Y|$ and we have $v_i \in X$, where $i$ is the minimum integer such that $v_i \in X \triangle Y$ (i.e., we use the lexicographic order).
Third, define a quasi-order $\preceq_2$ on tree-depth-$d$ structures in $G$ as follows. Given a tree-depth-$d$ structure $\T$, associate to $\T$
the following tuple of $d+1$ subsets of $V(G)$:
\begin{itemize}
    \item $V(\T)$,
    \item the set of all vertices of depth $1$ in $\T$ (i.e., the roots),
    \item the set of all vertices of depth $2$ in $\T$,\\
    \ldots
    \item the set of all vertices of depth $d$ in $\T$.
\end{itemize}
When comparing two tree-depth-$d$ structures with $\preceq_2$, we compare with $\preceq_1$ the first sets in the above tuple that differ.

For two distinct tree-depth-$d$ structures $\T$ and $\T'$, we have $\T \preceq_2 \T'$ or $\T' \preceq_2 \T$.
However, we may have both $\T \preceq_2 \T'$ or $\T' \preceq_2 \T$ (i.e., it is possible that, for two different tree-depth-$d$ structures $\T$ and $\T'$, we have $V(\T) = V(\T')$ and every vertex of $V(\T)$ has the same depth in $\T$ and in $\T'$).
So $\preceq_2$ is only a quasi-order on the set of all tree-depth-$d$ structures in $G$; it partitions tree-depth-$d$ structures into equivalence classes, and between the equivalence classes it is a total order.

In order to avoid this problem, we will show that we can convert any tree-depth-$d$ structure into one that is ``neat'', and that $\preceq_2$ is a total order on ``neat'' tree-depth-$d$ structures. Formally, a tree-depth-$d$ structure $\T$ of $G$ is \emph{neat} if
for any non-root node $v$ of $\T$, the graph $G$ has at least one edge joining the parent of $v$ in $\T$ with a descendant of $v$ in $\T$ (possibly $v$ itself). One can easily see that this is equivalent to the following condition: for every node $v$ of $\T$, the subgraph of $G$ induced by the descendants of $v$ (including $v$) is connected.

The following lemma is standard when working with elimination forests: any tree-depth-$d$ structure can be adjusted to a neat one without increasing the depth.

\begin{lemma}\label{lem:neat-td}
Given a graph $G$ and a tree-depth-$d$ structure $\T$ of $G$, one can in polynomial time compute a neat tree-depth-$d$ structure $\T'$ of $G$ such that $V(\T')=V(\T)$ and for each $v\in V(\T)$, the depth of $v$ in $\T'$ is at most the depth of $v$ in~$\T$.
\end{lemma}
\begin{proof}
While possible, perform the following improvement step. If $v \in V(\T)$ is such that $v$ is not a root of $\T$, but has a parent $u$,
and the subtree $\T_v$ of $\T$ rooted at $v$ does not contain a vertex of $N_G(u)$, then
reattach $\T_v$ to the parent of $u$ if $u$ is not a root or detach it as a separate component of $\T$ otherwise. It is immediate that
the new rooted forest is also a tree-depth-$d$ structure of $G$ and, furthermore, that the depths of the elements of $\T_v$ decreased by one.
This in particular implies that there will be at most $|V(G)|^2$ improvement steps. Each of them can be executed in polynomial time. Once no more improvement steps are possible, the resulting tree-depth-$d$ structure is neat, as desired.
\end{proof}

Next, we show that $\preceq_2$ is a total order on neat tree-depth-$d$ structures. In fact we show something slightly stronger: that each neat tree-depth-$d$ structure is in a singleton equivalence class.

\begin{lemma}\label{lem:dp:neat-equal}
If $\T$ and $\T'$ are tree-depth-$d$ structures such that $\T$ is neat, $\T \preceq_2 \T'$, and $\T' \preceq_2 \T$, then $\T = \T'$.
\end{lemma}
\begin{proof}
Since $\T \preceq_2 \T'$ and $\T' \preceq_2 \T$, we have that $V(\T) = V(\T')$ and every vertex has the same depth in $\T$ and $\T'$.
We prove inductively on $i$ that the set of all vertices of depth at least $d-i$ induces the same forest in $\T$ and $\T'$. The base case of $i=0$ holds since the depth-$d$ vertices are an independent set in both $\T$ and $\T'$. For the inductive step, it suffices to show that each depth-$(d-i)$ vertex $v$ has the same parent in $\T$ and $\T'$. So let $u$ and $u'$ be the parent of $v$ in $\T$ and $\T'$, respectively.
From the inductive hypothesis, the subtrees of $\T$ and $\T'$ rooted at $v$ are equal.
Since $\T$ is neat, there is a vertex $w$ in this subtree that is adjacent to $u$ in $G$.
Since $\T'$ is an elimination forest, $u'$ and $w$ are comparable in $\T'$. Since $u'$ and $u$ have the same depth
in $\T$ and $\T'$, this is only possible if $u=u'$.
\end{proof}

Finally, given two partial solutions $(\T,X,\Sol)$ and $(\T',X',\Sol')$ in a graph $G$, we say that $(\T,X,\Sol) \preceq (\T',X',\Sol')$ if:
\begin{enumerate}
    \item the weight of $X$ is larger than the weight of $X'$, or
    \item the weights of $X$ and $X'$ are equal, but $X \prec_1 X'$;
    \item $X = X'$, but $\Sol \prec_1 \Sol'$;
    \item $X = X'$ and $\Sol = \Sol'$, but $\T \preceq_2 \T'$.
\end{enumerate}
We say that $(\T,X,\Sol)$ is \emph{better} than $(\T',X',\Sol')$
(or that $(\T',X',\Sol')$ is \emph{worse} than $(\T,X,\Sol)$) if $(\T,X,\Sol) \preceq (\T',X',\Sol')$ and some comparison above is strict.

Using this quasi-order, we can now look for a partial solution $(\T, X, \Sol)$ such that $\T$ is maximal and neat. This is based on the following observation.

\begin{lemma}\label{lem:dp:best-is-neat}
For any $X \subseteq \Sol \subseteq V(G)$ such that $G[\Sol]$ has tree-depth at most $d$, there exists a tree-depth-$d$ structure $\T$ such that $(\T, X, \Sol)$ is a partial solution. Moreover, if one chooses $\T$ so that $(\T, X, \Sol)$ is $\preceq$-minimal (among all choices of $\T$, for fixed $X$ and $\Sol$), then $\T$ is maximal and neat.
\end{lemma}
\begin{proof}
For the first claim, any depth-$d$ elimination forest of $G[\Sol]$ can serve as $\T$.
For the second claim, fix a $\preceq$-minimal partial solution $(\T,X,\Sol)$. Since the first comparison is on the sizes of $V(\T)$, we have that $\T$ is maximal.

Now, by Lemma~\ref{lem:neat-td}, there exists a neat tree-depth-$d$ structure $\T'$ such that $V(\T') = V(\T)$ and, for each $v \in V(\T)$, the depth of $v$ in $\T'$ is at most the depth of $v$ in $\T$. Thus $\T' \preceq_2 \T$. So, since $(\T',X,\Sol)$ is not better than $(\T,X,\Sol)$, we also have that $\T \preceq_2 \T'$. Since $\T'$ is neat, Lemma~\ref{lem:dp:neat-equal} says that $\T' = \T$. So $\T$ is neat, as desired.
\end{proof}

It is convenient to conclude this subsection by defining extensions of partial solutions. Roughly, an ``extension'' of a partial solution $(\T, X, \Sol)$ in a graph $G$ is any partial solution $(\T', X', \Sol')$ that can be obtained from $(\T, X, \Sol)$ by adding new vertices which are not ancestors of any node of $\T$. More formally, $(\T', X', \Sol')$ is an \emph{extension} of $(\T, X, \Sol)$ if $\T$ is an induced subgraph of $\T'$, every root of $\T$ is a root of $\T'$, and $X' \cap \T = X$ and $\Sol' \cap \T = \Sol$. We define extensions of tree-depth-$d$ structures analogously, omitting $X$ and $\Sol$.

We will use the following properties of extensions.

\begin{lemma}
\label{lem:extensionsTDS}
Let $G$ be a graph, $d$ be an integer, and $\T$ and $\T'$ be tree-depth-$d$ structures in $G$ such that $\T'$ is neat and extends $\T$. Then each connected component of $\T'- V(\T)$ is neat and induces a connected subgraph of $G$.
\end{lemma}
\begin{proof}Each connected component of $\T'- V(\T)$ is obtained by selecting a vertex $v \in V(\T')$ and then taking all descendants of $v$ in $\T'$ (including $v$ itself). Any such subtree of $\T'$ is neat. For the second part, we observe that any neat tree-depth-$d$ structure with just one root induces a connected subgraph of $G$.
\end{proof}

\subsection{Threshold automata}
\label{subsec:threshold}
Next, we introduce {\em{threshold automata}}, which capture through an abstract notion of a computation device, the idea of processing a labelled forest in a bottom-up manner using a dynamic programming procedure.
As we will comment on, the design of this automata model follows standard constructions that were developed in the 90s.

 We need to introduce some notation before stating the main definitions. For a finite alphabet $\Sigma$, a {\em{$\Sigma$-labelled forest}} is a rooted forest $F$ where every vertex $x\in V(F)$ is labelled with an element $\lbl(x)\in\Sigma$. Similarly, given an unlabelled rooted forest $F$, we call any function $\lbl$ from $V(F)$ to $\Sigma$ a \textit{$\Sigma$-labelling} of $F$.

We use the notation $\lmset\cdot\rmset$ for defining multisets. For a multiset $X$ and an integer $\tau\in \N$, let $X\wedge \tau$ be the multiset obtained from $X$ by the following operation: for every element $e$ whose multiplicity $k$ is larger than $2\tau$, we reduce its multiplicity to the unique integer in $\{\tau+1,\ldots,2\tau\}$ with the same residue as $k$ modulo $\tau$ (that is, we reduce it to $k - \tau \lfloor \frac{k - \tau - 1}{\tau} \rfloor$). This definition lets us track at the same time the residue modulo $\tau$ of the multiplicity as well as whether the multiplicity is greater than $\tau$ or not. For a finite set $Q$ and an integer $\tau\in \N$, we write $\Multi(Q,\tau)$ for the family of all multisets with elements from $Q$, where each element appears at most $2\tau$ times. Note that $|\Multi(Q,\tau)|=(2\tau+1)^{|Q|}$.

Informally, a threshold automaton is run bottom-up on a $\Sigma$-labelled forest $F$. As it runs, it assigns each vertex of $F$ a state from a finite set $Q$. The state of the next vertex $v \in V(F)$ depends only on $\lbl(v)$ and the ``reduced'' multiset $X \wedge \tau$, where $X$ denotes the multiset of the states of all children of $v$. The accepting condition is similarly determined by ``reducing'' the multiset of the states of the roots. The formal definition is as follows.

\begin{definition}
 A {\em{threshold automaton}} is a tuple $\Aa=(Q,\Sigma,\tau,\delta,C)$, where:
 \begin{itemize}[nosep]
  \item $Q$ is a finite set of states;
  \item $\Sigma$ is a finite alphabet;
  \item $\tau\in \N$ is a nonnegative integer called the {\em{threshold}};
  \item $\delta\colon \Sigma\times \Multi(Q,\tau)\to Q$ is the transition function; and
  \item $C\subseteq \Multi(Q,\tau)$ is the accepting condition.
 \end{itemize}
 For a $\Sigma$-labelled forest $F$, the {\em{run}} of $\Aa$ on $F$ is the unique labelling $\xi\colon V(F)\to Q$ satisfying the following property for each $x\in V(F)$:
 $$\xi(x)=\delta(\lbl(x),\lmset\xi(y)\colon y\textrm{ is a child of }x\rmset\wedge \tau).$$
 We say that $\Aa$ {\em{accepts}} $F$ if
 $$\lmset\xi(z)\colon z\textrm{ is a root of }F\rmset\wedge \tau \in C,$$
 where $\xi$ is the run of $\Aa$ on $F$.
\end{definition}

It turns out that threshold automata precisely characterize the expressive power of \cmso over labelled forests. Here, we consider the standard encoding of $\Sigma$-labelled forests as relational structures using one binary parent relation and $|\Sigma|$ unary relations selecting nodes with corresponding labels. Consequently, by \cmso over $\Sigma$-labelled forests we mean the logic in which
\begin{itemize}
    \item there are variables for single nodes and for node sets,
    \item in atomic formulas one can check equality, membership, modular counting predicates, parent relation, and labels of single nodes, and
    \item larger formulas can be obtained from atomic ones using standard boolean connectives, negation, and both universal and existential quantification over both sorts of variables.
\end{itemize}
The proof of the next statement is standard,
see for instance \cite[Theorem 5.3]{Courcelle90} for a proof in somewhat different terminology and~\cite[Section~7.6]{Libkin04} for the closely related settings of binary trees and ordered, unranked trees (the proof techniques immediately lift to our setting).
Hence, we only provide a sketch.

\begin{lemma}\label{lem:mso-threshold}
 Let $\Sigma$ be a finite alphabet. Then for every sentence $\varphi$ of \cmso over $\Sigma$-labelled forests, there exists a threshold automaton $\Aa$ with alphabet $\Sigma$ such that for any $\Sigma$-labelled forest $F$, we have $F\models \varphi$ if and only if $\Aa$ accepts $F$.
\end{lemma}
\begin{proof}[Sketch of proof]
Let the {\em{rank}} of $\varphi$ be the product of the quantifier rank of $\varphi$ (that is, the maximum number of nested quantifiers in~$\varphi$) and the least common multiple of all moduli featured in modular predicates present in $\varphi$.
It is well-known that there is only a finite number of pairwise non-equivalent \cmso sentences over $\Sigma$-labelled forests with rank at most $q$. Let then $\Tp^q$ be the set containing one such sentence from each equivalence class. Then $\Tp^q$ is finite, and we may assume that $\varphi\in \Tp^q$.

Consider a $\Sigma$-labelled forest $F$.
For a vertex $x\in V(F)$, let $F_x$ be the subtree of $F$ induced by $x$ and all of its descendants. The {\em{$q$-type}} of $F_x$ is the set of all sentences from $\Tp^q$ which are satisfied in $F_x$, that is,
$$\tp^q(F_x)\coloneqq \{\,\psi\in \Tp^q\ \mid\ F_x\models \psi\,\}.$$
A standard argument using Ehrenfeucht-Fra\"isse games shows that $\tp^q(F_x)$ is uniquely determined by $\lbl(x)$ and the multiset $\lmset\tp^q(F_y)\colon y\textrm{ is a child of }x\rmset\wedge q$. Similarly, the type $\tp^q(F)$, defined analogously as above, is uniquely determined by the multiset $\lmset\tp^q(F_r)\colon r\textrm{ is a root of }F\rmset\wedge q$. This means that we may define a threshold automaton $\Aa$ with state set $\Tp^q$ and threshold $q$ so that $\Aa$ accepts $F$ if and only if $\varphi\in \tp^q(F)$, which is equivalent to $F\models\varphi$.
\end{proof}

We would like to use Lemma~\ref{lem:mso-threshold} in order to verify that a given solution $(\Sol,X)$ to $(\td \leq d,\phi)$-\textsc{MWIS} indeed is such that $G[\Sol]$ satisfies $\phi(X)$. For this, our dynamic programming tables will be indexed not only by partial solutions of the form $(\T,X,\Sol)$, but also by guesses on ``partial evaluation'' of $\phi$ that occurs outside of $V(\T)$; or more formally, by an appropriate multiset of states of a threshold automaton associated with $\phi$. For this, we need to understand how to run threshold automata on treedepth-$d$ structures rather than just labelled forest. This will be done in a standard way: by labelling the forest underlying a treedepth-$d$ structure $\T$ so that the labels encode $\T$. This idea is formalized in the next definition.

\begin{definition} Let $d$ be an integer and $\Sigma^0, \Sigma$ be finite alphabets. Then a \emph{$(d,\Sigma^0,\Sigma)$-labeller} is a polynomial-time algorithm $\Lambda$ that, given an $\Sigma^0$-annotated graph $(G, \lambda^0)$ with a partial solution $(\T, X, \Sol)$ for the $(\td \leq d,\phi)$-\textsc{MWIS} problem, computes a $\Sigma$-labelling of $\T$ such that for every $v \in V(\T)$, the label of $v$ depends only on:
\begin{itemize}[nosep]
    \item the label $\lambda^0(v)$,
    \item the integer $h \in \{1,2,\ldots, d\}$ such that $v$ has depth $h$ in $\T$,
    \item the set of all indices $i \in \{1,2,\ldots, h-1\}$ such that $v$ is adjacent, in $G$, to the unique ancestor of $v$ in $\T$ with depth $i$, and
    \item which of the sets $X$ and $\Sol$ contain $v$.
\end{itemize}
That is, if we run $\Lambda$ again on another $G'$ and $(\T', X', \Sol')$, then any vertex $v' \in V(\T')$ with the same properties from above as $v$ is labelled the same as $v$.
\end{definition}

When $\Lambda$ and $(G, \lambda^0)$ are clear from context, we write $\lbl_{(\T, X, \Sol)}$ for the $\Sigma$-labelling on $\T$ which is returned by running $\Lambda$ on $(G, \lambda^0)$ and $(\T, X, \Sol)$. A key aspect of this definition is that, if $(\T', X', \Sol')$ is a partial solution which extends $(\T, X, \Sol)$, then each vertex $v \in V(\T)$ satisfies $\lbl_{(\T, X, \Sol)}(v) = \lbl_{(\T', X', \Sol')}(v)$.

We are now ready to state the main proposition of this subsection.

\begin{restatable}{proposition}{propThresh}
\label{prop:mainThreshold}
Given a fixed $(\td \leq d,\phi)$-\textsc{MWIS} problem, where $\phi$ is a \cmsotwo formula over the signature of $\Sigma^0$-annotated graphs, there exists a finite alphabet $\Sigma$, a $(d,\Sigma^0,\Sigma)$-labeller $\Lambda$, and a threshold automaton $\Aa$ with alphabet $\Sigma$ such that for any partial solution $(\T, X, \Sol)$ in any $\Sigma^0$-annotated graph $G$, we have that $(\Sol, X)$ is feasible for $(\td \leq d,\phi)$-\textsc{MWIS} in $G$ if and only if $\Aa$ accepts the $\Sigma$-labelled forest obtained from $\T$ by equipping it with $\lbl_{(\T, X, \Sol)}$.
\end{restatable}

We first prove several lemmas, and then we prove Proposition~\ref{prop:mainThreshold} by combining them. It is straightforward to rewrite formulas to obtain the following lemma.

\begin{lemma}\label{lem:dp:rewrite-phi}
For any $d \in \N$, finite set $\Sigma^0$, and \cmsotwo formula $\phi$ over the signature of $\Sigma^0$-annotated graphs with one free vertex set variable, there exists a \cmsotwo formula $\varphi'$ over the signature of $\Sigma^0$-annotated graphs with two free vertex set variables such that for any partial solution $(\T,X,\Sol)$ in any $\Sigma^0$-annotated graph $G$, we have that $(\Sol,X)$ is feasible for $(\td \leq d,\phi)$-\textsc{MWIS} in $G$ if and only if $G[\T] \models \varphi'(X, \Sol)$.
\end{lemma}

We now show how to obtain an alphabet $\Sigma$ and a $(d, \Sigma^0, \Sigma)$-labeller which lets us get rid of the graph entirely. That is, we will reduce the given sentence to a sentence in \cmso over a $\Sigma$-labelled forest.

\begin{lemma}\label{lem:mso}
 For any $d \in \N$ and finite set $\Sigma^0$, there exist a finite alphabet $\Sigma$ and a $(d,\Sigma^0,\Sigma)$-labeller $\Lambda$ so that the following holds. For any formula $\varphi$ of \cmsotwo over $\Sigma^0$-annotated graphs with two free vertex set variables, there exists a sentence $\wh{\varphi}$ of \cmso over $\Sigma$-labelled forests such that for any partial solution $(\T, X, \Sol)$ in any $\Sigma^0$-annotated graph $G$,
 $$G[\T] \models \varphi(X, \Sol)\qquad\textrm{if and only if}\qquad \wh{\T} \models \wh{\varphi},$$
 where $\wh{\T}$ is the $\Sigma$-labelled forest obtained from $\T$ by equipping it with $\lbl_{(\T, X, \Sol)}$.
\end{lemma}

\begin{proof}
 We let $\Sigma=\Sigma^0 \times \{1,\ldots,d\}\times \{0,1\}^{\{1,\ldots,d\}} \times \{0,1\}^{2}$, where the third coordinate is treated as a function from $\{1,\ldots,d\}$ to $\{0,1\}$; note that $|\Sigma|=|\Sigma^0| \cdot d\cdot 2^{d+2}$.

 Consider a graph $G$ and a partial solution $(\T, X, \Sol)$ in $G$. We now define the $(d,\Sigma^0,\Sigma)$-labeller $\Lambda$. Consider any $x\in V(\T)$. Let $h$ be the depth of $x$ in $\T$. Let $f$ be the function from $\{1,\ldots,d\}$ to $\{0,1\}$ defined as follows: for $i\geq h$ we set $f(i)=0$, and for $i<h$ we set $f(i)=1$ if and only if $x$ is adjacent to the unique ancestor of $x$ in $\T$ that has depth $i$.
 Let $\mathbbm{1}_X$ and $\mathbbm{1}_\Sol$ be equal the value $1$ if $v$ is in $X$ or $\Sol$, respectively, and $0$ otherwise. Then we set
 $$\lbl(x)\coloneqq (\lambda^0(x),h,f,\mathbbm{1}_X, \mathbbm{1}_\Sol).$$
 Note that this labelling function can be computed from $G[\T]$ and $(\T, X, \Sol)$ in polynomial time. Moreover, this algorithm is a $(d,\Sigma^0,\Sigma)$-labeller. Let $\wh{\T}$ denote the $\Sigma$-labelled forest obtained from $\T$ by equipping it with this labelling.

 We now apply the following syntactic transformation to $\varphi$ in order to obtain a sentence $\wh{\varphi}$ of \cmso over $\Sigma$-labelled forests.
 \begin{itemize}
  \item For every quantification over an edge $e$, replace it with a quantification over the pair $x,y$ of its endpoints, followed by a check that $x$ and $y$ are indeed adjacent. Since the depth of $\T$ is at most~$d$, which is a constant, this check can be performed using a first-order formula as follows: verify that $x$ and $y$ are in the ancestor-descendant relation in $\wh{\T}$, retrieve the depth of $x$ and $y$ in $\wh{\T}$ from their labels, and check that the label of the deeper of those two nodes contains information that the shallower one is adjacent to it.
  \item Replace each atom expressing that a vertex $z$ is incident to an edge $e$ by a disjunction checking that $z$ is one of the endpoints of~$e$.
  \item For every quantification over an edge set, say $\exists\,Y$, replace it with quantification of the form $\exists\,Y_1\ \exists\,Y_2\ \ldots\ \exists\,Y_{d-1}$, where $Y_i$ is interpreted as the set of all the deeper endpoints of those edges from~$Y$ whose shallower endpoint has depth $i$. This quantification is followed by checking that for each $x\in Y_i$, indeed $x$ is adjacent to its unique ancestor at depth $i$; this information is encoded in the label of $x$.
  \item Replace each atom $e\in Y$, where $e$ is an edge variable and $Y$ is an edge set variable, with a disjunction over $i\in \{1,\ldots,d\}$ of the following checks: denoting the endpoints of $e$ by $x$ and $y$, either $x$ is at depth $i$ and $y\in Y_i$, or vice~versa.
  \item Replace each check $\lambda^0(x) = \sigma$, $x \in X$, or $x \in \Sol$ with the corresponding check of the first, fourth, or fifth coordinate of the label of $x$.
 \end{itemize}
 It is straightforward to see that the sentence $\wh{\varphi}$ obtained in this manner satisfies the desired property. This completes the proof of Lemma~\ref{lem:mso}.
\end{proof}

We complete this section by proving Proposition~\ref{prop:mainThreshold}, which is restated below for convenience.

\propThresh*
\begin{proof}
Fix $d$, $\Sigma^0$, and $\phi$.
By Lemma~\ref{lem:dp:rewrite-phi}, there exists a \cmsotwo formula $\varphi'$ over the signature of $\Sigma^0$-annotated graphs with two free vertex set variables such that for any partial solution $(\T,X,\Sol)$ in any graph $G$, we have that $(\Sol,X)$ is feasible for $(\td \leq d,\phi)$-\textsc{MWIS} in $G$ if and only if $G[\T] \models \varphi(X, \Sol)$.
By Lemma~\ref{lem:mso}, there exist a finite alphabet $\Sigma$, a $(d,\Sigma^0,\Sigma)$-labeller $\Lambda$, and a sentence $\wh{\varphi}$ of \cmso over $\Sigma$-labelled forests such that for any partial solution $(\T, X, \Sol)$ in any $\Sigma^0$-annotated graph $G$,
 \[G[\T] \models \varphi'(X, \Sol)\qquad\textrm{if and only if}\qquad \wh{\T} \models \wh{\varphi},\]
 where $\wh{\T}$ is the $\Sigma$-labelled forest obtained from $\T$ by equipping it with $\lbl_{(\T, X, \Sol)}$.

 Finally, by Lemma~\ref{lem:mso-threshold}, there exists a threshold automaton $\Aa$ with alphabet $\Sigma$ such that for any $\Sigma$-labelled forest $F$, we have $F\models \wh{\varphi}$ if and only if $\Aa$ accepts $F$. Proposition~\ref{prop:mainThreshold} follows.
\end{proof}

\subsection{The algorithm}
Fix an integer $d$, a finite set $\Sigma^0$, and a \cmsotwo formula $\phi$ over the signature of $\Sigma^0$-annotated graphs. By Proposition~\ref{prop:mainThreshold}, there exists a finite alphabet $\Sigma$, a $(d,\Sigma^0,\Sigma)$-labeller $\Lambda$, and a threshold automaton $\Aa= (Q_\Aa,\Sigma, \tau_\Aa, \delta_\Aa, C_\Aa)$ such that for any partial solution $(\T, X, \Sol)$ in any $\Sigma^0$-annotated graph $(G, \lambda^0)$, we have that $(\Sol, X)$ is feasible for $(\td \leq d,\phi)$-\textsc{MWIS} in $(G, \lambda^0)$ if and only if $\Aa$ accepts the $\Sigma$-labelled forest obtained from $\T$ by equipping it with the labelling $\lbl_{(\T, X, \Sol)}$. The algorithm will make use of $\Sigma$, $\Lambda$, and $\Aa$.

For convenience, we say that a \emph{multistate assignment} of a rooted forest $F$ is any function $\xi:\{\emptyset\} \cup V(F) \to \Multi(Q_\Aa,\tau_\Aa)$. Consider a multistate assignment $\xi$ of a tree-depth-$d$ structure $\T$. Essentially, we use $\xi$ to specify the desired behavior of an extension of a partial solution $(\T, X, \Sol)$. In order to combine two extensions, sometimes we need to combine two multistate assignments $\xi_1$ and $\xi_2$ of a rooted forest $F$. So we write $\xi_1 \cup \xi_2$ for the multistate assignment of $F$ defined by setting $(\xi_1 \cup \xi_2)(v) \coloneqq (\xi_1(v) \cup \xi_2(v)) \wedge \tau_\Aa$ for each $v \in \{\emptyset\} \cup V(F)$.

Let $\C \subseteq 2^{V(G)}$ be a tree-depth-$d$ carver family of defect $k$ in $G$. A \emph{template} is a tuple $\sigma = (\T,X,\Sol, C, D, \xi)$ such that \begin{enumerate}
    \item $(\T, X, \Sol)$ is a partial solution in $G$,
    \item $C \in \C$,
    \item $D$ is a subset of $V(G)$ which is a union of zero or more components of $G-C$, and
    \item $\xi$ is a multistate assignment of $\T$.
\end{enumerate}
\noindent We say that $\sigma$ is \emph{simple} if $\T$ has at most $k$ leaves and $D$ is a component of $G-C$. A
\emph{(simple) pre-template} is a tuple $\alpha = (\T, X, \Sol, C)$ as in the definition of a (simple) template, except with $D$ and $\xi$ omitted. We say that a template $\sigma = (\T,X,\Sol, C, D, \xi)$ is \emph{over} the pre-template $(\T, X, \Sol, C)$.

The dynamic programming algorithm stores a table $M$ that has an entry $M[\sigma]$ for each simple template $\sigma$. We observe that the table has $\mathcal{O}(|\C|\cdot |V(G)|^{dk+1})$ entries, where the constant hidden in the big-$\mathcal{O}$ notation depends on $d$, $|\Sigma^0|$, $k$, and $\phi$. We initiate the value of each entry $M[\sigma]$ to a symbol $\bot$. As the algorithm proceeds, $M[\sigma]$ will be updated to contain a partial solution $(\T', X', \Sol')$ which is a ``valid extension'' (defined formally in the next paragraph) of $\sigma$. We only update $M[\sigma]$ when we discover a new valid extension better than the old one according to $\preceq$; we use the convention that every valid extension is better than $\bot$.

Now, let $\sigma = (\T,X,\Sol, C, D, \xi)$ be a template (which may or may not be simple). Then a \emph{valid extension} of $\sigma$ is any extension $(\T',X',\Sol')$ of $(\T, X, \Sol)$ such that $V(\T')\setminus V(\T) \subseteq D$ and, if $\zeta:V(\T') \to Q_{\mathcal{A}}$ denotes the run of $\mathcal{A}$ on the $\Sigma$-labelled forest obtained from $\T'$ by equipping it with $\lbl_{(\T', X', \Sol')}$, then
        \[ \lmset \zeta(z)~|~z\textrm{ is a root of }\T'\textrm{ but not of }\T \rmset \wedge \tau_\Aa = \xi(\emptyset), \]
         and, for every $v \in V(\T)$,
        \[ \lmset \zeta(z)~|~z \textrm{ is a child of }v \textrm{ in } \T' \textrm{ but not in }\T\rmset \wedge \tau_\Aa = \xi(v). \]
\noindent Note that if $z$ is a child of $v$ in $\T'$ but not in $\T$, then $z\notin V(\T)$ since, if it was, then it would have the same parent in $\T'$ and $\T$ by the definition of extensions.

The following observation about combining extensions is the crucial building block of the algorithm. To state the lemma, we need to know when we can combine two tree-depth-$d$ structures $\T$ and $\T'$ in a graph $G$. So we say that $\T$ and $\T'$ are \emph{compatible} if the sets $V(\T)\setminus V(\T')$ and $V(\T')\setminus V(\T)$ are anticomplete in $G$, and each vertex in $V(\T) \cap V(\T')$ has the same parent in $\T$ and $\T'$. (We think of the empty set as being the parent of a root; so in particular this means that every ancestor of a vertex in $V(\T) \cap V(\T')$ is also in $V(\T) \cap V(\T')$.) If $\T$ and $\T'$ are compatible, then there is a unique tree-depth-$d$ structure, which we denote by $\T \cup \T'$, such that:\begin{enumerate}
    \item the vertex set of $\T \cup \T'$ is $V(\T) \cup V(\T')$,
    \item each vertex in $V(\T)$ has the same parent in $\T \cup \T'$ and $\T$, and
    \item each vertex in $V(\T')$ has the same parent in $\T \cup \T'$ and $\T'$.
\end{enumerate}
\noindent Note that we can check if $\T$ and $\T'$ are compatible, and find $\T \cup \T'$ if they are, in polynomial time. Now we are ready to state the key lemma.

\begin{lemma}\label{lem:dp:basic-combine}
Let $\sigma_1 = (\T, X, \Sol, C, D_1, \xi_1)$ and $\sigma_2 = (\T, X, \Sol, C, D_2, \xi_2)$ be two templates over the same pre-template. Suppose that $D_1$ and $D_2$ are disjoint and that $(\T_i, X_i, \Sol_i)$ is a valid extension of $\sigma_i$ for $i=1,2$. Then $\T_1$ and $\T_2$ are compatible and $(\T_1 \cup \T_2, X_1 \cup X_2, \Sol_1 \cup \Sol_2)$ is a valid extension of $(\T, X, \Sol, C, D_1 \cup D_2, \xi_1 \cup \xi_2)$.

Moreover, if for $i=1,2$, $(\T_i', X_i', \Sol_i')$ is a valid extension of $\sigma_i$ which is not worse than $(\T_i, X_i, \Sol_i)$, then $(\T_1' \cup \T_2', X_1' \cup X_2', \Sol_1' \cup \Sol_2')$ is not worse than $(\T_1 \cup \T_2, X_1 \cup X_2, \Sol_1 \cup \Sol_2)$.
\end{lemma}
\begin{proof}
Observe that since $D_1$ and $D_2$ are disjoint and each of them is a union of components of $G-C$, they are also anticomplete. So the sets $V(\T_1) \setminus V(\T)$ and $V(\T_2) \setminus V(\T)$ are also disjoint and anticomplete. So $V(\T_1) \cap V(\T_2) = V(\T)$ and, by the definition of extensions, it follows that $\T_1$ and $\T_2$ are compatible and that $(\T_1 \cup \T_2, X_1 \cup X_2, \Sol_1 \cup \Sol_2)$ is an extension of the partial solution $(\T, X, \Sol)$. It is also clear that $V(\T_1 \cup \T_2)\setminus V(\T)$ is a subset of $D_1 \cup D_2$.

Now it just remains to consider the run $\zeta$ of $\Aa$ on the $\Sigma$-labelled forest obtained from $\T_1 \cup \T_2$ by equipping it with the labelling $\lbl_{(\T_1 \cup \T_2, X_1 \cup X_2, \Sol_1 \cup \Sol_2)}$. For this, observe that every component of the graph $\T_1 \cup \T_2- V(\T)$ is either a component of $\T_1 - V(\T)$
or a component of $\T_2 - V(\T)$.
For $i=1,2$, the run $\zeta$ within each component of $\T_i - V(\T)$
depends only on this component. Consequently,
the function $\zeta$ gives the same state to each vertex in $V(\T_i) \setminus V(\T)$ as does the run of $\Aa$ on $\T_i$ and $\lbl_{(\T_i, X_i, \Sol_i)}$. The first part of the lemma now follows from the fact that, for any disjoint multisets $A$ and $B$ whose elements are in $Q_{\mathcal{A}}$, we have that $(A \cup B) \wedge \tau_{\mathcal{A}} = \left( (A \wedge \tau_{\mathcal{A}}) \cup (B \wedge \tau_{\mathcal{A}})\right) \wedge \tau_{\mathcal{A}}$.

The second part of the lemma follows immediately from the used total ordering of partial
solutions.
\end{proof}

\subsubsection{Subroutine}
Given as input a simple pre-template $(\T,X,\Sol, C)$ and a sequence $(D_i)_{i=1}^r$ of pairwise distinct components of $G-C$,
we define the following subroutine.
For each $j \in \{0,1,\ldots, r\}$, we set $D_{\leq j} \coloneqq \bigcup_{i=1}^j D_i$. (So $D_{\leq 0}$ is the empty set.)
The subroutine creates an auxiliary table $M'$ with an entry $M'[j,\xi]$ for every $j \in \{0,1,\ldots, r\}$ and every multistate assignment $\xi$ of $\T$.
Each entry $M'[j, \xi]$ will be either the symbol $\bot$, or a valid extension of the template $(\T, X, \Sol, C, D_{\leq j}, \xi)$.
Initially all cells are set to $\bot$.

For $j=0$, there is only one multistate assignment $\xi$ of $\T$ such that the template $(\T, X, \Sol, C, \emptyset, \xi)$ might have a valid extension, and that is the function $\xi \equiv \emptyset$. The unique valid extension is $(\T,X,\Sol)$; so we set $M'[0, \xi \equiv \emptyset] \coloneqq (\T,X,\Sol)$. Then, for $j=1,2,\ldots,r$, we fill the cells $M'[j,\cdot]$ as follows.
We iterate over all multistate assignments $\xi_<$ and $\xi_=$ of $\T$, and, if neither $M'[j-1, \xi_<]$ nor $M[(\T, X, \Sol, C, D_j, \xi_=)]$ is $\bot$, then we apply Lemma~\ref{lem:dp:basic-combine} to combine them into a valid extension $(\T',X',\Sol')$ of $(\T, X, \Sol, C, D_{\leq j}, \xi_< \cup \xi_=)$. If this extension is better than the previous value of $M'[j,\xi_< \cup \xi_=]$, then we set $M'[j,\xi_< \cup \xi_=] \coloneqq (\T',X',\Sol')$.
This finishes the description of the subroutine.

\subsubsection{Outline}
In a preliminary phase, the algorithm iterates over every simple template $\sigma = (\T, X, \Sol, C, D, \xi)$ such that $\xi \equiv \emptyset$. Then it sets $M[\sigma] \coloneqq (\T, X, \Sol)$; note that $(\T, X, \Sol)$ is a valid extension of $\sigma$ which is better than~$\bot$.

In the main phase, the algorithm performs $2|V(G)|$ loops. In each loop, it iterates over every simple template $\sigma = (\T, X, \Sol, C, D, \xi)$ and simple pre-template $\alpha_0 = (\T_0, X_0, \Sol_0, C_0)$ such that $\T$ and $\T_0$ are compatible, $X \cap \T\cap \T_0 = X_0 \cap \T \cap \T_0$, and $\Sol \cap \T \cap \T_0 = \Sol_0 \cap \T \cap \T_0$. The algorithm will try to find a valid extension of $\sigma$ which is better than $M[\sigma]$. The building blocks for constructing this valid extension of $\sigma$ will be the valid extensions $M[\sigma_0]$ where $\sigma_0$ is a simple template over $\alpha_0$. In fact we will be slightly more restrictive about which components of $G-C_0$ we are allowed to ``extend $\alpha_0$ into''.

We call a component $D_0$ of $G-C_0$ \textit{useless} if $\T\cup \T_0$ is a maximal tree-depth-$d$ structure in the subgraph of $G$ induced by $V(\T)\cup V(\T_0) \cup (D \cap D_0)$; we call $D_0$ \textit{useful} otherwise. Note that if $D_0$ is useful, then in particular $D \cap D_0$ is non-empty. We now execute the subroutine on the simple pre-template $\alpha_0$ and the useful components of $G-C_0$, ordered arbitrarily. (If there are no useful components, then we still execute the subroutine on the empty sequence.) The subroutine returns an array $M'$. Write $r$ for the number of useful components of $G-C_0$ and $U \subseteq V(G)$ for their union. Then iterate over all multistate functions $\xi_0$ of $\T_0$ such that $M'[r, \xi_0] \neq \bot$. Thus $M'[r, \xi_0]$ is a valid extension of $(\T_0, X_0, \Sol_0, U, \xi_0)$, which we denote by $(\T_0', X_0', \Sol_0')$.

Now, let $A$ denote the set of all vertices of $\T_0'$ which are an ancestor, in $\T_0'$, of at least one vertex in $D$. If $\T_0'[A]$ and $\T$ are compatible, and if the tuple $$(\T_0'[A] \cup \T, (X_0' \cap A) \cup X, (\Sol_0' \cap A) \cup \Sol)$$is a valid extension of $\sigma$, then update $M[\sigma]$ to the above if it is better than the previous value of $M[\sigma]$. This can be done in polynomial time.

After completing the main phase consisting of $2|V(G)|$ loops as above, the algorithm performs the following finalizing step, which is very similar to the above routine except without $\sigma$. So, for every simple pre-template $\alpha = (\T , X, \Sol, C)$, we execute the subroutine on $\alpha$ and the components of $G-C$ in an arbitrary order. The subroutine returns an array $M'$. Then, writing $r$ for the number of components of $G-C$, we iterate over all multistate functions $\xi$ of $\T$ such that $M'[r, \xi] \neq \bot$. We then check if the valid extension $M'[r, \xi]$ is a feasible solution to the problem.  That is, if $M'[r, \xi] = (\T',X',\Sol')$, we check whether $\mathcal{A}$ accepts the $\Sigma$-labelled forest obtained from $\T'$ by equipping it with the labelling $\lbl_{(\T', X', \Sol')}$. By Proposition~\ref{prop:mainThreshold}, this is equivalent to $(\Sol', X')$ being feasible for $(\td \leq d,\phi)$-\textsc{MWIS} in $G$. Finally, we return the best
solution found, or that there is no solution if none was found.

This concludes the description of the algorithm. Clearly, it runs in $\Oh(|\carvers|^2 \cdot |V(G)|^{2dk+\Oh(1)})$ time. It remains to prove correctness.

\subsection{Correctness}
We may assume that the $(\td \leq d,\phi)$-\textsc{MWIS} problem is feasible since the algorithm checks for feasibility before returning a solution. So there exists a partial solution $(\T, X, \Sol)$ which is $\preceq$-minimal among all partial solutions $(\T', X', \Sol')$ such that $(\Sol', X')$ is feasible for $(\td \leq d,\phi)$-\textsc{MWIS} in $G$. By Lemma~\ref{lem:dp:best-is-neat}, we have that $\T$ is maximal and neat, and $X$ has maximum possible weight among all feasible solution for $(\td \leq d,\phi)$-\textsc{MWIS} in $G$. By Lemma~\ref{lem:dp:neat-equal}, there is no other partial solution in the same equivalence class of $\preceq$ as $(\T, X, \Sol)$.

Since $\carvers$ is a tree-depth-$d$ carver family of defect $k$ in $G$, there exists a tree decomposition $(T,\beta)$ of $G$ as in Definition~\ref{def:carvers}.
That is, for each $t \in V(T)$, we can fix a set of vertices $C_{t} \in \carvers$ such that
\textit{(i)} $C_{t} \cap \T$ contains $\beta(t) \cap \T$ and has size at most $k$, and
\textit{(ii)} for each component $D$ of $G-C_{t}$, there exists a component $T'$ of $T-\{t\}$
such that $D$ is contained in $\beta(t) \cup \bigcup_{s \in T'} \beta(s)$.

We root $T$ in an arbitrary node. We argue that without loss of generality we can assume that
\begin{equation}\label{eq:T-bound}
\mathrm{height}(T) \leq 2|V(G)|.
\end{equation}
Indeed, it is straightforward to verify that one can perform the following operations on $(T,\beta)$ exhaustively:
\begin{enumerate}
\item if there exists a leaf $t$ of $T$ with a parent $t'$ such that $\beta(t) \subseteq \beta(t')$, delete $t$ from $(T,\beta)$;
\item if there exists a node $t$ of $T$ with only one child $t'$ such that $\beta(t) \subseteq \beta(t')$, contract the edge $tt'$ in $(T,\beta)$.
\end{enumerate}
If none of this operation is applicable to $(T,\beta)$, then $T$ has at most $|V(G)|$ leaves and for every edge between a parent $t$ and a child $t'$,
either $t$ has more than one children, or $\beta(t) \setminus \beta(t') \neq \emptyset$.
There are at most $|V(G)|-1$ vertices $t$ of the first type in $T$ and on every root-to-leaf path there are at most $|V(G)|$ edges of the second type.
The bound~\eqref{eq:T-bound} follows.

Consider now a fixed node $t \in V(T)$. We say that a \emph{child component of $t$} is any component of $G-C_t$ which is contained in the union of all bags $\beta(s)$ such that $s$ is a descendant of $t$ in $T$ (including $t$ itself). We define a partial solution $(\T_t, X_t, \Sol_t)$ corresponding to $t$ as follows. Let $\T_t$ be the subgraph of $\T$ induced by all vertices which are an ancestor of at least one vertex in $C_t$. Set $X_t \coloneqq X \cap \T_t$ and $\Sol_t \coloneqq \Sol \cap \T_t$. It is convenient to write $\alpha_t \coloneqq (\T_t, X_t, \Sol_t, C_t)$; so $\alpha_t$ is a simple pre-template. Finally, let $h_t$ be the height of $t$ in the subtree of $T$ rooted at $t$; so the leaves of $T$ have $h_t=1$, for instance.

We will show that after $h_t$ iterations of the algorithm, the following holds for each child component $D$ of $t$: there exists a multistate function $\xi_{t,D}$ of $\T_t$ such that $M[(\alpha_t, D, \xi_{t,D})]$ is precisely the partial solution ``induced by the ancestors of $D \cup C_t$ in $(\T, X, \Sol)$.'' This lemma, which is stated as Lemma~\ref{lem:dp:correct}, will essentially complete the proof. (After $|V(G)|$ rounds, we will consider the child components of the root node of $T$.) However, it is convenient to give some more definitions before stating the lemma.

So consider a fixed node $t \in V(T)$ and a fixed set $D \subseteq V(G)$ which is the union of zero or more components of $G-C_t$. First we define a partial solution $(\T_{t,D}, X_{t,D}, \Sol_{t,D})$ as follows. Let $\T_{t,D}$ be the subgraph of $\T$ induced by all vertices which are an ancestor of at least one vertex in $D\cup C_t$. Set $X_{t,D} \coloneqq X \cap \T_{t,D}$ and $\Sol_{t,D} \coloneqq \Sol \cap \T_{t,D}$. We note that $V(\T_{t,D})\setminus V(\T_t)$ is actually contained in~$D$. To see this, observe that by Lemma~\ref{lem:extensionsTDS}, since $\T$ is neat and extends $\T_t$, each component of $\T-V(\T_t)$ induces a connected subgraph of $G$. Therefore each component of $\T - V(\T_t)$ is either disjoint from or contained in~$D$.

Finally, let $\xi_{t,D}$ denote the multistate function of $\T_t$ defined as follows. If $\xi$ is the run of $\Aa$ on $\T_{t,D}$ equipped with $\lbl_{(\T_{t,D}, X_{t,D}, \Sol_{t,D})}$, then we set:
\[ \xi_{t,D}(\emptyset) \coloneqq \lmset \xi(z)~|~z \textrm{ is a root of } \T_{t,D} \textrm{ but not of } \T_t \rmset \wedge \tau_\Aa, \]
 and, for every $v \in V(\T_t)$,
        \[ \xi_{t,D}(v) = \lmset \xi(z)~|~z \textrm{ is a child of }v \textrm{ in }\T_{t,D}\textrm{ but not in }\T_t \rmset \wedge \tau_\Aa. \]
Notice that $(\T_{t,D}, X_{t,D}, \Sol_{t,D})$ is a valid extension of $(\alpha_t, D, \xi_{t,D})$; denote the latter by $\sigma_{t,D}$. So, if $D$ is a component of $G-C_t$, then $\sigma_{t,D}$ is a simple template.

Our tie-breaking quasi-order and the choice of $(\T, X, \Sol)$ imply that, in fact, $(\T_{t,D},X_{t,D}, \Sol_{t,D})$ is the unique
$\preceq$-minimal valid extension of~$\sigma_{t,D}$.

\begin{lemma}\label{lem:dp:best}
Let $t \in V(T)$ and let $D\subseteq V(G)$ be the union of zero or more components of $G-C_t$. Then $(\T_{t,D}, X_{t,D}, \Sol_{t,D})$ is the only valid extension of $\sigma_{t,D}$ which is not worse than $(\T_{t,D}, X_{t,D}, \Sol_{t,D})$.
\end{lemma}
\begin{proof}
Let $(\T', X', \Sol')$ be a valid extension of $\sigma_{t,D}$ which is not worse than $(\T_{t,D}, X_{t,D}, \Sol_{t,D})$. Let $D_0$ denote the union of all components of $G-C_t$ which are not in $D$. We already observed that $(\T_{t,D_0}, X_{t,D_0}, \Sol_{t,D_0})$ is a valid extension of $\sigma_{t,D_0}$. So, since $D$ and $D_0$ are disjoint, Lemma~\ref{lem:dp:basic-combine} tells us that $\T'$ and $\T_{t,D_0}$ are compatible, and that the component-wise union of $(\T', X', \Sol')$ and $(\T_{t,D_0}, X_{t,D_0}, \Sol_{t,D_0})$ is a valid extension of $(\alpha_t, D \cup D_0, \xi_{t, D} \cup \xi_{t, D_0})$. Also by Lemma~\ref{lem:dp:basic-combine}, this valid extension is not worse than the component-wise union of $(\T_{t,D}, X_{t,D}, \Sol_{t,D})$ and $(\T_{t,D_0}, X_{t,D_0}, \Sol_{t,D_0})$. The latter equals $(\T, X, \Sol)$ and is a valid extension of that same template $(\alpha_t, D \cup D_0, \xi_{t, D} \cup \xi_{t, D_0})$.

In general, the runs of $\Aa$ on any two valid extensions of the same template are the same. By Proposition~\ref{prop:mainThreshold}, the run of $\Aa$ determines whether a partial solution yields a solution to $(\td \leq d, \phi)$-\textsc{MWIS} on $G$. So, by the choice of $(\T, X, \Sol)$ and by Lemma~\ref{lem:dp:neat-equal} applied to the tree $\T$, which is neat, we find that\begin{align*}(\T' \cup \T_{t,D_0}, X' \cup X_{t,D_0}, \Sol' \cup \Sol_{t,D_0}) = (\T, X, \Sol).\end{align*}
\noindent It follows that $(\T', X', \Sol') = (\T_{t,D}, X_{t,D}, \Sol_{t,D})$, as desired.
\end{proof}

We are now ready to prove the main lemma.
\begin{lemma}\label{lem:dp:correct}
Let $t \in V(T)$, and assume that at least $h_t$ iterations of the algorithm have been executed.
Then for any child component $D$ of $t$, we have $M[\sigma_{t,D}] = (\T_{t,D}, X_{t,D}, \Sol_{t,D})$. Furthermore, if the subroutine is executed on $\alpha_t$ and any sequence of child components of $t$, then, where we write $M'$ for the array which is returned, $r$ for the number of child components under consideration, and $U$ for their union, we have $M'[r,\xi_{t,U}] = (\T_{t,U}, X_{t,U}, \Sol_{t,U})$.
\end{lemma}
\begin{proof}
We may assume that the lemma holds for every child of $t$ by induction on $h_t$. We will argue about the first claim of the lemma for the node $t$. Note that the second claim follows
from the first claim and Lemmas~\ref{lem:dp:basic-combine} and~\ref{lem:dp:best} (using induction on $r$). So, fix a child component $D$ of $t$. Note that we only have to show that $M[\sigma_{t,D}]$ is set to $(\T_{t,D}, X_{t,D}, \Sol_{t,D})$ at some point; Lemma~\ref{lem:dp:best} implies that, once this occurs, $M[\sigma_{t,D}]$ is never changed.

For the base case of $h_t = 1$, we have that $t$ is a leaf of $T$ and $D \subseteq \beta(t)$. So, since $C_t \cap \T$ contains $\beta(t) \cap \T$ by the definition of a carver family, the set $D \cap \T$ is empty. Thus $\T_{t,D}=\T_t$ and $\xi_{t, D} \equiv \emptyset$. It follows that, in the preliminary phase, we set $M[\sigma_{t,D}] = (\T_{t}, X_{t}, \Sol_{t})$, as desired. So we may assume that $h_t>1$.

Thus, using the definition of a carver family, there exists a child $s$ of $t$ in $T$ such that we have $D\subseteq \beta(t) \cup \bigcup_{t' \in T'}\beta(t')$, where $T'$ denotes the component of $T-\{t\}$ which contains $s$. (If $D \subseteq \beta(t)$, then there may be more than one such vertex~$s$, and we choose $s$ arbitrarily.) We focus on the $h_t$-th iteration of the algorithm and the moment when the algorithm considers
the simple template $\state_{t,D}$ and the simple pre-template $\alpha_s$. Note that $\T_t$ and $\T_s$ are compatible, and that $\T_t \cup \T_s$ is precisely the subgraph of $\T$ induced by the ancestors of $C_t \cup C_s$. Recall that a component $D_s$ of $G-C_s$ is \textit{useful} if $\T_t \cup \T_s$ is not a maximal tree-depth-$d$ structure in the subgraph of $G$ induced by $V(\T_t) \cup V(\T_s) \cup (D \cap D_s)$.

We need the following key observation.

\begin{claim}\label{cl:dp:D-to-child}
Every useful component of $G-C_s$ is a child component of $s$.
\end{claim}
\begin{claimproof}
Suppose towards a contradiction that $D_s$ is a useful component of $G-C_s$ which is not a child component of $s$. Then the definition of a carver family tells us that $D_s \subseteq \beta(s) \cup \bigcup_{t' \in T'}\beta(t')$, where $T'$ is the component of $T-\{s\}$ which contains $t$. Since $D$ is a child component of $t$, we have that $D \cap D_s \subseteq \beta(t) \cup \beta(s)$.

Since $D \cap D_s$ is disjoint from $C_t \cup C_s$, and the latter contains all vertices of $(\beta(t) \cup \beta(s)) \cap \T$, we also have that $D \cap D_s$ is disjoint from $V(\T)$. Furthermore, $D \cap D_s$ is the union of some subset of components of $G-(C_t \cup C_s)$. By Lemma~\ref{lem:extensionsTDS}, since $\T$ is neat, each component of $\T-(V(\T_t) \cup V(\T_s))$ induces a connected subgraph of $G$; so the vertex set of each such component is either contained in or disjoint from $D \cap D_s$. Hence, the maximality of $\T$ implies that $\T_t \cup \T_s$ is also a maximal tree-depth-$d$ structure in the subgraph of $G$ induced by $V(\T_t) \cup V(\T_s) \cup (D \cap D_s)$. This contradicts the fact that $D_s$ is useful.
\end{claimproof}

As in the outline of the algorithm, let $D_1,\ldots,D_r$ be the useful components of $G-C_s$, in an arbitrary order. Claim~\ref{cl:dp:D-to-child} implies that every $D_j$ is a child component of $s$.
Hence, from the inductive hypothesis, at the beginning of the $h_t$-th iteration we have, for every $1 \leq j \leq r$, that $M[\sigma_{s,D_j}] = (\T_{s,D_j}, X_{s,D_j}, \Sol_{s,D_j})$. We now claim the following.

\begin{claim}\label{cl:dp:subroutine}
In the run of the subroutine, we have for every $0 \leq j \leq r$, that\begin{align*}
 M'[j,\xi_{s,D_{\leq j}}] = (\T_{s,D_{\leq j}}, X_{s,D_{\leq j}}, \Sol_{s,D_{\leq j}}).
 \end{align*}
\end{claim}
\begin{claimproof}
 We prove the claim by induction on $j$. For $j=0$ the claim holds since $\xi_{s,\emptyset} \equiv \emptyset$ and thus $M'[0,\xi_{s,\emptyset}] = (\T_s,X_s, \Sol_s)$.
 For $j > 0$, from the inductive hypothesis on $j$ we have $M'[j-1, \xi_{s, D_{\leq j-1}}] = (\T_{s,D_{\leq j-1}}, X_{s,D_{\leq j-1}}, \Sol_{s,D_{\leq j-1}})$
 and from before, we have $M[\sigma_{s,D_j}] = (\T_{s,D_j}, X_{s,D_j}, \Sol_{s,D_j})$.
 Hence, the partial solution $(\T_{s,D_{\leq j}}, X_{s,D_{\leq j}}, \Sol_{s,D_{\leq j}})$ is considered for $M'[j,\xi_{s,D_{\leq j}}]$;
 Lemma~\ref{lem:dp:best} ensures that it is assigned there and stays till the end. This proves the claim.
\end{claimproof}

After the subroutine is executed, the algorithm iterates over every multistate function $\xi_s$ of $\T_s$ and attempts to use $M'[r, \xi_s]$ to find a better valid extension of $\sigma_{t,D}$ than $M[\sigma_{t,D}]$. By Lemma~\ref{lem:dp:best}, it suffices to prove that when $\xi_{s, D \leq r}$ is considered, the resulting valid extension $(\T_{t,D}, X_{t,D}, \Sol_{t,D})$ of $\sigma_{t,D}$ is found.

By Claim~\ref{cl:dp:subroutine} we have $M'[r,\xi_{s,D_{\leq r}}] = (\T_{s,D_{\leq r}}, X_{s,D_{\leq r}}, \Sol_{s,D_{\leq r}})$. As in the outline of the algorithm, let $A$ denote the set of all vertices of $\T_{s,D_{\leq r}}$ which are an ancestor of at least one vertex in $D$. Note that $\T_{s,D_{\leq r}}[A]=\T[A]$, that $\T[A]$ and $\T_t$ are compatible, and that $\T[A] \cup \T_t$ is precisely the subtree of $\T$ induced by the ancestors of vertices in $D \cap (D_{\leq r} \cup C_s)$ and $C_t$. Thus, it just remains to show that this induced subtree is $\T_{t, D}$, or, equivalently, that every vertex in $(D\cap \T)\setminus(C_s \cup C_t)$ is in a useful component of $G-C_s$ (i.e., in $D_{\leq r}$). This holds by the definition of useful components, because such a vertex can be added to the tree-depth-$d$ structure $\T_s \cup \T_t$. This finishes the proof of Lemma~\ref{lem:dp:correct}.
\end{proof}

By~\eqref{eq:T-bound}, after $2|V(G)|$ iterations,
Lemma~\ref{lem:dp:correct} can be applied to the root of $T$, which we denote by $t$.
Consider now the finalizing step of the algorithm and the moment it considers the pre-template $\alpha_{t}$. Let $M'$ denote the computed array, $r$ the number of components of $G-C_t$, and $U$ the set $V(G)\setminus C_t$. Since every component of $G-C_t$ is a child component of $t$, Lemma~\ref{lem:dp:correct} implies that $M'[r, \xi_{t,U}] = (\T_{t,U},X_{t,U}, \Sol_{t,U})$.
So, as $U = V(G) \setminus C_{t}$, we have $\T_{t,U} = \T$ and thus $M'[r,\xi_{t,U}] = (\T,X,\Sol)$.
As $(\T,X,\Sol)$ is the unique $\preceq$-minimal partial solution such that $(\Sol, X)$ is feasible for $(\td \leq d, \phi)$-\textsc{MWIS} in $G$, the algorithm returns $(\T,X,\Sol)$.

This finishes the proof of Theorem~\ref{thm:dp}.

\section{Minimal separator carving}
\label{sec:minSepCarving}

Given a graph $G$ and a minimal separator $S$ of $G$, we say that a set $\widetilde{S}$ \textit{carves away} a component $D$ of $G-S$ if no component of $G-\widetilde{S}$ intersects both $D$ and another component of $G-S$. (We say that two sets \textit{intersect} if their intersection is non-empty.) In this section we find ``carvers'' for minimal separators. We break up minimal separators into four different types based on which of their full components can be carved away.

First of all, a minimal separator $S$ is \emph{subordinate} if there exists a minimal separator $S'$ and two full sides $A'$ and $B'$ of $S'$
such that $S \subseteq S'$ and some full component of $S$ is disjoint from $A' \cup S' \cup B'$. Notice that any minimal separator which is not subordinate has exactly two full
sides; otherwise we could take $S'=S$ and $A'$ and $B'$ to be two full components of $S$.

The other three types of minimal separator are based on how many full components are ``mesh''. A graph $H$ is \emph{mesh} if its complement $\overline{H}$ is not connected. Otherwise $\overline{H}$ is connected, and we call $H$ \emph{non-mesh}. We say that a minimal separator $S$ is \textit{mesh}/\textit{mixed}/\textit{non-mesh} (respectively) if $S$ is not subordinate and has exactly $2$/$1$/$0$ full components which are mesh.

Now we define carvers for minimal separators based on their type.

\begin{definition}
Let $G$ be a graph, $d$ be a positive integer, $\T$ be a treedepth-$d$ structure in $G$, and let $S$ be a $\T$-avoiding minimal separator of $G$. Then a \textit{$\T$-carver for $S$} is a set $\widetilde{S} \subseteq V(G)$ such that $\widetilde{S} \cap \T = S \cap \T$ and\begin{enumerate}
    \item if $S$ is subordinate or non-mesh, then $\widetilde{S} = S$;
    \item if $S$ is mixed, then $\widetilde{S}$ carves away the mesh full component of $S$; and
    \item if $S$ is mesh, then $\widetilde{S}$ carves away every component of $G-S$.
\end{enumerate}
\end{definition}

\noindent In this section we show how to find a subset of $2^{V(G)}$ which contains carvers for all appropriate $\T$ and $S$; see Proposition~\ref{prop:sepCarving} for a precise statement. Our approach to proving this proposition is based on the theory of modular decompositions.

A \textit{module} of a graph $G$ is a set $X\subseteq V(G)$ such that every vertex in $V(G)\setminus X$ is adjacent to either all of $X$ or none of $X$. A module is \textit{strong} if it  does not cross any other module, where two sets \emph{cross} if they intersect and neither is contained in the other.
A strong module is \textit{maximal} if it is not $V(G)$ and it is not properly contained in any strong module besides $V(G)$.
We do not need the full theory of modular decompositions, just the following fact.

\begin{lemma}[see~\cite{surveyModularDecomp10}]
\label{lem:maximalStrongModules}
The maximal strong modules of a graph $G$ are disjoint and, if $G$ is mesh, then they are the vertex sets of the components of $\overline{G}$.
\end{lemma}

We typically guess two vertices which satisfy the following lemma.

\begin{lemma}
\label{lem:findingPQ}
Let $G$ be a graph, $d$ be a positive integer, $\T$ be a maximal treedepth-$d$ structure in $G$, $S$ be a $\T$-avoiding minimal separator, and $A$ be a full component of $S$. Then $A \cap \T$ is non-empty, and there exists a vertex $p_A \in A$ which has at most $d-1$ neighbors in $\T$. Moreover, if $|A|>1$, then there exists a vertex $q_A \in A$ which is adjacent to $p_A$ and in a different maximal strong module of~$A$ than $p_A$.
\end{lemma}
\begin{proof}
    First notice that $A$ contains a vertex in $\T$. Otherwise, each vertex $a \in A$ would satisfy $N(a) \cap \T\subseteq S$.
    As $\T \cap S$ is contained in a vertical path of $\T$ and does not contain any depth-$d$ vertex of $\T$, we could add $a$ to $\T$ as a leaf without increasing its height beyond $d$, thus contradicting the maximality of $\T$.

    Now choose a vertex $p_A \in A\cap \T$ which has maximum depth among all vertices in $A \cap \T$.
    All vertices in $N(p_A) \cap A \cap \T$ are ancestors of $p_A$ in $\T$.
    All vertices in $N(p_A) \cap \T$ that are descendants of $p_A$ must be in $S$ and thus they are contained in a vertical path of $\T$.
    This means that all vertices in $N(p_A) \cap \T$ are contained in a single vertical path in $\T$, which also contains $p_A$.
    Hence, $|N[p_A] \cap \T| \leq d$, thus $|N(p_A) \cap \T| \leq d-1$.

    Finally, suppose that $|A|>1$. Lemma~\ref{lem:maximalStrongModules} tells us that the maximal strong modules of $A$ partition $A$. There is more than one part since $|A|>1$. So, since $A$ is connected, we can choose a neighbor $q_A$ of $p_A$ which is in a different part from $p_A$.
\end{proof}

We frequently apply the following lemmas from~\cite{p6FreeMaxInd22} to two vertices which come from Lemma~\ref{lem:findingPQ}.

\begin{lemma}[\citeReference{Lemma~4.2}{p6FreeMaxInd22}]
\label{lem:neighborDecomp}
Let $G$ be a graph, let $S$ be a minimal separator, let $A$ be a full component of $S$, and let $p_A$ and $q_A$ be adjacent vertices which are in different maximal strong modules of $A$. Then for any $u \in S$, at least one of the following conditions holds:\begin{enumerate}
    \item there is an induced $P_4$ of the form $uAAA$,
    \item at least one of $p_A$ and $q_A$ is adjacent to $u$, or
    \item the graph $A$ is mesh, and each of its maximal strong modules is either complete or anticomplete to $u$.
\end{enumerate}
\end{lemma}

\noindent We note that the outcomes in Lemma~\ref{lem:neighborDecomp} are not exclusive.

The next lemma helps us to take care of minimal separators which are mesh.

\begin{lemma}[\citeReference{Lemma~4.4}{p6FreeMaxInd22}]
\label{lem:neighborCover}
Let $G$ be a $P_6$-free graph, $S$ be a minimal separator,  $A$ and $B$ be full mesh components of $S$, and  $p_A$ and $q_A$ (respectively, $p_B$ and $q_B$) be adjacent vertices which are in different maximal strong modules of $A$ (respectively, $B$). Then there exist $r_A \in A$ and $r_B \in B$ so that $S \subseteq N(p_A, q_A, r_A, p_B, q_B, r_B)$.
\end{lemma}
Note that since $A$ (resp., $B$) is mesh and $p_A$ and $q_A$ (resp., $p_B$ and $q_B$) are in different maximal strong modules,
we can equivalently say that $N[p_A, q_A, r_A, p_B, q_B, r_B]=A \cup S \cup B$.

We also need the following simplified version of Lemma~\ref{lem:neighborCover} which applies to every type of minimal separator.

\begin{lemma}[\citeReference{Lemma~4.5}{p6FreeMaxInd22}]
\label{lem:neighborCoverGeneral}
Let $G$ be a $P_6$-free graph, $S$ be a minimal separator, and $A$ and $B$ be two full components of $S$. Then there exist $A' \subseteq A$ and $B' \subseteq B$ such that $|A'|\leq 3$, $|B'|\leq 3$, and $S \subseteq N(A' \cup B')$.
\end{lemma}

Note that Lemma~\ref{lem:neighborCoverGeneral} is sufficient to find all subordinate separators.

\begin{corollary}
\label{cor:subordinate}
There is a polynomial-time algorithm which takes in a $P_6$-free graph $G$ and returns a collection $\mathcal{S}_{\textrm{sub}} \subseteq 2^{V(G)}$ which contains each subordinate minimal separator.
\end{corollary}
\begin{proof}
    Let $S$ be a subordinate minimal separator. Then there exists a minimal separator $S'$ and two full sides $A'$ and $B'$ of $S'$ so that $S \subseteq S'$ and some full component of $S$ is disjoint from $A' \cup S' \cup B'$. By Lemma~\ref{lem:neighborCoverGeneral}, there exist $A'' \subseteq A'$ and $B'' \subseteq B'$ so that $|A''|\leq 3$, $|B''|\leq 3$, and $S' \subseteq N(A'' \cup B'')$. We guess $A''$ and $B''$. Then, for each component $D$ of $G-N(A'' \cup B'')$, we insert $N(D)$ into $\mathcal{S}$. The full component of $S$ which is disjoint from $A' \cup S' \cup B'$ is itself such a component $D$. So $\mathcal{S}$ contains $S$ and $|\mathcal{S}|\leq |V(G)|^6$.
\end{proof}

Finally, some types of minimal separator can be taken care of very quickly using the following lemma. We state it in a slightly weaker fashion than in~\cite{p6FreeMaxInd22}.

\begin{lemma}[\citeReference{Lemma~5.5}{p6FreeMaxInd22}]
\label{lem:nonMesh}
There is a polynomial-time algorithm which takes in a $P_6$-free graph $G$ and returns a collection $\mathcal{F} \subseteq 2^{V(G)}$ which contains each full component of a non-mesh separator of $G$.
\end{lemma}

Now we are ready to prove the main result of this section.

\begin{proposition}
\label{prop:sepCarving}
For each positive integer $d$, there exists a polynomial-time algorithm which takes in a $P_6$-free graph $G$ and returns a collection $\mathcal{S} \subseteq 2^{V(G)}$ such that for any maximal treedepth-$d$ structure $\T$ in $G$ and any $\T$-avoiding minimal separator $S$, the collection $\mathcal{S}$ contains a $\T$-carver for $S$.
\end{proposition}
\begin{proof}
    Let $d$, $G$, $\T$, and $S$ be as in the statement of the proposition. We will show how to construct $\mathcal{S}$ by making ``guesses'' among polynomially-many options. We will separately consider four cases, depending on the type of $S$. We output the collection $\mathcal{S}$ consisting of all sets $\widetilde{S}$ constructed as described below.

\smallskip
    \paragraph{\textbf{Case 1. $S$ is subordinate.}}  Recall that in Corollary~\ref{cor:subordinate} we constructed the family $\mathcal{S}_{\textrm{sub}}$ that contains all subordinate minimal separators $S$. As $S$ is a $\T$-carver for $S$, it is sufficient to include $\mathcal{S}_{\textrm{sub}}$ in the output family $\mathcal{S}$.

\smallskip
    \paragraph{\textbf{Case 2. $S$ is non-mesh.}} By Lemma~\ref{lem:nonMesh} we can, in polynomial time, find a collection $\mathcal{F} \subseteq 2^{V(G)}$ which contains each full component of a non-mesh separator of $G$. For each $D \in \mathcal{F}$, we insert $N(D)$ into $\mathcal{S}$. So $\mathcal{S}$ contains $S$, which is a $\T$-carver for $S$.

\medskip

    From now on we may assume that $S$ is either mixed or mesh. Thus $S$ has exactly two full components, and at least one of them is mesh. Let $A$ and $B$ be the two full components of $S$. (We are not guessing $A$ and $B$, we are just giving them names.) Next, guess the vertices in $S \cap \T$ (there are at most $d-1$ of them) and add them to a set $\widetilde{S}$. As we proceed throughout the proof, we will add more and more vertices of $V(G)\setminus \T$ to $\widetilde{S}$. Thus we will always have that $\widetilde{S}\cap \T = S \cap \T$, and we are trying to show that $\widetilde{S}$ eventually becomes a $\T$-carver for $S$.

    By Lemma~\ref{lem:findingPQ}, there exists a vertex $p_A \in A$ (respectively, $p_B \in B$) which has at most $d-1$ neighbors in $\T$. For convenience, write $T_A \coloneqq N(p_A)\cap A \cap \T$ and $T_B \coloneqq N(p_B)\cap B \cap \T$. We guess the vertices $p_A$ and $p_B$ and the sets $T_A$ and $T_B$. We then add the vertices in $N(p_A)\setminus T_A$ and $N(p_B)\setminus T_B$ to $\widetilde{S}$.

    If either $A$ or $B$ has size one, then this set $\widetilde{S}$ already contains $S$ and is therefore a $\T$-carver for $S$. So we may assume that $|A|>1$ and $|B|>1$. Thus, by Lemma~\ref{lem:findingPQ}, there exists $q_A \in A$ (respectively, $q_B \in B$) so that $p_A$ and $q_A$ (respectively, $p_B$ and $q_B$) are adjacent vertices which are in different maximal strong modules of $A$ (respectively, $B$). We add every vertex which is in both $N(\{p_A, q_A\} \cup T_A)$ and $N(\{p_B, q_B\} \cup T_B)$ to $\widetilde{S}$; note that these newly added vertices are a subset of $S$.

    It is helpful to state the following observation; note that it will hold even after we add more vertices to $\widetilde{S}$.

    \smallskip
    \begin{enumerate}
        \item[(1)] Each vertex $u \in S \setminus \widetilde{S}$ is non-adjacent to $p_A$ and $p_B$ and therefore in a $P_3$ of the form $uAA$ and in a $P_3$ of the form $uBB$.
    \end{enumerate}
In particular, note that when applying Lemma~\ref{lem:neighborDecomp} for any $u \in S \setminus \widetilde{S}$ and $A$ (resp., $B$) we never obtain the first outcome, as then we would get an induced $P_6$ of the form $BBuAAA$ (resp., $AAuBBB$).

\smallskip
    \paragraph{\textbf{Case 3. $S$ is mixed.}}
    We claim that $\widetilde{S}$ is already a $\T$-carver for $S$. By symmetry between $A$ and $B$, we may assume that $A$ is mesh and $B$ is non-mesh. So it just remains to show that $A$ is carved away by $\widetilde{S}$; that is, that no component of $G-\widetilde{S}$ intersects both $A$ and another component of $G-S$. We will do this by showing that $S \setminus \widetilde{S}$ and $A \setminus \widetilde{S}$ are anticomplete. So consider a vertex $u \in S \setminus \widetilde{S}$. By (1), the second outcome of Lemma~\ref{lem:neighborDecomp} holds for $B$, and $u \in N(p_B, q_B)$. So $u$ is anticomplete to $\{p_A, q_A\} \cup T_A$; otherwise we would have $u\in \widetilde{S}$. Now the third outcome of Lemma~\ref{lem:neighborDecomp} holds for $A$; that is, each maximal strong module of $A$ is either complete or anticomplete to $u$. As $A$ is mesh and $p_A \notin N(u)$, each neighbor of $u$ in $A$ is in $N(p_A)\setminus T_A$. Since $N(p_A)\setminus T_A \subseteq \widetilde{S}$, this completes the proof that $\widetilde{S}$ is a $\T$-carver for $S$.

    \smallskip
    \paragraph{\textbf{Case 4. $S$ is mesh.}}
    Then by Lemma~\ref{lem:neighborCover}, there exist $r_A \in A$ and $r_B \in B$ so that $S \subseteq N(p_A, q_A, r_A, p_B, q_B, r_B)$, i.e., $N[p_A, q_A, r_A, p_B, q_B, r_B]=A \cup S \cup B$. Guess these vertices $r_A$ and $r_B$.
    So for each component $D$ of $G-N[p_A, q_A, r_A, p_B, q_B, r_B]$, add the vertices in $N(D)$ to $\widetilde{S}$.
    Furthermore, we add to $\widetilde{S}$ all vertices in $N[q_A,r_A] \cap N[q_B, r_B]$.
    Note that these newly added vertices are a subset of $S$.

    We will show that now $\widetilde{S}$ is a $\T$-carver for~$S$. It just remains to show that $\widetilde{S}$ carves away the components of $G-S$: that is, that each component of $G-\widetilde{S}$ intersects at most one component of $G-S$. Since $N(D)$ as explicitly added to $\widetilde{S}$ for each component $D$ of $G-[p_A, q_A, r_A, p_B, q_B, r_B]$, the only possibility that needs to be checked is that some component of $G-\widetilde{S}$ intersects both $A$ and $B$. So it suffices to show that $S\setminus \widetilde{S}$ has a partition into two parts, $S_A$ and $S_B$, so that $S_A$ and $S_B$ are anticomplete, $S_A$ and $B\setminus \widetilde{S}$ are anticomplete, and $S_B$ and $A\setminus \widetilde{S}$ are anticomplete.

    Let $S_A$ (respectively, $S_B$) be the set of all vertices in $S \setminus \widetilde{S}$ which are in $N(q_A,r_A)$ (respectively, $N(q_B,r_B)$).
    The sets $S_A$ and $S_B$ partition $S \setminus \widetilde{S}$ by observation~(1) and the definitions of $r_A,r_B$ and $\widetilde{S}$.
    Now consider a vertex $u \in S_A$. Again using (1), the third outcome of Lemma~\ref{lem:neighborDecomp} holds for $B$; each maximal strong module of $B$ is either contained in or disjoint from the neighborhood of $u$. So each neighbor of $u$ in $B$ is in $N(p_B)\setminus T_B$, and therefore also in $\widetilde{S}$. By this and the symmetric argument for $S_B$, we have proven that $S_A$ and $B\setminus \widetilde{S}$ are anticomplete, and $S_B$ and that $A\setminus \widetilde{S}$ are anticomplete.

    It just remains to show that $S_A$ and $S_B$ are anticomplete. For this we need to be slightly more careful about the argument above; notice that we actually have that if $u \in S_A$, then $u$ is anticomplete to every component of $\overline{B}$ which intersects $\{p_B, q_B, r_B\}$. Let $M_B$ denote the union of these components of $\overline{B}$. Note that $M_B$ induces a connected subgraph of $B$ since $p_B$ and $q_B$ are in different components of $\overline{B}$. Thus, if $u$ was adjacent to a vertex $v \in S_B$, then we could find a $P_6$ of the form $M_AM_AuvM_BM_B$, where, symmetrically, $M_A$ is the union of the components of $\overline{A}$ which intersect $\{p_A, q_A, r_A\}$.

    This completes all four cases and therefore the proof of Proposition~\ref{prop:sepCarving}.
\end{proof}

\section{Improving carvers for mixed minimal separators}\label{sec:impMixedSeps}
We need a more refined understanding of mixed minimal separators.
We will use the tools developed here in Section~\ref{sec:2SidedAlignedPMCs} to find carvers for the so-called \emph{two-sided PMCs}.

So let $G$ be a graph, $S$ be a mixed minimal separator of $G$, and $A$ and $B$ be the mesh and non-mesh full sides of $S$, respectively. Given a set $\widetilde{S} \subseteq V(G)$, we say that a component $\widetilde{D}$ of $G-\widetilde{S}$ is \textit{clarified} if it is disjoint from $A \cup B$. In this section we show how to ``carve away'' all of the clarified components; see Proposition~\ref{prop:betterMixedCarver}.

To prove this proposition, we will use the following enumeration routine to obtain a ``fuzzy'' version of the mesh full component. Given a graph $G$, a \emph{fuzzy version} of a set $A \subseteq V(G)$ is a set $A^+ \subseteq V(G)$ such that $A \subseteq A^+$
and every vertex of $A^+ \setminus A$ is complete to $A$.

\begin{lemma}[\citeReference{Lemma~5.6}{p6FreeMaxInd22}]
\label{lem:mixedFuzzy}
There is a polynomial-time algorithm which takes in a $P_6$-free graph $G$ and returns a collection $\mathcal{A} \subseteq 2^{V(G)}$ such that for every mixed minimal separator $S$ in $G$ with $A$ as its full mesh component, there exists $A^+ \in \mathcal{A}$
that is a fuzzy version of $A$.
\end{lemma}

We also use the following lemma about minimal elements in quasi-orders. A \textit{quasi-order} is a pair $(X, \preceq)$ so that $X$ is a set and $\preceq$ is a reflexive and transitive relation on $X$.

\begin{lemma}[\citeReference{Lemma~4.1}{p6FreeMaxInd22}]
\label{lem:biRanking}
Let $X$ be a non-empty finite set, and let $(X, \preceq_0)$ and $(X, \preceq_1)$ be quasi-orders such that each pair of elements of $X$ is comparable either with respect to $\preceq_0$ or with respect to $\preceq_1$ (or both). Then there exists an element $x \in X$ such that for every $y \in X$, either $x \preceq_0 y$ or $x \preceq_1 y$ (or both).
\end{lemma}

We use Lemma~\ref{lem:biRanking} to prove the following lemma, which will help us recognize an independent set which is contained in a mixed minimal separator.

\begin{lemma}
\label{lem:indSet}
Let $G$ be a $P_6$-free graph,  $S \subseteq V(G)$ be a set with a mesh full component $A$, and  $I \subseteq S$ be a non-empty independent set. Then there exist a component $M_I$ of $\overline{A}$ and a vertex $x \in I \cap N(M_I)$ so that every vertex $y \in I\setminus N(M_I)$ is a neighbor of every component $D$ of $G-(A \cup S)$ so that $x\in N(D)$.
\end{lemma}
\begin{proof}
    For convenience, let $\mathcal{D}$ denote the collection of components of $G-(A \cup S)$, and let $\mathcal{M}$ denote the collection of components of $\overline{A}$; we will obtain one quasi-order from $\mathcal{D}$ and another from $\mathcal{M}$. Notice that if there are two vertices $u,v \in I$ such that there exists both a pair $D_u,D_v \in \mathcal{D}$ so that $N(D_u) \cap \{u,v\} = \{u\}$ and $N(D_v) \cap \{u,v\} = \{v\}$, and a pair $M_u,M_v \in \mathcal{M}$ so that $N(M_u) \cap \{u, v\} =  \{u\}$ and $N(M_v) \cap \{u, v\} =  \{v\}$, then there is a $P_6$ of the form $D_uuM_uM_vvD_v$.

    Consider the quasi-orders $\preceq_0$ and $\preceq_1$ on $I$ defined as follows:
    \begin{align*}
    u \preceq_0 v &\Longleftrightarrow \{D \in \mathcal{D}~|~u \in N(D)\} \subseteq \{D \in \mathcal{D}~|~v \in N(D)\} \textrm{, and}\\
    u \preceq_1 v &\Longleftrightarrow \{M \in \mathcal{M}~|~u \in N(M)\} \subseteq \{M \in \mathcal{M}~|~v \in N(M)\}.
    \end{align*}
    Any two $u,v \in I$ are comparable in at least one of these orders. Hence, Lemma~\ref{lem:biRanking} asserts that there exist $x \in I$ such that for every $y \in I$ either $x \preceq_0 y$ or $x \preceq_1 y$. We pick any $M_I \in \mathcal{M}$ with $x$ as a neighbor (it exists since $I \subseteq S = N(A)$).
\end{proof}

We are ready to prove the main proposition about improving carvers for mixed minimal separators.

\begin{proposition}
\label{prop:betterMixedCarver}
For each positive integer $d$, there exists a polynomial-time algorithm which takes in a $P_6$-free graph $G$ and a set $\widetilde{S} \subseteq V(G)$ and returns a collection $\mathcal{S}' \subseteq 2^{V(G)}$ so that for any maximal treedepth-$d$ structure $\T$ in $G$ and any $\T$-avoiding mixed minimal separator $S$ of $G$, there exists $S' \in \mathcal{S}'$ so that \begin{enumerate}
    \item $S'$ contains $\widetilde{S}$,
    \item $S'\cap \T \subseteq S \cup \widetilde{S}$, and
    \item for each clarified component $\widetilde{D}$ of $G-\widetilde{S}$, no component of $\widetilde{D}-S'$ intersects more than one component of $G-S$.
\end{enumerate}
\end{proposition}
\begin{proof}
    Let $d$, $G$, $\widetilde{S}$, $S$, and $\T$ be as in the lemma statement. Let $A$ and $B$ denote the mesh and non-mesh full sides of $S$, respectively. Additionally, let $\mathcal{D}$ denote the set of all vertices of $G-S$ which are in a clarified component of $G-\widetilde{S}$. So $\mathcal{D}$ is the union of some components of $G-(A \cup S \cup \widetilde{S} \cup B)$, and the graph $G-\widetilde{S}$ has no path between $\mathcal{D}$ and $A \cup B$. We will find a set $S'$ which satisfies conditions~\textit{(i)} and~\textit{(ii)} of the proposition and includes $N(\mathcal{D}) \cap S$; this implies condition~\textit{(iii)}.

    Notice that there are at most $d$ components of $\overline{A}$ which intersect $\T$; let $M\subseteq V(G)$ denote the union of these components. We claim that we can guess $M$. By Lemma~\ref{lem:mixedFuzzy}, we can, in polynomial-time, obtain a set $\mathcal{A} \subseteq 2^{V(G)}$ which includes a fuzzy version of $A$. That is, there exists $A^+\in \mathcal{A}$ so that $A \subseteq A^+$ and $A^+ \setminus A$ is complete to $A$. Guess this set $A^+ \in \mathcal{A}$; there are polynomially-many choices. Now each component of $\overline{A}$ is also a component of the complement of $A^+$ and can thus be guessed. So indeed we can guess $M$, as it is the union of at most $d$ components of $\overline{A^+}$. We will use the fact that $M$ is non-empty, which follows from the fact that $A \cap \T$ is non-empty by Lemma~\ref{lem:findingPQ}.

    Now we define an intermediate set $X \subseteq V(G)$ which contains $\widetilde{S}$ and is our current best guess at $S'$. To begin with we set $X\coloneqq \widetilde{S} \cup N(M)$; these vertices are safe to include since $N(M) \cap \T \subseteq S$. Next, by Lemma~\ref{lem:findingPQ}, there exists a vertex $p_B \in B$ which has at most $d-1$ neighbors in $\T$. We guess this vertex, along with which of its neighbors are in $\T\cap B$, and then we add all of its other neighbors to $X$. This completes the definition of $X$. Notice that $X \cap \T \subseteq S \cup \widetilde{S}$, that $S \setminus X$ is anticomplete to $M \cup \{p_B\}$, and that $G-X$ has no path between $\mathcal{D}$ and $A \cup B$ (this follows from the fact that $\widetilde{S} \subseteq X$). We also remark that $X \subseteq A \cup B \cup S \cup \widetilde{S}$.

    We claim that there exists a vertex $q_B \in B$ which is complete to $S \setminus X$. If $S \setminus X$ is empty, then this is trivially true, so assume that it is non-empty. Then $|B|>1$ since $S \setminus X$ is anticomplete to $p_B$. So by Lemma~\ref{lem:findingPQ}, there is a vertex $q_B \in B$ so that $p_B$ and $q_B$ are adjacent and in different maximal strong modules of $B$. If $S \setminus X$ is not complete to $q_B$, then by Lemma~\ref{lem:neighborDecomp} applied to the full component $B$ of $S$, we obtain a vertex $u \in S \setminus X$ which is in a $P_4$ of the form $uBBB$. However, $u$ is also in a $P_3$ of the form $uAA$ since $S \setminus X$ is anticomplete to $M$ (which is non-empty). But then we obtain a $P_6$ of the form $AAuBBB$, which contradicts the fact that $G$ is $P_6$-free. Consequently, that $S \setminus X$ is complete to $q_B$. We guess such a vertex $q_B$.

    Now form an independent set $I \subseteq S\setminus X$ as follows. For each component of $S\setminus X$ which has a neighbor in $\mathcal{D}$, choose one vertex with a neighbor in $\mathcal{D}$ and add that vertex to $I$. (We are not saying that we can guess $I$, just that it exists.) We may assume that $I$ is non-empty since otherwise the proposition holds with $S' \coloneqq X$. Now apply Lemma~\ref{lem:indSet} to the subgraph induced on $A \cup S  \cup \mathcal{D}$. Thus, there exist a component $M_I$ of $\overline{A}$ and a vertex $x \in I \cap N(M_I)$ so that every vertex $y \in I\setminus N(M_I)$ is a neighbor of every component $D$ of $\mathcal{D}$ so that $x\in N(D)$. We can guess $M_I$ for the same reason we were able to guess $M$ (because $M_I$ is a component of $\overline{A}$ and we can guess the fuzzy version $A^+$ of $A$).

    We will prove that the following set $S'$ satisfies the proposition. First we add $X$ and $N(M_I)\cap N(q_B)$ to $S'$. These vertices are safe to add since $X\cap \T \subseteq S \cup \widetilde{S}$ and $N(M_I)\cap N(q_B) \subseteq S$.
We observe that since $X \subseteq A \cup B \cup S \cup \widetilde{S}$, we have $S' \subseteq A \cup B \cup S \cup \widetilde{S}$ at this moment.
    Now consider each component $D$ of $G-X - N(q_B)$ which has $x$ as a neighbor.
    Clearly, $D$ is disjoint from $S$ as $S \setminus X \subseteq N(q_B)$.
    Let $H$ be a component of $N(q_B) \setminus X$ that contains a neighbor of $D$.
    Since $x$ has a neighbor in $\mathcal{D}$, there is a component of $G-\widetilde{S}$ that contains $H$, $D$, $x$, and a component of $\mathcal{D}$, hence, it is disjoint from $A \cup B$.
    In particular, $D$ is disjoint with $A \cup B$, so $N(D) \subseteq S \cup \widetilde{S}$ as $X \subseteq A \cup B \cup S \cup \widetilde{S}$.
    Furthermore, we have $H \subseteq S$. Over all choices of $D$ and $H$ as above, we add the component $H$ to $S'$.

    We have already proved that conditions~\textit{(i)} and~\textit{(ii)} of the proposition hold for $S'$. Recall that, in order to obtain the final condition~\textit{(iii)}, it is enough to show that $S'$ contains $N(\mathcal{D}) \cap S$. So, going for a contradiction, suppose that there exists a vertex $u \in \mathcal{D}$ which has a neighbor $v \in S \setminus S'$. Let $H$ be the component of $S \setminus X$ which contains $v$. Then $x$ is disjoint from and anticomplete to $H \cup \{u\}$, since otherwise we would have added $v$ to $S'$. However, now there is a $P_6$ of the form $uvq_BxM_IM$, which contradicts the fact that $G$ is $P_6$-free. (To see that there is a $P_6$ of this form, recall that $q_B$ is complete to $S \setminus X$, $x$ has a neighbor in $M_I$ while $u$ and $v$ do not, and $S \setminus X$ is anticomplete to $M$, which is non-empty; therefore $M_I$ is a component of $\overline{A}$ which is not any of the components of $\overline{A}$ we used to define $M$.) This contradiction completes the proof of Proposition~\ref{prop:betterMixedCarver}.
\end{proof}

\section{Not-two-sided PMCs}
\label{sec:not2Sided}
A potential maximal clique $\pmc$ in a graph $G$ is \emph{two-sided}
if there exist two distinct connected components $D_1,D_2$ of $G-\pmc$
such that for every connected component $D$ of $G-\pmc$, we have $N(D) \subseteq N(D_1)$
or $N(D) \subseteq N(D_2)$.

The following statement has been essentially proven in~\cite{p6FreeMaxInd22}.
However, it has been proven only with the \textsc{Max Weight Independent Set} problem in mind,
so we need to slightly adjust the argumentation to fit the more general setting of this paper.

\begin{theorem}\label{thm:not-two-sided}
For every positive integer $d$
there exists a polynomial-time algorithm that, given a $P_6$-free graph $G$
outputs a family $\carvers \subseteq 2^{V(G)}$ with the following guarantee:
for every maximal tree-depth-$d$ structure $\T$ in $G$
and every potential maximal clique $\pmc$ of $G$ that is $\T$-avoiding and not two-sided,
there exists $C \in \carvers$ that is a container for $\pmc$, i.e., $\pmc \subseteq C$
and $C \cap V(\T) = \pmc \cap V(\T)$.
\end{theorem}

As mentioned, Theorem~\ref{thm:not-two-sided} is essentially proven in Section~5 of~\cite{p6FreeMaxInd22}.
There, for a fixed maximal independent set $I$, a PMC $\pmc$ is \emph{$I$-free} if it is disjoint with $I$.
This assumption here is replaced with $\pmc$ being $\T$-avoiding for a fixed tree-depth-$d$ structure $\T$.
Informally speaking, to adjust it to our setting, we need to make three adjustments within the proof of~\cite{p6FreeMaxInd22}.
\begin{enumerate}
\item Often, when mesh component $D$ is analyzed, it is argued that the independent set $I$ intersects at most one maximal module $M_p$ of $D$, and a vertex $p \in M_p \cap I$ is guessed. This step is usually followed by a guess of an arbitrary vertex $q$ in a different maximal strong module of $D$.

In our case, the tree-depth-$d$ structure $\T$ can intersect at most $d$ modules of $D$, and the guess of $p$ is replaced
with a guess of a set $P$ of at most $d$ vertices of $\T \cap D$, one vertex from each maximal strong module of $D$ that intersects $\T$.
For $q$, it is enough to take an arbitrary vertex of $D$, unless $|P|=1$ (i.e., $\T$ intersects only one maximal strong module of $D$)
where we need to pick $q$ from a different maximal strong module. In this manner, we maintain the property that $P \cup \{q\}$ contains
vertices of at least two maximal strong modules of $D$, so in particular $D \subseteq N[P \cup \{q\}]$.
Whenever later the proof of~\cite{p6FreeMaxInd22} considers $N[p]$ or $N[p,q]$, we consider here $N[P]$ or $N[P \cup \{q\}]$ instead.

In what follows, we call such a set $P$ a \emph{footprint of $\T$ in $D$} and the vertex $q$ a \emph{satellite of the footprint $P$}.
\item When a PMC $\pmc$ that is disjoint with the maximal independent set $I$ is analyzed, and we often argue that the maximality of $I$ implies that every $v \in \pmc$ has a neighbor in $I$ that is outside $\pmc$.
In our case, Lemma~\ref{lem:PMCmaximality} gives the same corollary, except for the vertices of $\T \cap \pmc$, but there are fewer than $d$
of them and they can be guessed separately.
\item Finally, the notion of a \emph{neighbor-maximal} component of Section~5.3 of~\cite{p6FreeMaxInd22} is a bit incompatible
with our statement, as it considers two components $D_1,D_2$ of $G-\pmc$ with $N(D_1) = N(D_2)$ both
\emph{not} neighbor-maximal. This definition restricts the set of all PMCs with more than two neighbor-maximal components.
We observe that the assumption ``more than two neighbor-maximal components'' is used only once in the proof
and can be easily replaced with the (slightly weaker) assumption of being not two-sided.
\end{enumerate}

Let us now have a closer look at Section~5 of~\cite{p6FreeMaxInd22} and provide formal details.
The toolbox in the earlier sections nor Lemmas~5.2 up to Lemma~5.7 use the notion of $I$-freeness, so they
can be used in our setting without any modifications.

Lemma~5.8 of~\cite{p6FreeMaxInd22}, the main result of Section~5 there, would now obtain the following form.
\begin{lemma}[analog of Lemma~5.8 of~\cite{p6FreeMaxInd22}]\label{lem:gkpp-58}
For every integer $d$ there exists a polynomial-time algorithm that, given on input
a $P_6$-free graph $G$, outputs two families $\mathcal{F}_9^1$ and $\mathcal{F}_9^2$ such that the following
holds: for every maximal tree-depth-$d$ structure $\T$ in $G$
and every potential maximal clique $\pmc$ of $G$ that is $\T$-avoiding and not two-sided,
either $\mathcal{F}_9^1$ contains $\pmc$ or $\mathcal{F}_9^2$ contains a triple $(\pmc \cup D_1 \cup D_2, D_1^+, D_2^+)$
for some components $D_1,D_2$ of $G-\pmc$ that are mesh, where $D_i^+$ is a fuzzy version of $D_i$ for $i \in \{1,2\}$.
\end{lemma}
Note that Theorem~\ref{thm:not-two-sided} follows easily from
Lemma~\ref{lem:gkpp-58}: we insert into $\carvers$ every element of $\mathcal{F}_9^1$ and, for
every $(K, L_1,L_2) \in \mathcal{F}_9^2$, every choice of at most $d$ maximal strong modules of $L_1$
and every choice of at most $d$ maximal strong modules of $L_2$, we insert into $\carvers$ the set $K$ minus the chosen modules.
Thus, it remains to prove Lemma~\ref{lem:gkpp-58}.

The proof of Lemma~5.8 of~\cite{p6FreeMaxInd22} splits into three lemmas: Lemma~5.9, Lemma~5.10, and Lemma~5.11.
These statements have a fixed $P_6$-free graph $G$ and a maximal independent set $I$ in their context.
In our setting, instead of $I$ we fix an integer $d$ and a maximal tree-depth-$d$ structure $\T$ in $G$.

Lemma~5.9 of~\cite{p6FreeMaxInd22} takes the following form.
\begin{lemma}[analog of Lemma~5.9 of~\cite{p6FreeMaxInd22}]\label{lem:gkpp-59}
Suppose $\pmc$ is a $\T$-avoiding PMC in $G$ and $D$ is a component of $G-\pmc$ which is mesh.
Let $P$ be a footprint of $\T$ in $D$ and let $q$ be a satellite of $P$.
Let $J \subseteq N(D)$ be an independent set with the following property: for every $v \in J$,
    the set $N(v) \cap D$ consists of some maximal strong modules of $D$ and is disjoint with $\T \cap D$.
Then there exists $w \in D$ and a component $D'$ of $G-\pmc$, distinct from $D$, such that
$J \subseteq (\T \cap \pmc) \cup N(w) \cup N(D')$.
\end{lemma}
\begin{proof}[Proof sketch.]
The proof of Lemma~5.9 of~\cite{p6FreeMaxInd22} uses $I$-freeness of $\pmc$
in only one place: to argue that if $v \in J$ is anti-complete to all vertices of $I$
in $D$, then it needs to be adjacent to a vertex of $I$ in another component $D'$ of $G-\pmc$, so in particular
it is adjacent to some other component of $G-\pmc$.
In our case, $J$ is anti-complete to $D \cap \T$, and Lemma~\ref{lem:PMCmaximality} gives the same corollary, except for
the vertices of $\T \cap \pmc$ that need to be added there separately.
The rest of the proof is the same.
\end{proof}

Similarly we adjust Lemma~5.10 of~\cite{p6FreeMaxInd22}.
\begin{lemma}[analog of Lemma~5.10 of~\cite{p6FreeMaxInd22}]\label{lem:gkpp-510}
Given a family $\mathcal{X} \subseteq 2^{V(G)}$, one can in time polynomial in the size of $G$ and the size of $\mathcal{X}$
compute a family $\mathcal{F}_7(\mathcal{X}) \subseteq 2^{V(G)}$ with the following properties:
for every $\T$-avoiding PMC $\pmc$ and every component $D$ of $G-\pmc$, if all components of $G-\pmc$, except
for possibly $D$, belong to $\mathcal{X}$, then all components of $G-\pmc$ belong to $\mathcal{F}_7(\mathcal{X})$.
\end{lemma}
\begin{proof}[Proof sketch.]
The assumption on $I$-freeness of Lemma~5.10 of~\cite{p6FreeMaxInd22} comes into play in the proof only
in the last case, namely Case~3, where in particular $D$ is mesh.

First, after Claim~1, we guess a vertex $p \in I \cap M_p$ for the unique maximal strong module $M_p$ of $D$ that
intersects $I$, and a vertex $q$ in another maximal strong module.
Here, we perform the standard adjustment, guessing instead a footprint of $\T$ in $D$ and its satellite.

Second, in the definition of $Y$, we also want to exclude the vertices of $\T \cap D$ from it (there are fewer than $d$
of them, so we just try all possibilities).

Third, after Claim~7 we invoke Lemma~5.9. Because of the previous adjustment, the set $J$ here is
disjoint with $\T \cap \pmc$. Hence, we can invoke the adjusted Lemma~\ref{lem:gkpp-59} instead.
\end{proof}

We now move to Lemma~5.11 of~\cite{p6FreeMaxInd22}.
\begin{lemma}[analog of Lemma~5.11 of~\cite{p6FreeMaxInd22}]\label{lem:gkpp-511}
One can in polynomial time compute a family $\mathcal{F}_8$ such that the following holds:
Take any $\T$-avoiding PMC $\pmc$ and assume there are different components $D_1,D_2$ of $G-\pmc$
that are meshes. Then $\mathcal{F}_8$ contains either $D_1$, or $D_2$, or $\pmc \cup D_1 \cup D_2$.
\end{lemma}
\begin{proof}[Proof sketch.]
Again, the proof in~\cite{p6FreeMaxInd22} starts by selecting, for every $i \in \{1,2\}$,
  a vertex $p_i \in I$ in the unique maximal strong module of $D_i$ that intersects $I$.
We adjust it in the standard way by selecting a footprint $P_i$ of $\T$ and its satellite $q_i$.

Then, when defining $X$ and $Z$, we need to also include $\T \cap \pmc$ into $X$, so $Z$ is disjoint with $\T$.
Since $\T \cap \pmc$ is of size less than $d$, we just try all possibilities.

Finally, Claim~11 relies on $I$-freeness. It argues that a vertex $z \in Z \subseteq N(D_2)$ that does not have
a neighbor in $I \cap D_2$, needs to have a neighbor
in $I$ in another component of $G-\pmc$, in particular it is adjacent to another component of $G-\pmc$.
In our case, $z$ is not in $\T$ (as it is in $Z$) and $z$ has no neighbor in $\T \cap D_2$,
   so Lemma~\ref{lem:PMCmaximality} gives the same corollary.
\end{proof}

With the above three lemmas in hand, we can now adjust the proof of Lemma~5.8 of~\cite{p6FreeMaxInd22}
to show Lemma~\ref{lem:gkpp-58}.
The crucial insight is that if there are two components $D_1,D_2$ of $G-\pmc$ with $N(D_1) = N(D_2)$,
then $N(D_1)$ is subordinate (because $N(D_1)$ has three full sides, $D_1$, $D_2$, and a component containing $\pmc \setminus N(D_1)$)
and hence it belongs to the family $\mathcal{S}_{\textrm{sub}}$ provided by Corollary~\ref{cor:subordinate}.

Therefore, by adding all full components of subordinate separators to a constructed set $\mathcal{G}$, we obtain the same properties
as in the proof of Lemma~5.8 of~\cite{p6FreeMaxInd22} under the weaker assumption that $\pmc$ is not two-sided:
$\mathcal{G}$ contains either all components of $G-\pmc$, or
all except for at most two mesh components $D_1$ and $D_2$.
The first outcome allows us to recover $\pmc$ exactly. In the second outcome we use Lemma~\ref{lem:gkpp-511}: we either
get exactly $\pmc$ or the set $\pmc \cup D_1 \cup D_2$. In the latter case,
it remains to get, for every $i\in\{1,2\}$, a fuzzy version of $D_i$.

Since $D_i \notin \mathcal{G}$, $N(D_i)$ is not subordinate.
Since $\pmc$ is not two-sided, there is another component $D'$ of $G-\pmc$,
distinct from $D_1$ and $D_2$, such that $N(D') \not\subseteq N(D_i)$. This component is in $\mathcal{G}$.
Then, Lemma~5.7 of~\cite{p6FreeMaxInd22} gives a polynomial number of candidates for a fuzzy version of $D_i$.

  This completes the proof sketch of Lemma~\ref{lem:gkpp-58} and thus concludes
  the proof of Theorem~\ref{thm:not-two-sided}.

\section{Analysis of two-sided aligned PMCs}
\label{sec:2SidedAlignedPMCs}

In this section we deal with the last remaining type of PMCs: two-sided aligned PMCs. Contrary to the previous sections, we need to make some delicate surgery on the clique tree in order to adjust it before generating a small family of carvers. More precisely, we will need the following special property of a clique tree $(T,\beta)$ of a chordal completion $G+F$. (Recall here that full components of adhesions were defined following Lemma~\ref{lem:cliqueTreesAdhesion}.)

\begin{enumerate}
\item[$(\spadesuit)$] There are no two distinct edges $st,tu \in E(T)$ such that
\begin{enumerate}
    \item $\sigma(st) \subseteq \sigma(tu)$;
    \item $\sigma(tu)$ is a mixed minimal separator; and
    \item the full component of $\sigma(tu)$ on the $u$-side is non-mesh.
\end{enumerate}
\end{enumerate}

The next lemma verifies that property $(\spadesuit)$ can be always achieved, even without changing the completion set $F$.

\begin{lemma}
\label{lem:findingCliqueTree}
For any graph $G$ and minimal chordal completion $G+F$ of $G$, there exists a clique tree $(T, \beta)$ of $G+F$ with property~$(\spadesuit)$.
\end{lemma}
\begin{proof}
    We already know that $G+F$ has some clique tree. We will choose a clique tree which maximizes a certain count; for this definition we need to orient some edges of the tree. So, given a clique-tree $(T, \beta)$ of $G+F$, orient each edge of $T$ whose adhesion is a mixed minimal separator ``towards the non-mesh side''. That is, if $tu \in E(T)$ is such that $\sigma(tu)$ is a mixed minimal separator and the full component of $\sigma(tu)$ on the $u$-side is non-mesh, then orient $tu$ as $(t, u)$.

    Now, choose a clique tree $(T, \beta)$ of $G+F$ which maximizes the sum, over all nodes $u\in V(T)$, of the number of undirected edges which are incident to a node of $T$ that can be reached from $u$ via a directed path (that is, a path which does not use any undirected edge and which follows the directed edges according to their direction). Such a choice exists since all clique trees have the same number of nodes. We will prove that $(T, \beta)$ satisfies the conditions of the lemma. So, going for a contradiction, suppose that there exist distinct edges $st,tu \in E(T)$ so that \textit{(i)}, \textit{(ii)}, and \textit{(iii)} of property~$(\spadesuit)$ hold. By conditions~\textit{(ii)} and \textit{(iii)}, $tu$ is oriented as $(t,u)$.

    For convenience, set $S \coloneqq \sigma(tu)$, and let $A$ (respectively, $B$) denote the full component of $S$ on the $t$-side (respectively, $u$-side).
    Since $S$ is mixed, it has exactly two full components: $A$ that contains $\beta(t) \setminus S$ and
    $B$ that contains $\beta(u) \setminus S$.
    Since $\sigma(st) \subseteq S$ and $\sigma(st)$ separates $\beta(s) \setminus \sigma(st)$
    from both $\beta(t) \setminus S$ and $\beta(u)\setminus S$, it follows that the full component of $\sigma(st)$ on the $s$-side is disjoint from $A \cup S \cup B$.
    In particular,  this component cannot be a full component of $S$ (which has only two full sides, $A$ and~$B$), hence $\sigma(st)\subsetneq S$.
    Therefore $\sigma(st)$ is subordinate and the edge $st$ of $T$ is undirected.

    Now we define a new clique tree $(T', \beta')$ of $G+F$ as follows. Replace the edge $st$ of $T$ with the edge $su$; that is, reattach the component of $T-\{st\}$ that contains $s$ to be connected via an edge $su$ instead of the edge $st$. Since $\sigma(st) \subseteq S = \sigma(tu)$, the resulting tree is in fact a clique tree of $G+F$. Furthermore, the orientations of the edges do not change; $su$ is an undirected edge as $S$ is a subordinate separator, while for every other edge of $T$ the full sides considered in the orientation remain the same. Moreover, the relevant count of $(T', \beta')$ is strictly larger than that of $(T, \beta)$: the count for $u$ increases by one, while no other count decreases. This contradicts the choice of $(T, \beta)$ and completes the proof of Lemma~\ref{lem:findingCliqueTree}.
\end{proof}

Now we are ready to prove the main result of this section.

\begin{proposition}\label{prop:two-sided}
For each positive integer $d$, there exists a polynomial-time algorithm which takes in a $P_6$-free graph $G$
and returns a collection $\C_1 \subseteq 2^{V(G)}$ so that for any maximal treedepth-$d$ structure $\T$ in $G$
and any $\T$-aligned minimal chordal completion $G+F$ of $G$,
there exists a clique tree $(T,\beta)$ of $G+F$
such that for each node $t$ of $T$, if $\beta(t)$ is two-sided and $\T$-avoiding,
then the set $\C_1$ contains a $(\T, (T,\beta))$-carver for $\beta(t)$.
\end{proposition}

\begin{proof}
    Let $d$, $G$, $\T$, $F$ be as in the lemma statement. Let $(T,\beta)$ be a clique tree of $G+F$ which satisfies property~$(\spadesuit)$, its existence is guaranteed by Lemma~\ref{lem:findingCliqueTree}. We orient some of the edges of $T$ as in the proof of Lemma~\ref{lem:findingCliqueTree}. That is, for each edge $tu$ of $T$ so that $\sigma(tu)$ is a mixed minimal separator, we orient $tu$ as $(t,u)$ if the full component of $\sigma(tu)$ on the $u$-side is non-mesh, and as $(u,t)$ otherwise.
In this language, property~$(\spadesuit)$ becomes the following.

    \smallskip
    \begin{enumerate}
        \item[$(\clubsuit)$] There do not exist distinct edges $st, tu \in E(T)$ such that $\sigma(st) \subseteq \sigma(tu)$ and $tu$ is oriented towards $u$.
    \end{enumerate}
    \smallskip

    Now fix $t \in V(T)$ such that $\beta(t)$ is two-sided and $\T$-avoiding. We will argue how to construct a $(\T,(T,\beta))$-carver for $\beta(t)$ using guesswork with only polynomially-many options. To this end, set $\pmc \coloneqq \beta(t)$, and let $D_0$ and $D_1$ be the components of $G-\pmc$ which witness that $\Omega$ is two-sided. Throughout the rest of this proof we write indices on subscripts modulo~$2$.

    First of all, for each $v \in V(G)$, we add the set $N[v]$ to $\C_1$.
    Note that this takes care of all PMCs $\pmc$ which contain a vertex $v$ that does not have a neighbor outside $\pmc$.
    Indeed, by the characterization of PMCs in Proposition~\ref{prop:PMCsChar}, we would have $N[v] = \pmc$ and thus $\pmc \in \C_1$.

    Thus from now on we may assume that each vertex from $\pmc$  has a neighbor outside of $\pmc$.
    We now use the characterization of PMCs in Proposition~\ref{prop:PMCsChar} to infer the following claim.

    \begin{claim}\label{claim:two-sided-struct}
    The following properties hold:
        \begin{enumerate}
            \item  $N(D_0) \cup N(D_1) = \pmc$,
            \item the sets $N(D_0) \setminus N(D_1)$ and $N(D_1) \setminus N(D_0)$ are nonempty and complete to each other, and
            \item there exists $j \in \{0,1\}$ such that $D_j$ is complete to $N(D_j) \setminus N(D_{j+1})$.
        \end{enumerate}
    \end{claim}
    \begin{claimproof}
        By assumption, each vertex in $\pmc$ has a neighbor outside of $\pmc$. Since $\pmc$ is two-sided, we infer that $N(D_0) \cup N(D_1) = \pmc$. Since $N(D_0)$ and $N(D_1)$ are proper subsets of $\pmc$ (see Proposition~\ref{prop:PMCsChar}), we have that both $N(D_0) \setminus N(D_1)$ and $N(D_1) \setminus N(D_0)$ are nonempty. From~ Proposition~\ref{prop:PMCsChar}, we infer that $N(D_0) \setminus N(D_1)$ is complete to $N(D_1) \setminus N(D_0)$, as there is no connected component of $G-\pmc$ that is adjacent to some vertices in both those sets.

        Finally, suppose towards a contradiction that for every $i \in \{0,1\}$, there exists $v_i \in N(D_i) \setminus N(D_{i+1})$ that is not complete to $D_i$. Then there is a $P_6$ of the form $D_0D_0v_0v_1D_1D_1$. This contradiction completes the proof of Claim~\ref{claim:two-sided-struct}.
    \end{claimproof}

    By Lemma~\ref{lem:PMCComponents}, for $i \in \{0,1\}$, the set $N(D_i)$ is a minimal separator of $G$ which has a full side $D_i^\pmc \neq D_i$ that contains $\pmc \setminus N(D_i)$. Since $\pmc$ is two-sided and $N(D_0) \cup N(D_1) = \Omega$ by part~\textit{(i)} of Claim~\ref{claim:two-sided-struct}, it follows that $D_i^\pmc$ is precisely the union of $\pmc \setminus N(D_i)$ and the components of $G-\pmc$ which have a neighbor in $\Omega\setminus N(D_i)$. Thus, in particular, $D_0^\pmc \cap D_1^\pmc =\emptyset$ since $\pmc$ is two-sided.

    We now show how the adhesions relate to the components of $G-\pmc$. For each component $D$ of $G-\pmc$, we write $T_D$ for the component of $T-\{t\}$ so that $D \subseteq \bigcup_{s \in T_D} \beta(s)$; this component exists and is unique. We write $t_D$ for the node of $T_D$ which is a neighbor of $t$ in $T$. We also write $t_0$ and $t_1$ as shorthand for $t_{D_0}$ and $t_{D_1}$, respectively, and similarly for $T_0$ and $T_1$.

    We now attempt to ``capture'' the minimal separators $N(D_0)$ and $N(D_1)$. By Proposition~\ref{prop:sepCarving}, we can, in polynomial-time, obtain a collection $\mathcal{S} \subseteq 2^{V(G)}$ which contains a $\T$-carver for each $\T$-avoiding minimal separator. So in particular, $\mathcal{S}$ contains $\T$-carvers $S_0$ and $S_1$ for $N(D_0)$ and $N(D_1)$, respectively. We can guess these sets $S_0$ and $S_1$ since $\mathcal{S}$ also has polynomial size.

We will use the following observation twice.
\begin{claim}\label{cl:no-trash-in-pmc}
Let $k \in \{0,1\}$ be such that $tt_k$ is not oriented towards $t$.
Then no component of $G-S_k$ intersects both $N(D_k) \setminus N(D_{k+1})$ and $D_k^\pmc$.
\end{claim}
\begin{claimproof}
Let $D$ be a component of $G-S_k$ that intersects $D_k^\pmc$. Since $tt_k$ is not oriented towards $t$, by the properties of $S_k$ we have that $S_k$ carves away $D_k^\pmc$, hence $D \subseteq N(D_k) \cup D_k^\pmc$.

Assume there exists $v \in D \cap (N(D_k) \setminus N(D_{k+1}))$. Since $v \in N(D_k) \setminus S_k$ while $S_k \cap \T = N(D_k) \cap \T$, we have $v \notin \T$.
By Lemma~\ref{lem:PMCmaximality}, there exists $w \in \T \setminus \pmc$ that is a neighbor of $v$.
Since $w \in \T \setminus \pmc$ while $S_k \cap \T = N(D_k) \cap \T$, we have $w \notin S_k$, thus $w \in D \setminus \pmc$.
As $D \subseteq N(D_k) \cup D_k^\pmc$, every component $D'$ of $G-\pmc$ that intersects $D$ satisfies
$N(D') \subseteq N(D_{k+1})$. This is a contradiction with $v \in N(D_k) \setminus N(D_{k+1})$.
\end{claimproof}

A \emph{precarver} is
a set $\widetilde{S} \subseteq V(G)$ such that
$\widetilde{S} \cap \T = \pmc \cap \T$ and there exists $k \in \{0,1\}$ such that
for every component $\widetilde{D}$ of $G-\widetilde{S}$
at least one of the following conditions holds:
\begin{itemize}
\item there exists a component $T'$ of $T-\{t\}$ with $\widetilde{D} \subseteq \beta(t) \cup \bigcup_{t' \in V(T')} \beta(t')$, or
\item $tt_k$ is oriented and $\widetilde{D}$ is clarified with regards to the mixed separator $N(D_k)$ (i.e., $\widetilde{D}$ is disjoint with $D_k \cup D_k^\pmc$).
\end{itemize}
If we are able to guess a precarver $\widetilde{S}$, then
we apply Proposition~\ref{prop:betterMixedCarver} for the minimal separator $N(D_k)$ to guess
a superset $C$ of $\widetilde{S}$ with $C \cap \T = \widetilde{S} \cap \T$. Then the properties of the precarver together with Proposition~\ref{prop:betterMixedCarver} imply that $C$ will be
a $(\T,(T,\beta))$-carver for $\pmc$. Hence, in the remainder of the proof
we focus on guessing a precarver.

We observe that the first bullet of the definition of a precarver holds immediately
for a component $\widetilde{D}$
if $\widetilde{D} \subseteq \pmc$ or there exists a component $D$ of $G-\pmc$
such that $\widetilde{D} \subseteq D \cup \pmc$. The latter applies in to the case $\widetilde{D}\cap \pmc=\emptyset$.

We perform now case distinction on how the edges $tt_0$ and $tt_1$ are oriented in $(T,\beta)$,
   which is in fact a case distinction on the types of separators $N(D_0)$ and $N(D_1)$.
\medskip

\noindent\textbf{Case 1. There exists $k \in \{0,1\}$ such that $tt_k$ is undirected.}
We claim that then $\widetilde{S} = S_0 \cup S_1$ is a precarver.
To this end, let $\widetilde{D}$ be a component of $G-\widetilde{S}$.

If $\widetilde{D} \subseteq \pmc$, there is nothing to prove, so assume otherwise.
Let $D$ be a component of $G-\pmc$ that intersects $\widetilde{D}$.
If $N(D) \subseteq N(D_k)$, then $D$ is a component of $G-N(D_k)$ and thus, as $tt_k$ is undirected and $S_k$
is a carver for $tt_k$, we have $\widetilde{D} \subseteq D \cup N(D_k) \subseteq D \cup \pmc$.

If $tt_{k+1}$ is undirected too, then a symmetric argument resolves the case $N(D) \subseteq N(D_{k+1})$.
Since $\pmc$ is two-sided, this completes the proof in this case.

Otherwise, $tt_{k+1}$ is directed; without loss of generality assume $k=1$.
Recall that we are left with analysing a component $\widetilde{D}$ of $G-\widetilde{S}$ that satisfies the following: for every component $D$ of $G-\pmc$ that intersects $\widetilde{D}$,
we have $N(D) \not\subseteq N(D_1)$ (so $N(D) \subseteq N(D_0)$ and $N(D) \cap (N(D_0) \setminus N(D_1)) \neq \emptyset$, as $\pmc$ is two-sided).
This implies that $\widetilde{D} \subseteq \pmc \cup D_1^\pmc$.

\noindent\textbf{Case 1.1. $tt_0$ is oriented towards $t$.}
As $\widetilde{D}$ intersects $D_1^\pmc$,
from Claim~\ref{cl:no-trash-in-pmc} for $k=1$ we infer that $\widetilde{D}$ is disjoint with $N(D_1) \setminus N(D_0)$.
Recall that $\widetilde{D}$ is also disjoint with every component $D$ of $G-\pmc$ with $N(D) \subseteq N(D_1)$.
Thus, $\widetilde{D}$ is disjoint with $D_0^\pmc$,
as $D_0^\pmc$ consists of $\pmc \setminus N(D_0) = N(D_1) \setminus N(D_0)$ and every component $D$ of $G-\pmc$
with $N(D) \subseteq N(D_1)$ and $N(D) \cap (N(D_1) \setminus N(D_0)) \neq \emptyset$.

If $\widetilde{D}$ intersects $D_0$, then, by the properties of the carver $S_0$, $\widetilde{D} \subseteq D_0 \cup N(D_0)$
and we are done.
Otherwise, $\widetilde{D}$ is clarified with regards to the mixed separator $N(D_0)$, because it is disjoint with both full sides: $D_0$ and $D_0^\pmc$.
Hence, $\widetilde{S}$ is a precarver.

\noindent\textbf{Case 1.2. $tt_0$ is oriented towards $t_0$.}
Since $N(D) \not\subseteq N(D_1)$, we have $t_D = t_0$, as otherwise the edge $tt_D$ is an undirected edge
with $\sigma(tt_D) \subseteq \sigma(tt_0)$, violating property~$(\clubsuit)$.
As the above holds for every component $D$ of $G-\pmc$ that intersects $\widetilde{D}$,
   we have $\widetilde{D} \subseteq \pmc \cup \bigcup_{t' \in V(T_0)} \beta(t')$ and we are done.

\medskip

\noindent\textbf{Case 2. Both $tt_0$ and $tt_1$ are oriented towards $t$.}
 Then both $D_0$ and $D_1$ are mesh. Note that $N(D_0)\cap N(D_1)$ is a minimal separator with full sides $D_0$ and $D_1$ in a induced subgraph of $G$. So by Lemma~\ref{lem:neighborCover} applied to this induced subgraph, we can pick at most three elements of $D_0$ and at most three elements of $D_1$ so that every vertex in $N(D_0)\cap N(D_1)$ is a neighbor of one of these six (or fewer) vertices. By adding at most one more vertex from a different component of $\overline{D_0}$, and similarly for $\overline{D_1}$, we obtain sets $D_0' \subseteq D_0$ and $D_1' \subseteq D_1$ so that $|D_0'| \leq 4$, $|D_1'| \leq 4$, and every vertex in $D_0$, $D_1$, and $N(D_0)\cap N(D_1)$ is in $N[D_0' \cup D_1']$. Guess these sets $D_0'$ and $D_1'$.

    For every $i \in \{0,1\}$, recall that there are at most $d$ components of $\overline{D_i}$ which intersect $\T$; let $M_i\subseteq V(G)$ denote the union of these components.
    Since Lemma~\ref{lem:mixedFuzzy} allows us to guess a fuzzy version of $D_i$, we can guess $M_i$, as every component of $\overline{D_i}$ is a component of the complement of a fuzzy version of $D_i$.
    We set
    \[ \widetilde{S} := S_0 \cup S_1 \cup (N[D_0' \cup D_1'] \setminus (M_0 \cup M_1)). \]
    We claim that $\widetilde{S}$ is a precarver.
    As $N(D_0) \cup N(D_1) = \pmc$, it is immediate that $\widetilde{S} \cap \T = \pmc \cap \T$.

    Recall from part~\textit{(iii)} of Claim~\ref{claim:two-sided-struct} that there exists $j \in \{0,1\}$ such that $D_j$ is complete to $N(D_j) \setminus N(D_{j+1})$.
    By symmetry, we can assume that $D_1$ is complete to $N(D_1) \setminus N(D_0)$. Hence, $N(D_1) \setminus N(D_0) \subseteq \widetilde{S}$.
    Since also $N(D_0) \cap N(D_1) \subseteq \widetilde{S}$ due to the inclusion of $N[D_0' \cup D_1'] \setminus (M_0 \cup M_1)$, we have $N(D_1) \subseteq \widetilde{S}$.

    Consider now a component $\widetilde{D}$ of $G-\widetilde{S}$. We claim that either $\widetilde{D}$
    is contained in $D \cup \pmc$ for a single component $D$ of $G-\pmc$ or
    $\widetilde{D}$ is clarified with regards to the minimal separator $N(D_0)$ (whose full sides are $D_0$ and $D_0^\pmc$).
    The claim is trivial if $\widetilde{D} \subseteq \pmc$.
    If there exists $i \in \{0,1\}$ such that $\widetilde{D}$ intersects $D_i$,
    then $\widetilde{D} \subseteq \pmc \cup D_i$ due to the inclusion of the carvers $S_0$ and $S_1$ in $\widetilde{S}$.
    If $\widetilde{D}$ intersects a component $D\notin\{ D_0,D_1\}$ of $G-\pmc$ such that $N(D) \subseteq N(D_1)$, then
    $\widetilde{D} \subseteq D$ as $N(D_1) \subseteq \widetilde{S}$.
    In the remaining case, $\widetilde{D}$ intersects a component $D \notin \{D_0,D_1\}$ with $N(D) \subseteq N(D_0)$,
    $N(D) \cap (N(D_0) \setminus N(D_1)) \neq \emptyset$.
    Furthermore, due to the exclusion of the previous cases, $\widetilde{D}$ is disjoint both with $D_0$ and with $D_0^\pmc$,
    as the latter consists of $D_1$, $N(D_1) \setminus N(D_0)$ (which is a subset of $\widetilde{S}$)
    and all components $D'$ of $G-\pmc$ with $N(D') \subseteq N(D_1)$, $N(D') \cap (N(D_1) \setminus N(D_0)) \neq \emptyset$.
    Hence, $\widetilde{D}$ is clarified with regards to $N(D_0)$. This finishes the proof that $\widetilde{S}$ is a precarver.

\medskip

\noindent\textbf{Case 3. Both $tt_0$ and $tt_1$ are oriented away from $t$.}
By Lemma~\ref{lem:adh-to-comp}, for every $s \in N_T(t)$ we have $\sigma(st) \subseteq N(D_0)$
or $\sigma(st) \subseteq N(D_1)$.
Hence, by property~$(\clubsuit)$, $t_0$ and $t_1$ are the only two neighbors of $t$ in $G$.

We claim that $\widetilde{S} = S_0 \cup S_1$ is a precarver in this case.
Consider a component $\widetilde{D}$ of $G-\widetilde{S}$.

If there exists $k \in \{0,1\}$ such that $\widetilde{D}$ intersects $D_k^\pmc$,
then, by the properties of the carver $S_k$, we have $\widetilde{D} \subseteq D_k^\pmc \cup N(D_k)$.
Consequently, for every component $D$ of $G-\pmc$ that intersects $\widetilde{D}$,
it holds that $N(D) \subseteq N(D_{k+1})$, $N(D) \cap (N(D_{k+1}) \setminus N(D_k)) \neq \emptyset$.
We infer $t_D = t_{k+1}$ for every such component $D$. Hence, $\widetilde{D} \subseteq \pmc \cup \bigcup_{t' \in T_{k+1}} \beta(t')$.

If $\widetilde{D}$ is disjoint with $D_0^\pmc \cup D_1^\pmc$, then it is disjoint also with
$D_0 \cup D_1$ as $D_{k+1} \subseteq D_k^\pmc$ for every $k \in \{0,1\}$.
Hence, $\widetilde{D}$ is clarified with regards to both $N(D_0)$ and $N(D_1)$.
This finishes the proof that $\widetilde{S}$ is a precarver.

\medskip

\noindent\textbf{Case 4. One of the edges $tt_0$ and $tt_1$
  is oriented towards $t$ and one is oriented away from $t$.}
Without loss of generality, assume $tt_0$ is oriented towards $t_0$ and $tt_1$ is oriented towards $t$.

We distinguish the following two subcases.

\noindent\textbf{Case 4.1. There exists a component $D$ of $G-\pmc$, $D \neq D_1$, with $N(D) \cap (N(D_1) \setminus N(D_0)) \neq \emptyset$.}
Let $D$ be such a component and let $v \in N(D) \cap (N(D_1) \setminus N(D_0))$. We argue that
\begin{align}\label{eq:case4.1}
&\mathrm{For\ every\ }u \in (N(D_0) \cap N(D_1)) \setminus N(D),\mathrm{\ there\ is\ no\ }P_4\mathrm{\ of\ the\ form\ }\\\nonumber &uD_0D_0D_0,
\mathrm{and\ if\ additionally\ }uv \notin E(G),\mathrm{\ then\ }u\mathrm{\ is\ complete\ to\ }D_0.
\end{align}
Let $u \in (N(D_0) \cap N(D_1)) \setminus N(D)$.
Let $Q$ be an induced path consisting of a shortest path from $u$ to $v$ possibly via $D_1$ if $uv \notin E(G)$
and then a neighbor of $v$ in $D$. Observe that $Q$ has three vertices if $uv \in E(G)$ and at least four vertices
if $uv \notin E(G)$.

If there exists an induced $P_4$ of the form $uD_0D_0D_0$, then the concatenation of this $P_4$ with $Q$ yields a $P_6$,
a contradiction.
Similarly, if there exists an induced $P_3$ of the form $uD_0D_0$ (which is equivalent to $u$ not being complete to $D_0$),
then the concatenation of this $P_3$ with $Q$ yields a $P_6$ if $uv \notin E(G)$.
This proves~\eqref{eq:case4.1}.

For every $k \in \{0,1\}$, apply Lemma~\ref{lem:findingPQ} to the separator $N(D_k)$ with full component $D_k$,
obtaining a vertex $p_k \in D_k \cap \T$ with $A_k := \T \cap D_k \cap N(p_k)$ of size at most $d-1$
and, if $|D_k|>1$, a vertex $q_k \in D_k \cap N(p_k)$ in a different maximal strong module of $D_k$ than $p_k$. We set $q_k = p_k$ if $|D_k|=1$.

Let
\[ \widetilde{S} = S_0 \cup S_1 \cup \left(\bigcup_{k \in \{0,1\}} (N(p_k) \setminus A_k)\right) \cup (N(q_0) \cap N(\{v,q_1\})) \cup N(D). \]
Note that $\widetilde{S}$ can be guessed with polynomial number of options, as $N(D)$ is a subordinate separator and hence
can be guessed using Corollary~\ref{cor:subordinate}.

We claim that $\widetilde{S}$ is a precarver.
Since $N(q_0) \cap N(\{v,q_1\}) \subseteq N(D_0) \subseteq \pmc$, we have
$\widetilde{S} \cap \T = \pmc \cap \T$.

We now show that
\begin{equation}\label{eq:eat-intersection}
N(D_0) \cap N(D_1) \subseteq \widetilde{S}.
\end{equation}
Let $u \in N(D_0) \cap N(D_1)$. If $u \in N(D)$ or $u \in N(p_0)$, then $u \in \widetilde{S}$.
Otherwise, $u$ is not complete to $D_0$, so by~\eqref{eq:case4.1} we have $u \in N(v)$ and there is no $P_4$ of the form
$uD_0D_0D_0$. Lemma~\ref{lem:neighborDecomp} implies that $u \in N(q_0)$. Hence, $u \in \widetilde{S}$. This proves~\eqref{eq:eat-intersection}.

Consider now a component $\widetilde{D}$ of $G-\widetilde{S}$.
We distinguish two cases, depending on whether $\widetilde{D}$ intersects $D_0^\pmc$.

If $\widetilde{D}$ intersects $D_0^\pmc$, then by the properties of the carver $S_0$ we have
$\widetilde{D} \subseteq N(D_0) \cup D_0^\pmc$.
By Claim~\ref{cl:no-trash-in-pmc} for $k=0$, $\widetilde{D}$ is disjoint with $N(D_0) \setminus N(D_1)$.
By~\eqref{eq:eat-intersection}, $\widetilde{D}$ is disjoint with $N(D_0)$,
that is, $\widetilde{D} \subseteq D_0^\pmc$. In particular, $\widetilde{D}$ is disjoint with $D_1^\pmc$.

If $\widetilde{D}$ intersects $D_1$ then, by the properties of the carver $S_1$, we have $\widetilde{D} \subseteq D_1 \cup N(D_1)$.
Otherwise, $\widetilde{D}$ is disjoint with both $D_1$ and $D_1^\pmc$ and thus is clarified with regards to the separator $N(D_1)$.

In the other case, the component $\widetilde{D}$ is disjoint with $D_0^\pmc$. So $N(D)\subseteq N(D_0)$ for every component $D$ of $G-\pmc$ that intersects $\widetilde{D}$.
If there exists a component $D$ of $G-\pmc$ with $N(D) \subseteq N(D_0) \cap N(D_1)$ that intersects
$\widetilde{D}$, then $\widetilde{D} \subseteq D$ thanks to~\eqref{eq:eat-intersection}.
Otherwise, for every component $D$ of $G-\pmc$ that intersects $\widetilde{D}$ we have $N(D) \not\subseteq N(D_1)$.
By Lemma~\ref{lem:adh-to-comp} and property~$(\clubsuit)$, for every such component we have $t_D = t_0$.
Thus, $\widetilde{D} \subseteq \beta(t) \cup \bigcup_{t' \in V(T_0)} \beta(t')$.

This finishes the proof that $\widetilde{S}$ is a precarver in this case.

\noindent\textbf{Case 4.2. For every component $D$ of $G-\pmc$, either $D=D_1$ or $N(D) \subseteq N(D_0)$.}
Lemma~\ref{lem:adh-to-comp} and property~$(\clubsuit)$ imply that $t$ is of degree $2$ in $T$,
that is, $t_0$ and $t_1$ are the only two neighbors of $t$ in $T$.

For every $k \in \{0,1\}$, proceed as follows. Call a node $t$ of $T$ considered in this case {\em{special}}; since this is the last case, we may assume that for all non-special nodes of $T$, we already have constructed carvers for their bags.
Let $t_k'$ be the closest to $t$ node of $T_k$ that is not special.
Note that $t_k'$ exists and is unique, as every node of $T$ that is special has degree $2$ in $T$. (It may happen that $t_k' = t_k$).
Let $Q_k$ be the path in $T$ between $t$ and $t_k'$.

As this is the last case, we can guess a $(\T,(T,\beta))$-carver $C_1$ for $\beta(t_1')$.
(Note that this guesswork may involve Lemma~\ref{lem:leafyPMC} if $\beta(t_1')$ is not $\T$-avoiding or
 Theorem~\ref{thm:not-two-sided} if $\beta(t_1')$ is $\T$-avoiding but not two-sided.)
Let $A_1 = \T \cap (C_1 \setminus \pmc) = \T \cap (\beta(t_1') \setminus \pmc)$; as $|A_1| \leq d$, we can guess $A_1$.

We now perform an analysis of components of $G-\pmc$.
\begin{claim}\label{cl:case4.2}
Let $D$ be a component of $G-\pmc$ distinct from $D_0$ and $D_1$ and let $k^D \in \{0,1\}$ be such that $t_D = t_{k^D}$.
Then there exists an edge $t_A^D t_B^D$ of $T$ such that:
\begin{itemize}
    \item $\sigma(t_A^D t_B^D) = N(D)$.
    \item If $T_A^D$ is the component of $T-\{t_A^D t_B^D\}$ that contains $t_A^D$, then $D \subseteq \bigcup_{t' \in V(T_A^D)} \beta(t')$.
    \item The nodes $t_A^D$, $t_B^D$, $t_{k^D}'$, $t_{k^D}$, and $t$ lie on the unique path between $t_A^D$ and $t$ in $T$ in this order,
    with possibly $t_B^D = t_{k^D}'$ and/or $t_{k^D}' = t_{k^D}$.
\end{itemize}
In particular, if $T^D$ is the unique component of $T-\{t_{k^D}'\}$ that contains $t_A^D$, then $t \notin V(T^D)$ and
$D \subseteq \bigcup_{t' \in V(T^D)} \beta(t')$.
\end{claim}
\begin{claimproof}
Let $D$ be as in the statement.
By Lemma~\ref{lem:PMCComponents}, $N(D)$ is a minimal separator with full sides $D$ and $D^\pmc$, where $D^\pmc$ contains $\pmc \setminus N(D)$.
By the assumptions of the current case, $N(D) \subseteq N(D_0)$.
Furthermore, as $N(D_0)$ is mixed, $N(D_0)$ has only two full sides: $D_0$ and $D_0^\pmc$ that contains $\pmc \setminus N(D_0)$. As both of them are disjoint with $D$, it follows that $D$ is not a full component of $G-N(D_0)$, that is,
$N(D)$ is a proper subset of $N(D_0)$.
Hence, $D^\pmc$ contains not only $\pmc \setminus N(D)$, but also both $D_0$ and $D_0^\pmc$, which in turn contains $D_1$.

Since $N(D) \subseteq \pmc$, $N(D)$ is a clique in $G+F$.
We apply Lemma~\ref{lem:comp-to-adh} for $S = N(D)$, $A = D$, and $B = D^\pmc$, obtaining the edge $t_A^{D}t_B^{D}$.
The first two promised properties are immediate by Lemma~\ref{lem:comp-to-adh}.

For the third property, since $N(D)$ is a subordinate separator, $t_A^{D}t_B^{D}$ is an undirected edge of $T$.
Thus $t_A^{D}t_B^{D}$ lies in the component of $T-E(Q_0 \cup Q_1)$
that contains $t_{k^{D}}'$ and, furthermore, as $\beta(t) \cap D^\pmc \neq \emptyset$,
both $t_{k^{D}}'$ and $t_B^{D}$ lie on the unique path from $t_A^{D}$ to $t$ in $T$.
The claim follows.
\end{claimproof}

With the above claim in hand, we now prove that
\begin{equation}\label{eq:A1there}
A_1 \subseteq D_1.
\end{equation}
By contradiction, assume that $A_1$ intersects a component $D \neq D_1$ of $G-\pmc$.
As $A_1 \subseteq \beta(t_1') \subseteq \bigcup_{t' \in V(T_1)} \beta(t')$, we have $N(D) \subseteq N(D_1)$, $t_D = t_1$, and thus $D \neq D_0$ and $k^{D} = 1$.
By Claim~\ref{cl:case4.2}, $t_1'$ lies on the unique path from $t_B^{D}$ to $t$ in $T$ (possibly $t_1' = t_B^{D}$).
Hence, $A_1\cap D\subseteq \beta(t_1') \cap D \subseteq \beta(t_A^{D}) \cap \beta(t_B^{D}) = N(D) \subseteq \pmc$, a contradiction. This proves~\eqref{eq:A1there}.

Define
\[ C' := S_0 \cup S_1 \cup C_1\quad\mathrm{and}\quad C := C' \setminus A_1. \]
We claim that $C$ is a $(\T,(T,\beta))$-carver for $\pmc$. Clearly, $C \cap \T = \pmc \cap \T$.
(We would like to use $C'$ as the carver, but unfortunately $C'$ may contain vertices of $\T$ in $\beta(t_1') \setminus \pmc$, that is, $A_1$.
 Therefore we need to exclude them manually.)
Let $\widetilde{D}$ be a component of $G-C$; we want to show that there exists $k \in \{0,1\}$ such that $\widetilde{D} \subseteq \beta(t) \cup \bigcup_{t' \in V(T_k)} \beta(t')$.

If $\widetilde{D}$ intersects $D_1$, then by the properties of the carver $S_1$ we have $\widetilde{D} \subseteq N(D_1) \cup D_1$ and we are done with $k=1$.
Otherwise, $\widetilde{D} \cap A_1 = \emptyset$ by~\eqref{eq:A1there}.
Hence, $\widetilde{D}$ is also a component of $G-C'$.

Assume now that $\widetilde{D}$ intersects a component $D$ of $G-\pmc$ such that $D \notin \{D_0,D_1\}$ and $k^D = 1$.
Then, as $\widetilde{D}$ is a component of $G-C'$ and $C_1 \subseteq C'$,
   by the properties of the carver $C_1$ and Claim~\ref{cl:case4.2} we have
   \[ \widetilde{D} \subseteq \beta(t_1') \cup \bigcup_{t' \in V(T^{D})} \beta(t') \subseteq \pmc \cup \bigcup_{t' \in V(T_1)} \beta(t').\]

In the remaining case, for every component $D$ of $G-\pmc$ that intersects $\widetilde{D}$ we have $t_D = t_0$. Hence,
$\widetilde{D} \subseteq \pmc \cup \bigcup_{t' \in V(T_0)} \beta(t')$. This finishes the proof in this case.

\medskip

This completes the case analysis and thus the proof of Proposition~\ref{prop:two-sided}.
\end{proof}

\section{Wrap up}\label{sec:wrapUp}

We are now ready to conclude the construction of a treedepth-$d$ carver family for $P_6$-free graphs.

\begin{theorem}\label{thm:carvers}
For each positive integer $d$, there exists a polynomial-time algorithm that takes in a $P_6$-free graph
$G$ and outputs a family $\mathcal{F} \subseteq 2^{V(G)}$
that is a treedepth-$d$ carver family for $G$.
\end{theorem}
\begin{proof}
Fix any maximal treedepth-$d$ structure $\T$ in $G$, any $\T$-aligned minimal chordal completion $G+F$ of $G$,
and any maximal clique $\pmc$ of $G+F$.
The crucial observation is that any container for $\pmc$
is a $(\T,(T,\beta))$-carver for $\pmc$ regardless of the clique tree $(T,\beta)$ of $G+F$.
Hence, Proposition~\ref{prop:two-sided} gives a family of carvers handling two-sided maximal cliques of $G+F$
for a particular choice of the clique tree, while Theorem~\ref{thm:not-two-sided} and Lemma~\ref{lem:leafyPMC} handle the remaining maximal cliques of $G+F$ regardless of the choice of the clique tree.
\end{proof}
Theorem~\ref{thm:main-algo} follows by a direct combination of Theorem~ \ref{thm:carvers}, Theorem~\ref{thm:dp}, and Theorem~\ref{thm:deg2td}.

\section{Conclusions}
\label{sec:conclusions}

In this paper, we introduced the notion of carvers, a relaxation of the notion of containers, and showed its applicability
by proving that any $(\dege \leq k,\phi)$-\textsc{MWIS} problem is solvable in polynomial time on $P_6$-free graphs.

While in Definition~\ref{def:carvers} and Theorem~\ref{thm:dp} we only require that there exists
a tree decomposition $(T,\beta)$ that is represented in a carver family,
our proof in fact provides a carver family that works for some clique tree of every $\T$-aligned
chordal completion $G+F$, where $\T$ is any maximal treedepth-$d$ structure containing the solution.
(Note that in the context of \textsc{MWIS}, $d=1$ and $\T$ is just the sought solution, since it is a maximal independent set.)
We now present an example showing that if one aims for the ultimate goal of
proving the tractability of $(\dege \leq k,\phi)$-\textsc{MWIS} in $P_t$-free graphs for any fixed $t$,
in particular for $t=7$, one needs to either really use the flexibility of the choice of $(T,\beta)$, or  further adjust the notion of a carver. See Figure~\ref{fig:P7example} for a depiction of the example.

\begin{figure}
\centering
\begin{tikzpicture}[scale = .85, every node/.style={MyNode}]

    \def \r {1.3}
    \def \start {90}

    \foreach \i in {1, ..., 6}{
        \pgfmathparse{int(\i-1)}\edef\last{\pgfmathresult}
        \node (X\last) at ([xshift=-2.5cm]\start+\i*60-60:\r) {};
    }

    \foreach \i in {0, ..., 4}{
        \pgfmathparse{int(\i+1)}\edef\next{\pgfmathresult}
        \draw[thick] (X\i) -- (X\next);
    }
    \draw[thick] (X5) -- (X0);

    \foreach \i in {1, ..., 6}{
        \pgfmathparse{int(\i-1)}\edef\last{\pgfmathresult}
        \node (Y\last) at ([xshift=2.5cm]\start+\i*60-60:\r) {};
    }

    \foreach \i in {0, ..., 4}{
        \pgfmathparse{int(\i+1)}\edef\next{\pgfmathresult}
        \draw[thick] (Y\i) -- (Y\next);
    }
    \draw[thick] (Y5) -- (Y0);

    \node[label=above:{$a$}] (A) at (0,2.25) {};
    \node[label=below:{$b$}] (B) at (0,-2.25) {};

    \draw[thick] (A) -- (X0);
    \draw[thick] (A) -- (X2);
    \draw[thick] (A) -- (X4);
    \draw[thick] (A) -- (Y0);
    \draw[thick] (A) -- (Y2);
    \draw[thick] (A) -- (Y4);

    \draw[thick] (B) -- (X1);
    \draw[thick] (B) -- (X3);
    \draw[thick] (B) -- (X5);
    \draw[thick] (B) -- (Y1);
    \draw[thick] (B) -- (Y3);
    \draw[thick] (B) -- (Y5);

    \node[MyBoldNode] at (X0) {};
    \node[MyBoldNode] at (X3) {};
    \node[MyBoldNode] at (Y0) {};
    \node[MyBoldNode] at (Y3) {};

    \def \height {1}
    \def \width {4}
    \draw (-\width,\height) --++ (2*\width, 0) --++ (0,-2*\height) --++ (-2*\width, 0) -- cycle node [draw=none, fill=none, midway, left] {$S_f$};
\end{tikzpicture}
\caption{The graph $G_2$ with an independent set $I_f$ depicted as large red vertices and the corresponding separator $S_f$ boxed. (The definition of $f$ is not needed due to automorphisms of the graph.)}
\label{fig:P7example}
\end{figure}
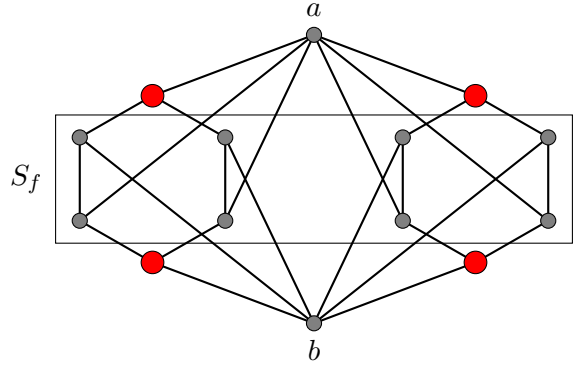

For an integer $n \geq 1$, construct a graph $G_n$ as follows; take $n$ copies of the $6$-vertex cycle, let
the vertices of the $i$-th cycle be $v_{i,0},\ldots,v_{i,5}$, $1 \leq i \leq n$, and add two vertices $a$ and $b$;
$a$ is adjacent to all vertices $v_{i,0},v_{i,2},v_{i,4}$ and $b$ is adjacent to all vertices $v_{i,1},v_{i,3},v_{i,5}$,
$1 \leq i \leq n$.
The graph $G_n$ is $P_7$-free.
For every $f : \{1,\ldots,n\} \to \{0,2,4\}$, the graph $G_n$ contains a maximal independent set
\[I_f = \{v_{i,f(i)}, v_{i,f(i)+3}~|~1 \leq i \leq n\} \]
and a minimal separator
\[ S_f = \{v_{i,f(i)+1},v_{i,f(i)+2},v_{i,f(i)+4},v_{i,f(i)+5}~|~1\leq i \leq n\} \]
with full mesh sides
\begin{align*}
 A_f &= \{a\} \cup \{v_{i,f(i)}~|~1 \leq i \leq n\}, \\
 B_f &= \{b\} \cup \{v_{i,f(i)+3}~|~1 \leq i \leq n\}.
\end{align*}
Here, the addition in the second index is performed modulo $6$.
In this example, if one wants to provide for every $f$
a carver for (an $I_f$-aligned PMC containing) the separator $S_f$ that separates $I_f \cap A_f$ from
$I_f \cap B_f$, one needs an exponential number of carvers.
However, the minimal separator $\{a,b\}$ instead of $S_f$ seems like a much better choice for the algorithm.

\bibliographystyle{amsplain}
\bibliography{p6-free}
\end{document}